\documentclass[11pt,reqno,english]{amsart}
\usepackage{amsfonts,amsmath,latexsym,verbatim,amscd,mathrsfs,color,array}
\usepackage[normalem]{ulem}
\usepackage{enumerate}
\usepackage[colorlinks=true]{hyperref}
\usepackage[hmargin=2.5cm, vmargin=2.5cm]{geometry}
\usepackage{amsmath,amssymb,amsthm,amsfonts,graphicx,color}
\usepackage{amssymb}
\usepackage{pdfsync}
\usepackage{epstopdf}
\usepackage{subcaption}

\usepackage[overload]{empheq}

\usepackage[justification=centering]{caption}

\usepackage{amsmath,amssymb,amsthm,amsfonts,graphicx,color,latexsym,bm,graphics}
\usepackage{amssymb}

\def\theequation{\thesection.\@arabic\c@equation}
\makeatother

\renewcommand{\theequation}{\thesection.\arabic{equation}}
\newtheorem{lemma}{Lemma}[section]

\newtheorem{proposition}{Proposition}[section]

\newtheorem{theorem}{Theorem}[section]

\newcommand{\R}{{\mathbb R}}

\newcommand{\ba}{\begin{array}}
\newcommand{\ea}{\end{array}}
\newcommand{\bed}{\begin{aligned}}
\newcommand{\eed}{\end{aligned}}
\newcommand{\la}{\lambda}

\newcommand{\s}{(-\Delta)^s}
\newcommand{\sh}{(-\Delta)^{\frac12}}

\newcommand{\e}{\varepsilon}

\newcommand{\hu}{\hat{U}}
\newcommand{\hv}{\hat{V}}

\numberwithin{equation}{section}

\title[fractional Gierer-Meinhardt system - supercritical case]{Spike Solutions to the Supercritical Fractional Gierer-Meinhardt System}
\author[D. Gomez]{Daniel Gomez}
\address{\noindent Daniel Gomez, ~Department of Mathematics, University of Pennsylvania, Philadelphia, PA USA 19104-6395}
\email{d1gomez@sas.upenn.edu}

\author[M. Medeiros]{Markus De Medeiros}
\address{\noindent  Markus Medeiros, ~Department of Mathematics, The University of British Columbia, Vancouver, BC Canada V6T 1Z2}
\email{markusdemedeiros@outlook.com	}

\author[J. Wei]{Jun-cheng Wei}
\address{\noindent Jun-cheng Wei, ~Department of Mathematics, The University of British Columbia, Vancouver, BC Canada V6T 1Z2}
\email{jcwei@math.ubc.ca}

\author[W. Yang]{ Wen Yang}
\address{\noindent  Wen Yang, Wuhan Institute of Physics and Mathematics, Innovation Academy for Precision Measurement Science and Technology, Chinese Academy of Sciences, P.O. Box 71010, Wuhan 430071, P. R. China }
\email{math.yangwen@gmail.com}
\begin{document}

\begin{abstract}
Localized solutions are known to arise in a variety of singularly perturbed reaction-diffusion systems. The Gierer-Meinhardt (GM) system is one such example and has been the focus of numerous rigorous and formal studies. A more recent focus has been the study of localized solutions in systems exhibiting anomalous diffusion, particularly with L\'evy flights. In this paper we investigate localized solutions to a one-dimensional fractional GM system for which the inhibitor's fractional order is supercritical. Using the method of matched asymptotic expansions we reduce the construction of multi-spike solutions to solving a nonlinear algebraic system. The linear stability of the resulting multi-spike solutions is then addressed by studying a globally coupled eigenvalue problem. In addition to these formal results we also rigorously establish the existence and stability of ground-state solutions when the inhibitor's fractional order is nearly critical. The fractional Green's function, for which we present a rapidly converging series expansion, is prominently featured throughout both the formal and rigorous analysis in this paper. Moreover, we emphasize that the striking similarities between the one-dimensional supercritical GM system and the classical three-dimensional GM system can be attributed to the leading order singular behaviour of the fractional Green's function. 
\medskip

\noindent {\bf Keywords}:{ Gierer-Meinhardt system, fractional Laplacian, L\'evy flights, localized solutions, singular perturbation.}

\end{abstract}

\maketitle

\section{Introduction}\label{sec:intro}

Reaction diffusion systems have consistently been at the forefront of pattern formation research since Alan Turing's seminal paper in 1952 \cite{turing_1952} in which he demonstrated that sufficiently large differences in the diffusivities of reacting agents can lead to the formation of spatial patterns. By specifying reaction-diffusion systems either phenomenologically or from first principles, studies have used linear stability analysis to explore pattern formation in complex systems with applications to a variety of biological phenomena \cite{murray_2003}. While these studies have traditionally assumed that individual agents undergo Brownian motion in which the mean squared displacement (MSD) is a linear function of the elapsed time, a recent growing body of literature has considered pattern formation in the context of anomalous diffusion in which there is alternative nonlinear relationships between the MSD and the elapsed time \cite{henry_2005,golovin_2008,zhang_2014,khudhair_2021}. Such anomalous diffusion may be better suited for describing the spatial distribution of agents in complex biological environments such as those found within individual cells \cite{bressloff_2014,oliveira_2019}.

Of particular importance, and relevance to the present paper, is the case of anomalous \textit{superdiffusion}  with L\'evy flights for which a heavy-tailed step-length distribution leads to an unbounded MSD. In this case the resulting \textit{fractional} reaction-diffusion system features the fractional Laplacian which for one-dimensional problems is given by
\begin{equation}\label{eq:fractional-laplacian}
	(-\Delta)^s \varphi(x) \equiv C_{s} \int_{-\infty}^{\infty} \frac{\varphi(x)-\varphi(\bar{x})}{|x-\bar{x}|^{1+2s}}d\bar{x},\qquad C_s \equiv \frac{2^{2s}s\Gamma(s+2^{-1})}{\sqrt{\pi}\Gamma(1-s)}.
\end{equation}
where $0<s<1$ and $\Gamma(z)$ is the Gamma function. A growing number of studies have considered such fractional-reaction diffusion systems with different reaction kinetics and using linear stability analysis have demonstrated that the introduction of anomalous diffusion can have a pronounced effect on pattern formation\cite{golovin_2008,khudhair_2021}. Furthermore by considering parameters near the Turing stability threshold the authors in \cite{khudhair_2021} used a weakly nonlinear analysis to investigate the resulting \textit{near equilibrium} patterns that emerge after the onset of linear instabilities. An increasing number of studies have also analysed \textit{far from equilibrium} solutions arising in the singularly perturbed limit where one of the diffusivities is asymptotically small \cite{nec_2012_levi,wei_2019_multi_bump,gomez_2022,medeiros_2022} and it is this latter thread of inquiry which we continue in this paper.

First proposed by Gierer and Meinhardt in 1972 \cite{gierer_1972}, the Gierer-Meinhardt (GM) system is a canonical reaction diffusion system that addresses the consequences of \textit{short range activation} and \textit{long range inhibition} on pattern formation. In this paper we will be interested in the fractional GM system which is given by
\begin{subequations}\label{eq:frac-gm-full-system}
	\begin{align}[left=\empheqlbrace]
		& u_t + \varepsilon^{2s_1} (-\Delta)^{s_1}u + u - v^{-1}u^2 = 0,&  -1<x<1, \label{eq:frac-gm-full-system-u}\\
		& \tau v_t + D(-\Delta)^{s_2}v + v - u^2 = 0,& -1<x<1, \label{eq:frac-gm-full-system-v}
	\end{align}
\end{subequations}
where $u(x,t)$ and $v(x,t)$ correspond to the activator and inhibitor concentrations respectively, 	$\tau>0$, $D>0$, and $0<\varepsilon\ll 1$. In this paper we impose periodic boundary conditions,
\begin{equation}
	u(x+2,t)=u(x,t),\qquad v(x+2,t)=v(x,t),
\end{equation}
which allows us to avoid the technical difficulties in assigning Dirichlet or Neumann boundary conditions \cite{lischke_2020} and with which \eqref{eq:fractional-laplacian} simplifies to
\begin{equation}\label{eq:fractional-laplacian-periodic}
	\begin{split}
		& (-\Delta)^s \varphi(x) \equiv C_s\int_{-1}^{1} K_s(x-\bar{x})\bigl(\varphi(x)-\varphi(\bar{x})\bigr)d\bar{x},\\
		& K_s(z)\equiv \frac{1}{|z|^{1+2s}} + \sum_{j=1}^\infty\biggl(\frac{1}{|z+2j|^{1+2s}} + \frac{1}{|z-2j|^{1+2s}}\biggr).
	\end{split}
\end{equation}
When $s_1=s_2=1$ in \eqref{eq:frac-gm-full-system} we recover the classical one-dimensional GM system which has been the focus of numerous studies. Specifically, in the classical case both rigorous \cite{wei_2014_book} and formal asymptotic methods \cite{iron_2001,ward_2002_asymmetric} have been used to study the existence and stability of \textit{spike} solutions. The analogous system in two- and three-dimensions have also been studied in \cite{wei_2001_gm_2d_weak,wei_2002_gm_2d_strong,gomez_2021}, while the system posed on a two-dimensional Riemannian manifold was considered in \cite{tse_2010}. In both rigorous and formal approaches the activator's asymptotically small diffusivity leads to a separation of spatial scales which yields a tractable reduction of the full system. More recently spike  solutions have been analysed for $s_1\in[1/2,1)$ and $s_2=1$ \cite{nec_2012_levi}, $s_1=s_2\in[1/2,1)$ \cite{wei_2019_multi_bump}, $s_1\in(1/4,1)$ and $s_2\in(1/2,1)$ \cite{gomez_2022}, and for $s_1=s_2=1/2$ \cite{medeiros_2022}. The case of anomalous \textit{subdiffusion}, in which the time derivative is of fractional order, was also considered in \cite{nec_2012_sub}. The individual treatment of each range of $s_1$ and $s_2$ values stems from differences in the singular behaviour of the outer inhibitor solution in a leading order theory. This is closely related to the limiting behaviour of the Green's function $G(x)$ satisfying
\begin{equation*}
	(-\Delta)^{s_2}G + D^{-1} G_D = \delta(x),\qquad 0<x<1,\qquad G_D(x+2)=G_D(x),
\end{equation*}
as $x\rightarrow 0$. In particular as discussed further in \S \ref{subsec:greens-func-properties} and Appendix \ref{app:greens-function} below, the Green's function remains bounded for \textit{subcritical} values of $s_2\in(1/2,1]$ but has a logarithmic singularity at the \textit{critical} value of $s_2=1/2$. There is a suggestive analogy here with the singular behaviour of the classical free space Green's function in one- and two-dimensions which we elaborate in \S \ref{subsec:greens-func-properties} below.

\begin{figure}[t!]
	\centering 
	\begin{subfigure}{0.475\textwidth}
		\includegraphics[width=\linewidth]{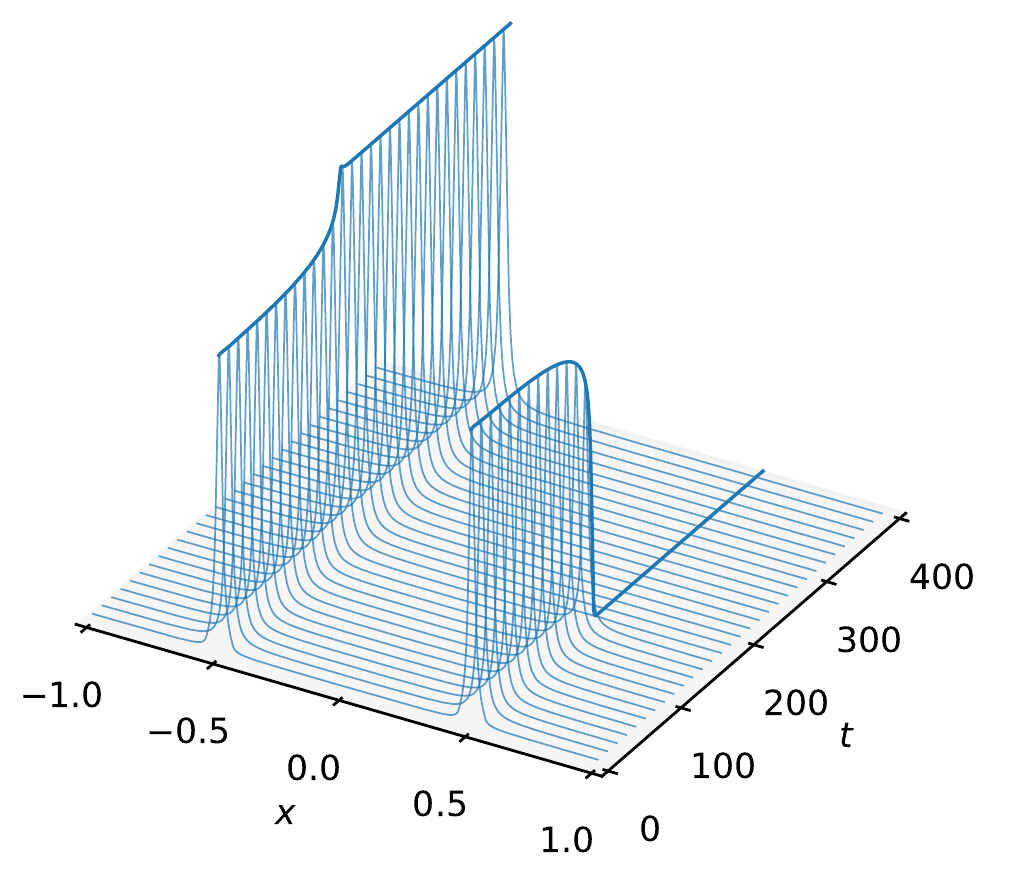}
		\caption{}
		\label{fig:competition-waterfall}
	\end{subfigure}\hfil 
	\begin{subfigure}{0.475\textwidth}
		\includegraphics[width=\linewidth]{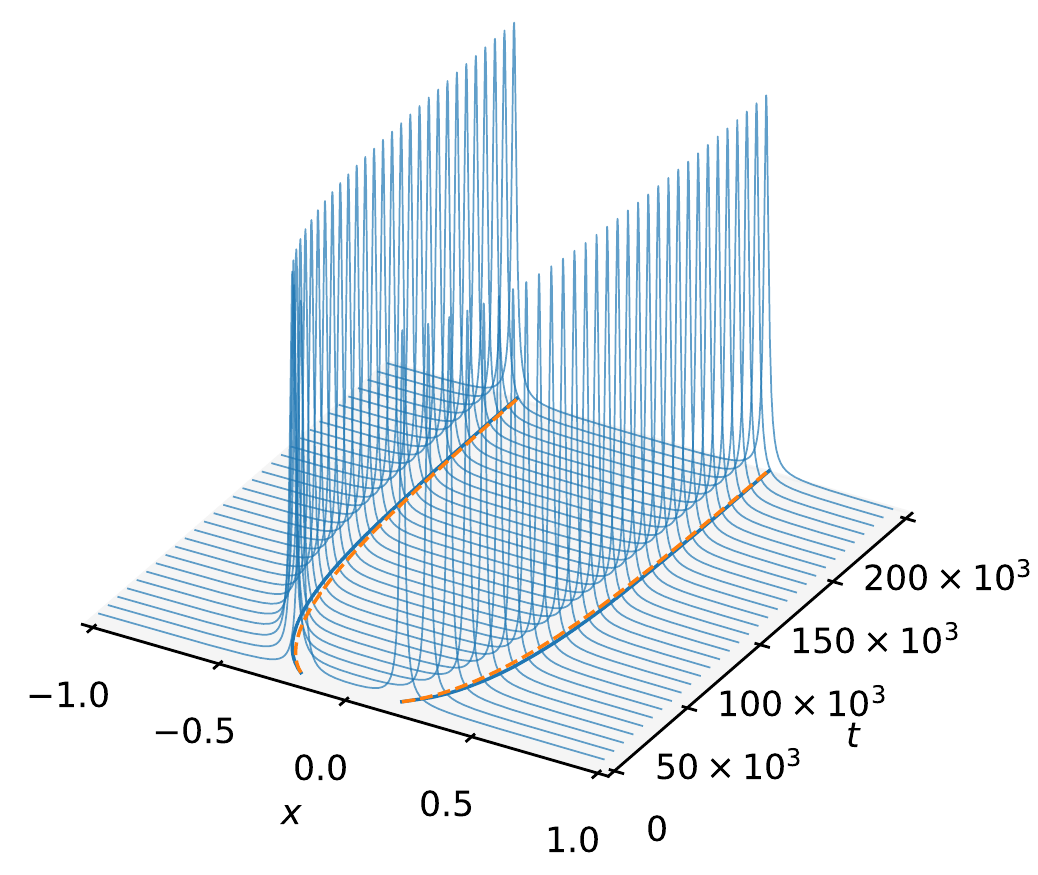}
		\caption{}
		\label{fig:slow-dynamics-waterfall}
	\end{subfigure}
	\caption{(A) Numerically calculated profile of the activator in a symmetric two-spike solution undergoing (A) a competition instability using parameters $s_1=0.5$, $s_2=0.39$, $\varepsilon=0.01$, $\tau=0.1$, and $D=1.095\varepsilon^{2s_2-1}$, and (B) slow dynamics over an $O(\varepsilon^{2s_2-3})$ timescale using parameters $s_1=0.4$, $s_2=0.35$, $\varepsilon=0.01$, $\tau=0.1$, and $D=0.768\varepsilon^{2s_2-1}$. The solid blue and dashed orange curves in the $(x,t)$ plane in (B) indicate the numerically and asymptotically calculated spike locations.}\label{fig:simulation-waterfalls}
\end{figure}

In this paper we use the method of matched asymptotic expansions to study the existence and stability of spike solutions to \eqref{eq:frac-gm-full-system} in the parameter regime
\begin{equation}
	1/4<s_1<1,\qquad 0<s_2<1/2.
\end{equation}
As discussed in \S \ref{subsec:greens-func-properties} and Appendix \ref{app:greens-function}, the Green's function in this regime has an algebraic (and in some cases an additional logarithmic) singularity as $x\rightarrow 0$ and for this reason we refer to the resulting fractional GM system as being \textit{supercritical}. Consequently the asymptotic analysis of spike solutions in this regime is analogous to that found in the classical three dimensional Schnakenberg \cite{tzou_2017_schnakenberg} and Gierer-Meinhard systems \cite{gomez_2021}. Using the  method of matched asymptotic expansions we thus construct multi-spike solutions by deriving an appropriate nonlinear algebraic system (NAS). From the NAS we identify two distinguished parameter regimes: the $D=O(1)$ regime and the $D=O(\varepsilon^{2s_2-1})$ regime. Whereas in the former the NAS admits only \textit{symmetric} solutions, we find that in the latter it admits both symmetric and \textit{asymmetric} solutions. The stability of the resulting multi-spike solutions can then be determined by studying a globally coupled eigenvalue problem (GCEP). From the GCEP we deduce that asymmetric solutions are always linearly unstable while symmetric solutions are susceptible to two types of bifurcations: competition instabilities (see Figure \ref{fig:competition-waterfall} for an example), and Hopf bifurcations. In addition to these bifurcation which occur over an $O(1)$ timescale, otherwise stable multi-spike solutions can also undergo drift motion over an $O(\varepsilon^{2s_2-3})$ timescale (see Figure \ref{fig:slow-dynamics-waterfall} for an example). This paper thus fully characterizes the equilibrium solutions to \eqref{eq:frac-gm-full-system} and their linear stability, while also identifying key parameter regimes for the diffusivity $D$.

The bulk of this paper uses formal asymptotic methods to characterize localized solutions as discussed in the preceding paragraph. While a rigorous justification of these results remains open for the full range of values $0<s_2<1/2$ there are some results that we can rigorously prove when $s_2<1/2$ is close to $s_2=1/2$. Specifically, for such values of $s_2$ we can rigorously prove the existence and stability of \textit{ground state} solutions to the core problem considered in \S\ref{sec:equilibrium}.  Specifically we have the following theorem.
\begin{theorem}
	\label{th1.exist}
	There exists an $\varepsilon_0>0$ such that for each $s\in \left(\frac12(1-\varepsilon_0),\frac12\right)$ the core problem
	\begin{subequations}
	\begin{align}[left=\empheqlbrace]
		& (-\Delta)^\frac12 U+U-V^{-1}U^2=0,\quad (-\Delta)^sV-U^2=0,  & -\infty<x<\infty,  \\
		& U,V>0, & -\infty<x<\infty, \\
		& U,V\to0, & \mbox{as}\quad |x|\to+\infty,
	\end{align}
	\end{subequations}
	admits a solution $(U,V)$ such that
	\begin{equation}
		\lim_{\varepsilon_0\to0}\left|\tau_s^{-1}U(x)-w(x)\right|=0,\qquad  \lim_{\varepsilon_0\rightarrow 0}\left|\tau_s^{-1}V(x)-1\right|=0,
	\end{equation}
	uniformly in compact sets in $x$. Here $U$ is the ground state solution of
	\begin{equation*}
		(-\Delta)^\frac12 w+w-w^2=0,
	\end{equation*}
	and
	\begin{equation*}
		\tau_s=\left(\frac{\Gamma(1-2s)\sin(s\pi)}{\pi}\int_{\mathbb{R}}w^2(x)dx\right)^{-1}.
	\end{equation*}
\end{theorem}
Interestingly, we can also study the stability and instability of the ground state solution constructed in Theorem \ref{th1.exist}. Writing the associated eigenvalue problem for the system as
\begin{subequations}\label{1.eig}
\begin{align}[left=\empheqlbrace]
	& (-\Delta)^\frac12\phi +\phi-2V^{-1}U\phi+V^{-2}U^2\psi+\lambda_s\phi=0, & -\infty<x<\infty,\\
	& (-\Delta)^s\psi-2U\phi+\tau\lambda_s\psi=0, & -\infty<x<\infty,
\end{align}
\end{subequations}
where $\lambda_s\in \mathbb{C}.$ Here we say $(U,V)$ is linearly stable if the real part of each eigenvalue is negative, while $(U,V)$ is called linear unstable if there exists a $\lambda_s$ such that its real part $\Re(\lambda_s)>0.$
\begin{theorem}
	\label{th1.stable}
	Let $(U,V)$ be the solution constructed in Theorem \ref{th1.exist}. There exists $\tau_1$ such that the solution is linearly stable for any $\tau<\tau_1.$
\end{theorem}

The remainder of this paper is organized as follows. In \S \ref{sec:equilibrium} we construct multi-spike quasi-equilibrium solutions by first considering the relevant core problem in \S \ref{subsec:core-problem} and then deriving the NAS in \S \ref{subsec:nas}. This is followed by \S\ref{subsec:sym_asym_regime} where we specifically consider symmetric and asymmetric solutions in the $D=O(\varepsilon^{2s_2-1})$ regime and by \S\ref{subsec:greens-func-properties} where we discuss in more detail the singular behaviour of the Green's function and its connection with higher dimensional problems. In \S\ref{sec:stability} we study the linear stability of multi-spike solutions by deriving the GCEP and focusing in particular on the $D\ll O(\varepsilon^{2s_2-1})$ and $D=O(\varepsilon^{2s_2-1})$ regimes in \S\ref{subsec:stability-reg-1} and \S\ref{subsec:stability-reg-2} respectively.  This is followed by \S\ref{sec:slow-dynamics} where we derive an ordinary differential equation (ODE) system governing the slow dynamics of multi-spike solutions. In \S\ref{sec:simulations} we then perform full numerical simulations of \eqref{eq:frac-gm-full-system} to validate our asymptotic theory. In \S\ref{sec:rigorous} we prove Theorems \ref{th1.exist} and \ref{th1.stable}. Finally in \S\ref{sec:discussion} we summarize our results and make some concluding remarks.

\section{Asymptotic Approximation of $N$-Spike Quasi-Equilibrium Solutions}\label{sec:equilibrium}

In this section we will use the method of matched asymptotic expansions to calculate asymptotic approximations of $N$-spike solutions to
\begin{subequations}\label{eq:frac-gm-equilibrium-system}
	\begin{align}[left=\empheqlbrace]
		&\varepsilon^{2s_1} (-\Delta)^{s_1}u + u - v^{-1}u^2 = 0,& -1<x<1, \label{eq:frac-gm-equilibrium-system-u}\\
		&D(-\Delta)^{s_2}v + v - u^2 = 0,& -1<x<1, \label{eq:frac-gm-equilibrium-system-v}
	\end{align}
	with periodic boundary conditions
	\begin{equation}
		u(x+2)=u(x),\qquad v(x+2) = v(x).
	\end{equation}
\end{subequations}
The successful use of the method of matched asymptotic expansions relies on the asymptotically small activator diffusivity $\varepsilon^{2s_1}\ll 1$ which leads to the emergence of two distinct length scales. Specifically the activator concentrates at $N$ points $-1<x_1<...<x_N<1$ that are well separated in the sense that $|x_i-x_j|\gg \varepsilon$ for all $i\neq j$ as well as $x_1+1\gg \varepsilon$ and $1-x_N\gg\varepsilon$. Over an $O(\varepsilon)$ length scale centred at each $x_1,...,x_N$ the system \eqref{eq:frac-gm-equilibrium-system} is approximated by a core problem in $\mathbb{R}$ whose solutions yields the local profile of the activator and inhibitor. This core problem depends on an undetermined spike strength parameter whose value determines the far-field behaviour of the core solution. On the other hand over an $O(1)$ length scale away from each spike location $x_1,...,x_N$ the nonlinear term appearing in \eqref{eq:frac-gm-equilibrium-system-v} can be approximated, in the sense of , by a sum of appropriately weighted Dirac delta functions centred at each $x_1,..,x_N$. As a consequence the inhibitor can be approximated as a weighted sum of Green's functions over an $O(1)$ length scale. By matching the behaviour of this sum of Green's functions as $x$ approaches each spike location with the far-field behaviour of each core solution we can then derive a NAS of $N$ equations in the $N$ undetermined spike strength parameters. The method of matched asymptotic expansions therefore reduces the original PDE system \eqref{eq:frac-gm-equilibrium-system} to a finite number of nonlinear algebraic equations whose solutions yields an asymptotic approximation of an $N$-spike solution.

Guided by the preceding discussion, in the remainder of this section we will first discuss the core problem and highlight some of its key properties. We will then use the method of matched asymptotic expansions as outlined above to derive the NAS. The remainder of the section will then be dedicated to a discussion on the existence of symmetric and asymmetric $N$-spike solutions as well as to some of the peculiarities of the fractional Gierer-Meinhardt system which distinguish it from the classical one- and three-dimensional Gierer-Meinhardt systems.

\subsection{The Core Problem}\label{subsec:core-problem}

The core problem is one of the key ingredients in deriving an asymptotic approximation of an $N$-spike quasi-equilibrium solutions as outlined above. It is given by
\begin{subequations}\label{eq:core_problem}
	\begin{align}[left=\empheqlbrace]
		& (-\Delta)^{s_1} U_c + U_c - V_c^{-1}U_c^2 = 0,\quad (-\Delta)^{s_2} V_c - U_c^2 = 0, & -\infty <y< \infty,\label{eq:core_problem-pde} \\
		& U_c\sim \nu(S)|y|^{-(1+2s_1)},\quad V_c\sim \mu(S) + S|y|^{2s_2-1}, & \text{as}\quad |y|\rightarrow\infty. \label{eq:core_problem-far-field}
	\end{align}
\end{subequations}
where $S>0$ is a parameter which we refer to as the \textit{spike strength} while $\nu(S)$ and $\mu(S)$ are two $S$-dependent constants. Solutions to \eqref{eq:core_problem} will be denoted by $U_c(y;S)$ and $V_c(y;S)$ to make explicit the dependence on the parameter $S$. The core problem \eqref{eq:core_problem} is a leading order approximation of \eqref{eq:frac-gm-equilibrium-system} after the rescaling $y = \varepsilon^{-1}(x-x_i)$ and its solutions yield the local profile of each spike in an $N$-spike quasi-equilibrium solution of \eqref{eq:frac-gm-equilibrium-system}. The far-field behaviour of $U_c(y;S)$ and $V_c(y;S)$ is a consequence of the following lemma

\begin{lemma}\label{lemma:decay}
	Let $0<s<1/2$ and suppose that $f(y) = O(|y|^{-\sigma})$ as $|y|\rightarrow\infty$ for $\sigma>0$.
	\begin{enumerate}
		\item If $\sigma > 1+2s$ then the solution to $$(-\Delta)^s\phi + \phi = f,\quad\text{for }-\infty<y<\infty;\qquad \phi\rightarrow 0,\quad \text{as }|y|\rightarrow\infty,$$ decays like $\phi \sim C|y|^{-1-2s}$ as $|y|\rightarrow\infty$.
		\item If $\sigma > 1$ then the solution to $$(-\Delta)^s\phi = f,\quad\text{for }-\infty<y<\infty;\qquad \phi\rightarrow 0,\quad \text{as }|y|\rightarrow\infty,$$ decays like $\phi\sim C|y|^{2s-1}$ as $|y|\rightarrow\infty$.
	\end{enumerate}
\end{lemma}
\begin{proof}
	The conclusion follows easily by using classical potential analysis and the decay properties of the Green's functions associated with the operators $(-\Delta)^s+I$ and $(-\Delta)^s$. Specifically the Green's function $G(x,y)$ of $(-\Delta)^s+I$ has the asymptotic behaviour
	\begin{equation}
		\label{2.asy-1}
		\lim_{|x|\to\infty}G(x)|x|^{1+2s}=C,
	\end{equation}
	for some constant $C>0$, while the Green's function $G_0(x,y)$ of $(-\Delta)^s$ has the form
	\begin{equation}
		\label{2.asy-2}
		G_0(x,y)=\frac{1}{\pi}\Gamma(1-2s)\sin(s\pi)|x|^{2s-1}.
	\end{equation}
	We refer the readers to \cite[Section 2]{wei_2019_multi_bump} and \cite[Section 1.12]{pozrikidis2018fractional} for the proof of \eqref{2.asy-1} and \eqref{2.asy-2} respectively.
\end{proof}

We can in fact be more explicit about the solution $V_c$ of \eqref{eq:core_problem} by taking the Fourier transform of the second equation in \eqref{eq:core_problem-pde} to get
\begin{equation}\label{eq:vc-fourier-rep}
	V_c(y;S) = C + \mathfrak{a}_{s_2}\int_{-\infty}^\infty |y-\bar{y}|^{2s_2-1} U_c(\bar{y};S)^2 d\bar{y},\qquad \mathfrak{a}_{s_2} \equiv -2\pi^{-1}s\Gamma(-2s_2)\sin(\pi s_2).
\end{equation}
Taking the limit as $|y|\rightarrow\infty$ then yields as a special case of Lemma \ref{lemma:decay} the limiting behaviour $V_c(y;S) \sim C + \mathfrak{a}_{s_2} |y|^{2s_2-1} \int_{-\infty}^\infty U_c(\bar{y};S)^2d\bar{y}$. Comparing this with the far-field behaviour of the core solution given in \eqref{eq:core_problem-far-field} we deduce the useful identity
\begin{equation}\label{eq:S-def}
	S = \mathfrak{a}_{s_2} \int_{-\infty}^\infty U_c(\bar{y};S)^2d\bar{y}.
\end{equation}
which in particular reinforces our assumption that $S>0$ since $\mathfrak{a}_{s_2}>0$ for $s_2<1/2$.

\begin{figure}[t!]
	\centering 
	\begin{subfigure}{0.25\textwidth}
		\includegraphics[scale=0.675]{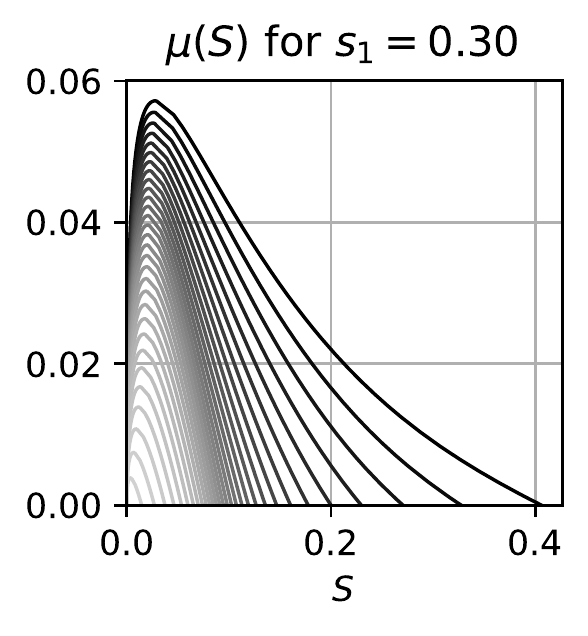}
		\caption{}
		\label{fig:core-far-field-1}
	\end{subfigure}\hfil 
	\begin{subfigure}{0.25\textwidth}
		\includegraphics[scale=0.675]{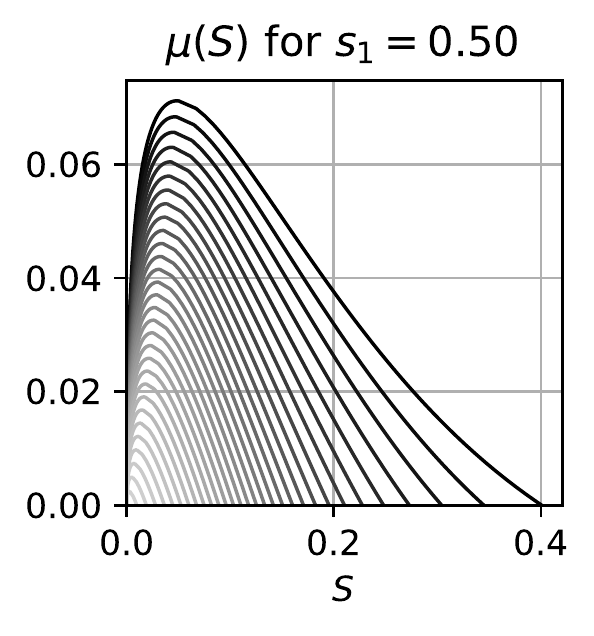}
		\caption{}
		\label{fig:core-far-field-2}
	\end{subfigure}\hfil 
	\begin{subfigure}{0.25\textwidth}
		\includegraphics[scale=0.675]{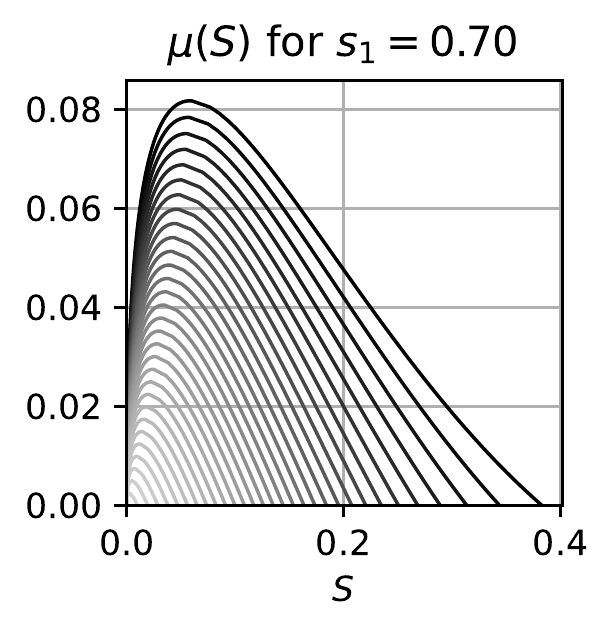}
		\caption{}
		\label{fig:core-far-field-3}
	\end{subfigure}
	\medskip
	\begin{subfigure}{0.25\textwidth}
		\includegraphics[scale=0.675]{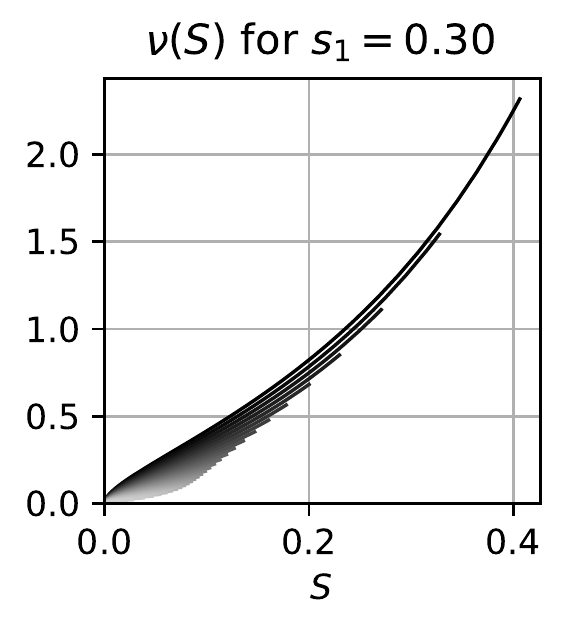}
		\caption{}
		\label{fig:core-far-field-4}
	\end{subfigure}\hfil 
	\begin{subfigure}{0.25\textwidth}
		\includegraphics[scale=0.675]{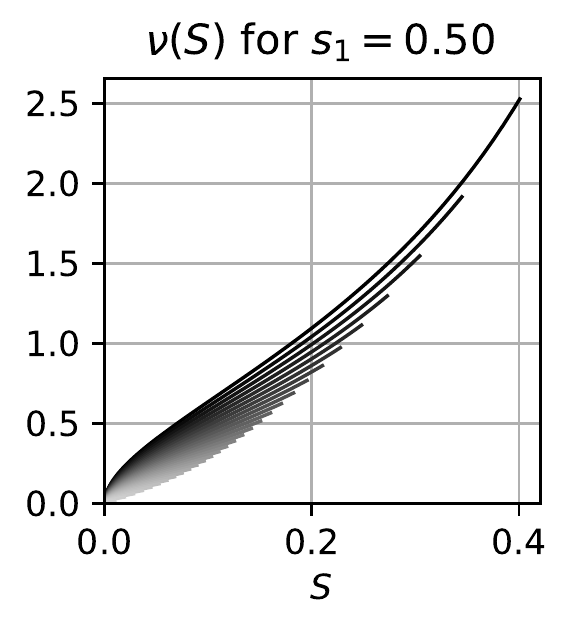}
		\caption{}
		\label{fig:core-far-field-5}
	\end{subfigure}\hfil 
	\begin{subfigure}{0.25\textwidth}
		\includegraphics[scale=0.675]{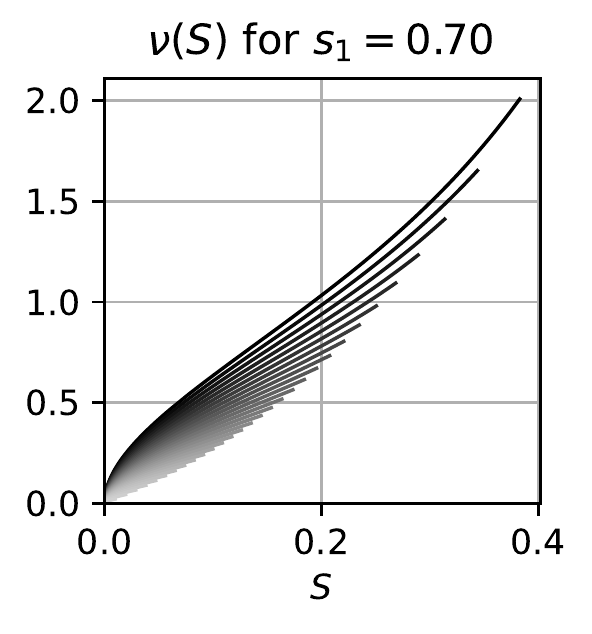}
		\caption{}
		\label{fig:core-far-field-6}
	\end{subfigure}
	\caption{Plots of core problem far-field constants $\mu(S)$ and $\nu(S)$ for distinct values of $1/4<s_1<1$. In each plot the darkest and lightest curves corresponds to $s_2=0.2$ and $s_2=0.49$ respectively, with the intermediate curves corresponding to $0.01$ increments in $s_2$.}
	\label{fig:core-far-field}
\end{figure}

In light of the above discussion the specification of the parameter $S$ is equivalent to fixing the $L^2(\mathbb{R})$ norm of $U_c$. By solving \eqref{eq:core_problem} for a fixed value of $S$ we can then extract the values of the far-field constants $\nu(S)$ and $\mu(S)$ by taking the limits
\begin{equation}\label{eq:nu_mu_def}
	\nu(S) = \lim_{y\rightarrow\infty} |y|^{1+2s_1}U_c(y,S),\qquad \mu(S) = \lim_{y\rightarrow\infty} \bigl(V_c(y;S) - S|y|^{2s_2-1}\bigr).
\end{equation}
The nonlinearity in the first equation of \eqref{eq:core_problem-pde} implies that we must have $V_c(y;S)>0$ for all $y\in\mathbb{R}$ and this leads us to the constraint $\mu(S)\geq 0$.  We next have to determine whether there any values of $S>0$ for which this constraint holds. To address this we first consider the small $S$-asymptotics. Specifically if $S\ll 1$ then \eqref{eq:S-def} implies that $U_c(y;S)=O(\sqrt{S})$ and by balancing terms in \eqref{eq:core_problem} we also deduce that $V_c(y;S)=O(\sqrt{S})$ and $\mu(S)=O(\sqrt{S})$. It is then straightforward to see that to leading order in $S\ll 1$ we have the asymptotic expansions
\begin{subequations}
	\begin{equation}\label{eq:small-S-asymptotics}
		U_c(y)\sim\sqrt{\tfrac{S}{\mathfrak{b}_{s_1}\mathfrak{a}_{s_2}}}w_{s_1}(y),\quad V_c(y)\sim \sqrt{\tfrac{S}{\mathfrak{b}_{s_1}\mathfrak{a}_{s_2}}},\qquad \mu (S)\sim\sqrt{\tfrac{S}{\mathfrak{b}_{s_1}\mathfrak{a}_{s_2}}},
	\end{equation}
	where
	\begin{equation}
		\mathfrak{b}_{s_1} \equiv \int_{-\infty}^\infty w_{s_1}(y)^2 dy,
	\end{equation}
\end{subequations}
and $w_{s_1}(y)$ is the fractional homoclinic solution satisfying
\begin{subequations}\label{eq:fractional-homoclinic}
	\begin{align}[left=\empheqlbrace]
		& (-\Delta)^{s_1}w_{s_1}  + w_{s_1} - w_{s_1}^2 = 0,& -\infty<y<\infty,\\
		& w_{s_1}(y) = O(|y|^{-(1+2s_1)}), &\text{as}\quad |y|\rightarrow\infty.
	\end{align}
\end{subequations}
We refer the reader to Section 4 in \cite{wei_2019_multi_bump} for further properties of the nonlinear problem \eqref{eq:fractional-homoclinic}. The small-$S$ asymptotics \eqref{eq:small-S-asymptotics} imply that $\mu(S)>0$ for $0<S\ll 1$. A numerical continuation in $S$ then further extends the range of $S$ values for which $\mu(S)>0$ holds (see Appendix \ref{subapp:core-problem} for details).

\begin{figure}[t!]
	\centering 
	\begin{subfigure}{0.25\textwidth}
		\includegraphics[scale=0.675]{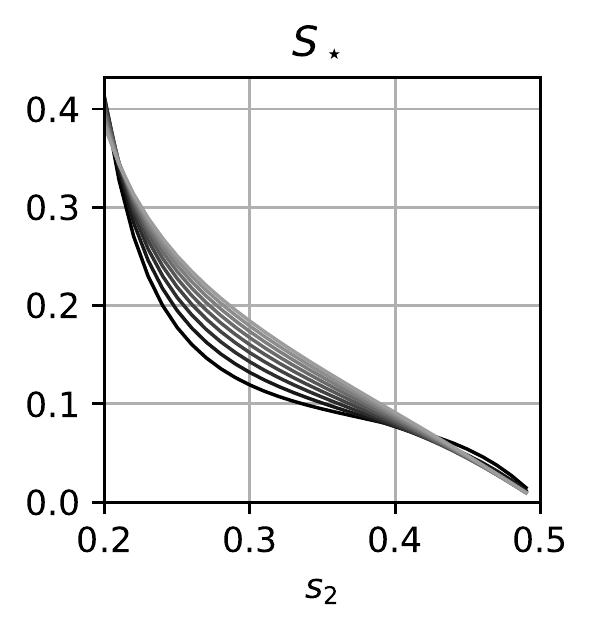}
		\caption{}
		\label{fig:core-critical-S-star}
	\end{subfigure}\hfil 
	\begin{subfigure}{0.25\textwidth}
		\includegraphics[scale=0.675]{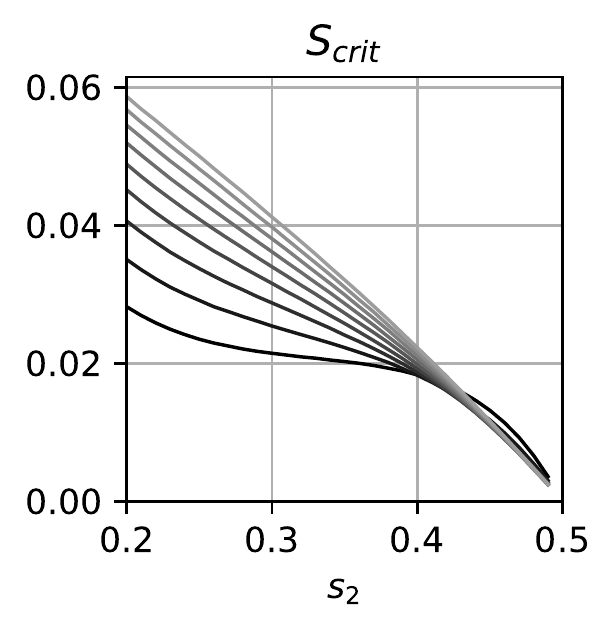}
		\caption{}
		\label{fig:core-critical-S-crit}
	\end{subfigure}\hfil
	\begin{subfigure}{0.25\textwidth}
		\includegraphics[scale=0.675]{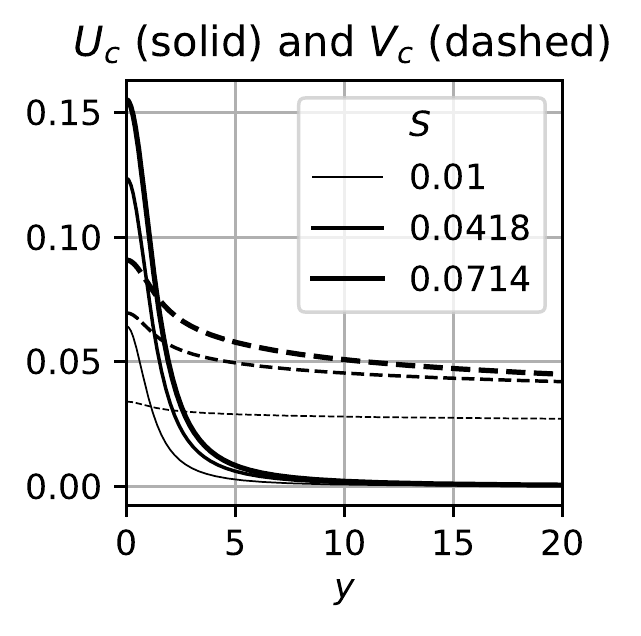}
		\caption{}
		\label{fig:core-sol}
	\end{subfigure}
	\caption{Plots of the critical values (A) $S=S_\star$ and (B) $S=S_\text{crit}$ at which the far-field constant $\mu(S)$ vanishes and attains its global maximum respectively. In both plots the darkest and lightest curves corresponds to values of $s_1=0.3$ and $s_1=0.7$ respectively with the intermediate curves corresponding intermediate values in increments of $0.05$. (C) The core solution for $s_1=0.5$ and $s_2=0.4$ at the indicated values of $S$.}\label{fig:core-critical-S-vals}
\end{figure}

Plots of the numerically calculated far-field constants $\mu(S)$ and $\nu(S)$ are shown in Figure \ref{fig:core-far-field}. These plots indicate that there exists a value of $S=S_\star>0$ beyond which $\mu(S)> 0$ no longer holds. In Figure \ref{fig:core-critical-S-star} we plot $S_\star$ as a function of $s_2$ for select values of $s_1$. In addition the plots in Figure \ref{fig:core-far-field} indicate that $\mu(S)$ attains a unique global maximum in $0<S<S_\star$ at some value $S=S_\text{crit}$ which we plot for distinct values of $s_1$ in Figure \ref{fig:core-critical-S-crit}. This critical value of $S=S_\text{crit}$ plays a crucial role in the leading order stability theory of multiple spike solutions as will be further discussed in \S \ref{sec:stability} below. Finally, in Figure \ref{fig:core-sol} we plot the profiles of the core solutions for select values of $S>0$ when $s_1=0.5$ and $s_2=0.4$.

We conclude by remarking that our preceding discussion has thus far been limited to numerical calculations of solutions to the core problem \eqref{eq:core_problem}. In \S\ref{sec:rigorous} we rigorously prove the existence and stability of \textit{ground state} solutions, i.e.\@ those for which $\mu(S)=0$, when $s_2\approx 1/2$. For more general values of $s_2<1/2$ the rigorous justification of such solutions remains an open problem. We remark also that the existence of a ground state is not guaranteed as can be seen, for example, in the case of the core problems associated with the three-dimensional Gray-Scott, Schnakenberg, and Brusselator systems \cite{gomez_2021}.

\subsection{Asymptotic Matching and the Nonlinear Algebraic System}\label{subsec:nas}

We now consider the asymptotic construction of an $N$ spike solution to \eqref{eq:frac-gm-equilibrium-system}. Assuming that the $N$-spikes concentrate at $N$ well separated (in the sense made precise above) points $-1<x_1<...<x_N<1$ we begin by making the ansatz that
\begin{equation}\label{eq:inner-expansion-0}
	u(x_i+\varepsilon y) \sim D\varepsilon^{-2s_2}\bigl( U_i(y) + o(1)\bigr),\quad v(x_i+\varepsilon y) \sim D\varepsilon^{-2s_2}\bigl( V_i(y;S_i) + o(1)\bigr).
\end{equation}
A simple change of variables then yields that $U_i$ and $V_i$ must satisfy
\begin{equation}\label{eq:inner-eps-dep}
	(-\Delta)^{s_1} U_i + U_i - V_i^{-1}U_i^2 = 0,\quad (-\Delta)^{s_2} V_i + D^{-1}\varepsilon^{2s_2}V_i - U_i^2 = 0 \quad -(1+x_i) <\varepsilon y< 1-x_i.
\end{equation}
Approximating the domain $-\varepsilon^{-1}(1+x_i)<y<\varepsilon^{-1}(1-x_i)$ with $-\infty<y<\infty$ and dropping the $D^{-1}\varepsilon^{2s_2}$ term in the $V_i$ equation we deduce that
\begin{equation}\label{eq:leading-order-inner}
	U_i(y)\sim U_c(y;S_i) + o(1),\qquad V_i(y)\sim V_c(y;S_i) + o(1),
\end{equation}
where $U_c(y;S)$ and $V_c(y;S)$ are the solutions of the core problem \eqref{eq:core_problem} discussed and $S_i>0$ is an as-of-yet undetermined constant. Implicit in the asymptotic approximation \eqref{eq:leading-order-inner} is the assumption that the inner profiles interact only through the far-field behaviour constants $S_i$, the nature of which is revealed by formulating the outer problem and deriving an appropriate matching condition.

Next we derive an outer problem valid for values of $-1<x<1$ such that $|x-x_i|\gg\varepsilon$ for all $i=1,...,N$. We first make note of the limit
\begin{equation*}
	u^2 \rightarrow \varepsilon^{1-4s_2} D^2\sum_{i=1}^N \int_{-\infty}^\infty U_c(y;S_i)^2dy\,\delta(x-x_i) =  \varepsilon^{1-4s_2} D^2\mathfrak{a}_{s_2}^{-1} \sum_{i=1}^{N} S_i\delta(x-x_i),
\end{equation*}
as $\varepsilon\rightarrow 0^+$ which is to be understood in the sense of distributions and for which we have used \eqref{eq:S-def} in the equality. The outer inhibitor solution must then be a $2$-periodic function satisfying
\begin{subequations}
	\begin{equation}\label{eq:v-outer-pde}
		(-\Delta)^{s_2}v + D^{-1}v = \varepsilon^{1-4s_2} \mathfrak{a}_{s_2}^{-1} D\sum_{i=1}^{N} S_i\delta(x-x_i), \qquad x\in(-1,1)\setminus\{x_1,...,x_N\},
	\end{equation}
	and having the limiting behaviour
	\begin{equation}\label{eq:v-outer-limit}
		v\sim D\varepsilon^{-2 s_2}\bigl(\varepsilon^{1-2s_2}S_i|x-x_i|^{2s_2-1} + \mu(S_i)\bigr), \qquad \text{as}\quad x\rightarrow x_i,
	\end{equation}
	for each $i=1,...,N$ obtained from the far-field behaviour of the inner solution \eqref{eq:core_problem-far-field}.
\end{subequations}
We let $G_D(x)$ be the the $2$-periodic fractional Green's function satisfying
\begin{equation}\label{eq:greens-equation}
	(-\Delta)^{s_2}G_D + D^{-1}G_D = \delta(x), \quad -1<x<1,\qquad G_D(x+2)=G_D(x),
\end{equation}
which can be written as (see Appendix \ref{app:greens-function})
\begin{equation}\label{eq:greens-series-rapid}
	G_D(x) = D\sum_{k=1}^{k_{\max}}\frac{(-1)^{k+1}\mathfrak{a}_{ks_2}}{D^{k}}|x|^{2ks_2-1} + R_D(x),\qquad k_{\max} \equiv \lceil \tfrac{1}{2s_2}-1\rceil,
\end{equation}
where $\mathfrak{a}_{ks_2} \equiv -2ks_2\pi^{-1}\Gamma(-2ks_2)\sin(\pi ks_2)$ and $R_D(x)$ is given explicitly by \eqref{eq:greens-series-rapid-R_D}. In terms of this Green's function the solution to \eqref{eq:v-outer-pde} can be explicitly written as
\begin{equation}\label{eq:v-outer}
	v(x) = \varepsilon^{1-4s_2}\mathfrak{a}_{s_2}^{-1}D \sum_{i=1}^N S_i G_D(x-x_i).
\end{equation}
Comparing the limiting behaviour of \eqref{eq:v-outer} as $x\rightarrow x_i$ with the limiting behaviour \eqref{eq:v-outer-limit} from the inner solution yields the algebraic equation
\begin{equation}\label{eq:matching-condition}
	\begin{split}
		\mu(S_i)+\varepsilon^{1-2s_2}S_i|x-x_i|^{2s_2-1} \sim \frac{\varepsilon^{1-2s_2}}{\mathfrak{a}_{s_2}}\biggl( D S_i\sum_{k=1}^{k_{\max}}\frac{(-1)^{k+1}\mathfrak{a}_{ks_2}}{D^k}|x-x_i|^{2ks_2-1} + S_i R_D(0)& \\
		+ \sum_{j\neq i}S_j G_D(|x_i-x_j|) + O(|x-x_i|)   & \biggr).
	\end{split}
\end{equation}
The $S_i|x-x_i|^{2s_2-1}$ term on the left-hand-side cancels the $k=1$ term in the sum on the right-hand-side while the remaining singular terms corresponding to $k=2,...,k_{\max}$ are cancelled out by higher order corrections to the inner solution. On the other hand the constant term $\mu(S)$ on the left-hand-side must be balanced with the constant terms appearing on the right-hand-side. Since this must hold for each value of $i=1,...,N$ we are thus led to the NAS
\begin{equation}\label{eq:NAS}
	\mu(S_i) = \frac{\varepsilon^{1-2s_2}}{\mathfrak{a}_{s_2}}\biggl( S_i R_D(0) + \sum_{j\neq i}S_j G_D(|x_i-x_j|)\biggr),\qquad i=1,...,N.
\end{equation}
Note that this NAS must in general be solved numerically since $\mu(S)$ can only be computed numerically (see \S \ref{subsec:core-problem}). It nevertheless provides a substantial reduction in the construction of multi-spike solutions to the equilibrium equation \eqref{eq:frac-gm-equilibrium-system}.

We remark that the NAS \eqref{eq:NAS} is $\varepsilon$-dependent and yields distinct leading order approximations depending on whether $D=O(1)$ or $D=D_0\varepsilon^{2s_2-1}$ where $D_0=O(1)$. Indeed in the former case \eqref{eq:NAS} implies that $\mu(S_i)=0$ to leading order and hence $S_i\sim S_\star + O(D\varepsilon^{1-2s_2})$. On the other hand if $D=D_0\varepsilon^{2s_2-1}\gg 1$ then the asymptotics
\begin{equation}\label{eq:greens-large-D}
	R_D(0)\sim \frac{1}{2}D + O(1),\qquad G_D(|x_i-x_j|)\sim \frac{1}{2}D + O(1)\quad\text{for }i\neq j \qquad (D\gg 1),
\end{equation}
imply that $S_1,..,S_N>0$ must solve the leading order system
\begin{equation}\label{eq:NAS-leading-order-distinguished-regime}
	\mu(S_i) = \frac{\kappa}{N}\sum_{j=1}^{N} S_j,\qquad \kappa\equiv \frac{ND_0}{2 \mathfrak{a}_{s_2}},
\end{equation}
for each $i=1,...,N$ with the next order correction being $O(D_0^{-1}\varepsilon^{1-2s_2})$. The shape of $\mu(S)$ illustrated in Figure \ref{fig:core-far-field} suggests the possibility that $S_1,...,S_N\in\{S_l,S_r\}$ for some $0<S_l<S_r<S_\star$. Thus whereas the $D=O(1)$ regime supports solutions in which the profiles of each spike are identical, i.e. the $N$-spike solution is \textit{symmetric}, the $D=D_0\varepsilon^{2s_2-1}$ regime may admit both symmetric and \textit{asymmetric} $N$-spike solutions which we discuss further in \S \ref{subsec:sym_asym_regime} below.

While the leading order approximations discussed above are suggestive of the solutions we may encounter it is important to highlight that their associated errors are $O(D\varepsilon^{1-2s_2})$ when $D=O(1)$, and $O(D_0^{-1}\varepsilon^{1-2s_2})$ when $D=D_0\varepsilon^{2s_2-1}$. Although these errors are small in the limit $\varepsilon\rightarrow0$ they may in practice be unacceptably large. For example if $\varepsilon=0.01$ and $s_2 = 0.4$ then $\varepsilon^{1-2s_2}\approx 0.4$. In contrast if we solve the $\varepsilon$-dependent NAS \eqref{eq:NAS} directly then the next order correction to the inner problem can be deduced from the matching condition \eqref{eq:matching-condition} and is either $O(D^{-1}\varepsilon^{2s_2})$ if $s_2<1/4$ or $O(D\varepsilon^{2-2s_2})$ if $1/4<s_2<1/2$. In particular this yields an $O(\varepsilon)$ error when $D=D_0\varepsilon^{2s_2-1}$ and For this reason we will be using the $D=D_0\varepsilon^{2s_2-1}$ regime when we perform numerical simulations of \eqref{eq:frac-gm-full-system} in \S \ref{sec:simulations} below.

\subsection{Symmetric and Asymmetric Solutions in the $D=D_0\varepsilon^{2s_2-1}$ Regime}\label{subsec:sym_asym_regime}

\begin{figure}[t!]
	\centering 
	\begin{subfigure}{0.25\textwidth}
		\includegraphics[scale=0.675]{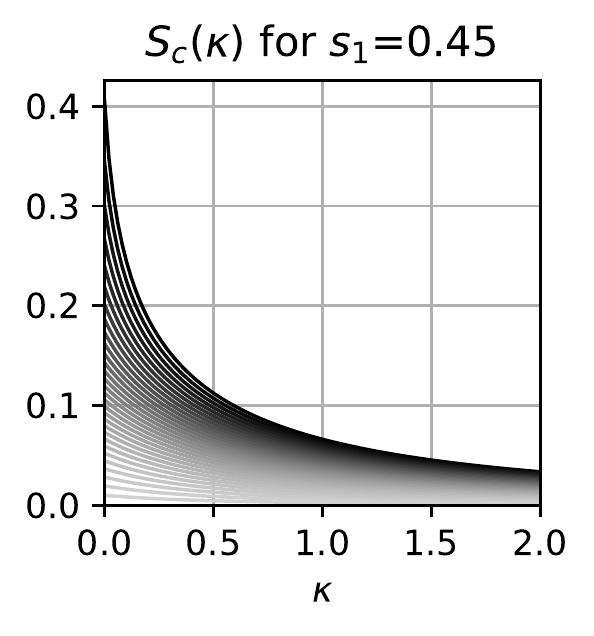}
		\caption{}
		\label{fig:Sc_plots}
	\end{subfigure}\hfil 
	\begin{subfigure}{0.25\textwidth}
		\includegraphics[scale=0.675]{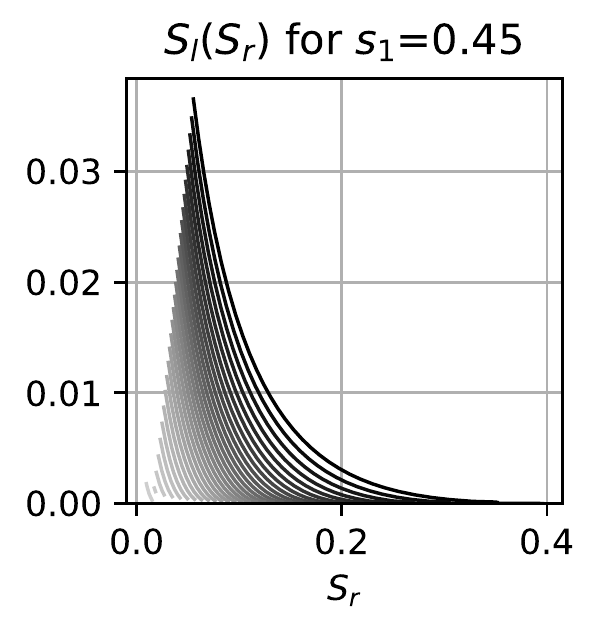}
		\caption{}
		\label{fig:Sl_plots}
	\end{subfigure}\hfil 
	\begin{subfigure}{0.25\textwidth}
		\includegraphics[scale=0.675]{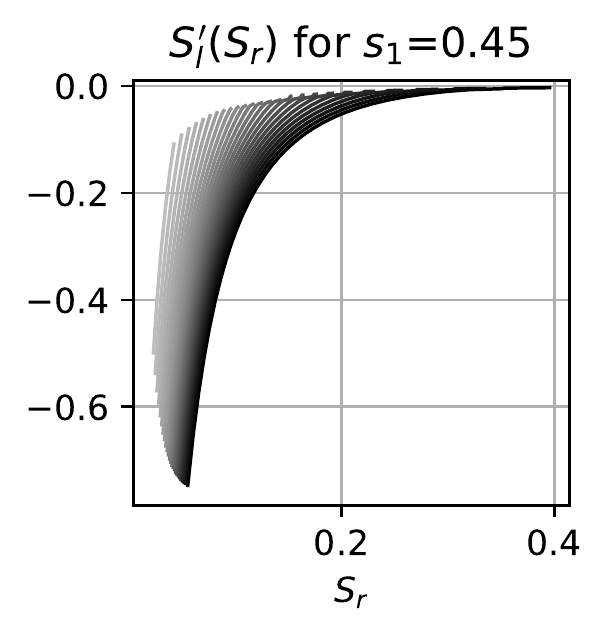}
		\caption{}
		\label{fig:Sl_prime_plots}
	\end{subfigure}
	\caption{Plots of (a) the common spike strength in a symmetric solution, (b) the small spike strength value corresponding to the large spike strength value in an asymmetric solution, and (c) its derivative. In each plot the darkest and lightest curves corresponds to $s_2=0.2$ and $s_2=0.49$ respectively, with the intermediate curves corresponding to $0.01$ increments in $s_2$}
	\label{fig:leading-order-S}
\end{figure}

As discussed above, the leading order equation of the NAS \eqref{eq:NAS} in the $D=D_0\varepsilon^{2s_2-1}$ regime given by \eqref{eq:NAS-leading-order-distinguished-regime} admits both symmetric and asymmetric $N$-spike solutions. In the following section we will explore these two types of solutions in more detail while also drawing parallels to the analogous solutions encountered in the case of the three-dimensional Gierer-Meinhardt system \cite{gomez_2021}.

Symmetric $N$-spike solutions are perhaps the easiest to analyze since in this case the spike strengths are all equal, $S_1=...=S_N=S_c$ and the leading order NAS \eqref{eq:NAS-leading-order-distinguished-regime} reduces to the scalar equation
\begin{equation}\label{eq:NAS-leading-order-distinguished-regime-symmetric}
	\mu(S_c) = \kappa S_c,\qquad 0<S_c<S_\star.
\end{equation}
From the plots of $\mu(S)$ in Figure \ref{fig:core-far-field} it is clear that $S_c\rightarrow S_\star$ as $D_0\rightarrow 0$ thereby providing a connection between the $D=O(1)$ and $D=O(\varepsilon^{2s_2-1})$ regimes. On the other hand as $D_0\rightarrow \infty$ we obtain $S\rightarrow 0$ and in particular using the small-$S$ asymptotics \eqref{eq:small-S-asymptotics} we find that $S\sim (\mathfrak{b}_{s_1}\mathfrak{a}_{s_2}\kappa^2)^{-1}$. In Figure \ref{fig:Sc_plots} we plot $S_c$ versus $\kappa$ for a selection of $s_1$ and $s_2$ values.

In addition to symmetric $N$-spike solutions the leading order NAS \eqref{eq:NAS-leading-order-distinguished-regime} and plots of $\mu(S)$ in Figure \ref{fig:core-far-field} further suggest the possibility of asymmetric $N$-spike solution. Specifically, recalling that $0<S_\text{crit}<S_\star$ is the value where $\mu(S)$ attains its unique maximum we deduce that for any $S_r\in [S_\text{crit},S_\star)$ there is a unique $S_l(S_r)\in(0,S_\text{crit}]$ which we plot for` $s_1=0.45$ and a selection of $s_2$ values in Figure \ref{fig:Sl_plots}. Notice from its definition that $S_l(S_\text{crit})=S_\text{crit}$ whereas $S_l(S_\star)=0$. Moreover, by differentiating $\mu(S_l(S(r)) = \mu(S_r)$ we obtain $S_l'(S_r) = [\mu'(S_l(S_r))]^{-1}\mu'(S_r)$ so that in particular $S_l'(S_r)\rightarrow 0$ as $S_r\rightarrow S_\star$ due to the small $S$ asymptotics \eqref{eq:small-S-asymptotics}. Plots of $S_l'(S_r)$ in Figure \ref{fig:Sl_prime_plots} further indicate that $-1\leq S_l'(S_r)\leq 0$.

We next consider the construction of asymmetric $N$-spike solutions consisting of $1\leq n\leq N-1$ \textit{large} and $N-n$ \textit{small} spikes by letting
$$
S_{\sigma(1)}=...=S_{\sigma(n)}=S_r,\qquad S_{\sigma(n+1)}=...=S_{\sigma(N)}=S_l(S_r),\qquad S_\text{crit}<S_r<S_\star,
$$
where $\sigma$ is a permutation of $\{1,...,N\}$. With this assumption the leading order system \eqref{eq:NAS-leading-order-distinguished-regime} reduces to the scalar equation
\begin{equation}\label{eq:NAS-leading-order-asy}
	\mu(S_r) = \kappa f(S_r,\tfrac{n}{N}), \qquad f(S,\theta)\equiv \theta S + (1-\theta) S_l(S).
\end{equation}
This scalar equation was previously encountered in the classical 3D Gierer-Meinhardt model \cite{gomez_2021}. For that model two key properties of $\mu(S)$ and $S_l(S_r)$ allowed for a complete characterization of the bifurcation structure of \eqref{eq:NAS-leading-order-asy}, the first being that $\mu'(S)<0$ for $S_\text{crit}<S<S_\star$, and the second that $-1 < S_l'(S_r) < 0$ for all $S_\text{crit}<S_r<S_\star$. Since these properties likewise hold for the $\mu(S)$ and $S_l(S_r)$ in our present case we will simply state the results from \cite{gomez_2021}, referring the interested reader to Section 2.3 of \cite{gomez_2021} for more details. The first result states that if
\begin{equation}\label{eq:asym-existence-1}
	0<\kappa<\kappa_{c1} \equiv \mu(S_{crit})/S_\text{crit},
\end{equation}
then \eqref{eq:NAS-leading-order-asy} has a unique solution for any $1\leq n\leq N-1$. In addition if $n\geq N-n$ then \eqref{eq:NAS-leading-order-asy} does not have a solution for any $\kappa\geq \kappa_{c1}$. If on the other hand $n<N-n$ then \eqref{eq:NAS-leading-order-asy} has exactly two distinct solutions for
\begin{equation}
	\kappa_{c1} < \kappa < \kappa_{c2}\equiv \mu(S_r^\star) / f(S_r^\star,n/N),
\end{equation}
where $S_\text{crit}<S_r^\star<S_\star$ is the unique solution to
\begin{equation}
	f(S_r^\star,n/N)\mu'(S_r^\star) = f'(S_r^\star,n/N)\mu(S_r^\star),
\end{equation}
and no solutions if $\kappa\geq \kappa_{c2}$.

\subsection{On the Fractional Green's Function}\label{subsec:greens-func-properties}

The preceding sections have highlighted the importance of the fractional Green's function satisfying \eqref{eq:greens-equation} in the asymptotic construction of quasi equilibrium solutions. We conclude this section by highlighting some of the key properties of the fractional Green's function and relating them to the behaviour of the classical Green's function in one-, two-, and three-dimensions.

The limiting behaviour of $G_D(x)$ as $x\rightarrow 0$ plays a crucial role in the existence and stability of multi-spike solutions. Interestingly this behaviour is markedly different when $s_2\in(1/2,1]$, $s_2\in(0,1/2)\setminus\{\tfrac{1}{2r}\,|\,r\in\mathbb{Z},\,r\geq 1\}$, and $s_2\in\{\tfrac{1}{2r}\,|\,r\in\mathbb{Z},\,r\geq 1\}$. In particular when $s_2\in(1/2,1)$ the Green's function is not singular with $G_D(x)\sim \mathfrak{a}_{s_2}|x|^{2s_2-1} + O(1)$ as $x\rightarrow 0$ \cite{gomez_2022}. On the other hand, referring to Propositions \ref{prop:greens} and \ref{prop:greens-log} in Appendix \ref{app:greens-function}, we have
{\small \begin{equation}
		G_D(x)\sim\begin{cases}  \sum_{k=1}^{k_{\max}}\tfrac{(-1)^{k-1}\mathfrak{a}_{ks_2}}{D^{k-1}}|x|^{2ks_2-1} +O(1),& s_2\in(0,1/2)\setminus\{\tfrac{1}{2r}\,|\,r\in\mathbb{Z},\,r\geq 1\},\\ \sum_{k=1}^{r-1}\frac{(-1)^{k-1}\mathfrak{a}_{ks_2}}{D^{k-1}}|x|^{2ks_2-1} + \frac{(-1)^r}{\pi D^{r-1}}\log|x| + O(1),& s_2 = \tfrac{1}{2r},\,\text{for } r\in\mathbb{Z},\,r\geq 1, \end{cases}
\end{equation}}
where $k_{\max} = \lceil\tfrac{1}{2s_2}-1\rceil$. The singular behaviour in each of these cases has direct analogies with the singular behaviour of the non-fractional Green's function in one-, two, and three-dimensions. Specifically, we may view the fractional Green's function as analogous to the one-, two-, and three-dimensional non-fractional Green's function when $s_2\in(1/2,1)$, $s_2=1/2$, and $s_2\in(0,1/2)\setminus\{\tfrac{1}{2r}\,|\,r\in\mathbb{Z},\, r\geq 1\}$ respectively. This analogy further extends to the methods used in the analysis of spike solutions as is evident by the similarities between the analysis in \cite{gomez_2022} for $s_2\in(1/2,1)$ and the classical Gierer-Meinhardt system (e.g. in \cite{iron_2001}), that in \cite{medeiros_2022} for $s_2=1/2$ and the two-dimensional Gierer-Meinhardt system \cite{wei_2001_gm_2d_weak,wei_2002_gm_2d_strong}, and that in the present paper with the analysis of spike solutions in the three-dimensional Schnakenberg \cite{tzou_2017_schnakenberg} and Gierer-Meinhardt \cite{gomez_2021} systems. For the remaining values of $s_2\in\{\tfrac{1}{2r}\,|\,r\in\mathbb{Z},\, r\geq 2\}$ the mixing between logarithmic and algebraic singularities leads to problems which don't appear to have a clear classical analog. The analysis of the fractional Gierer-Meinhardt system for these remaining parameter values is not addressed in this paper but is an interesting direction for future research.

\begin{figure}[t]
	\centering
	\includegraphics[scale=0.7]{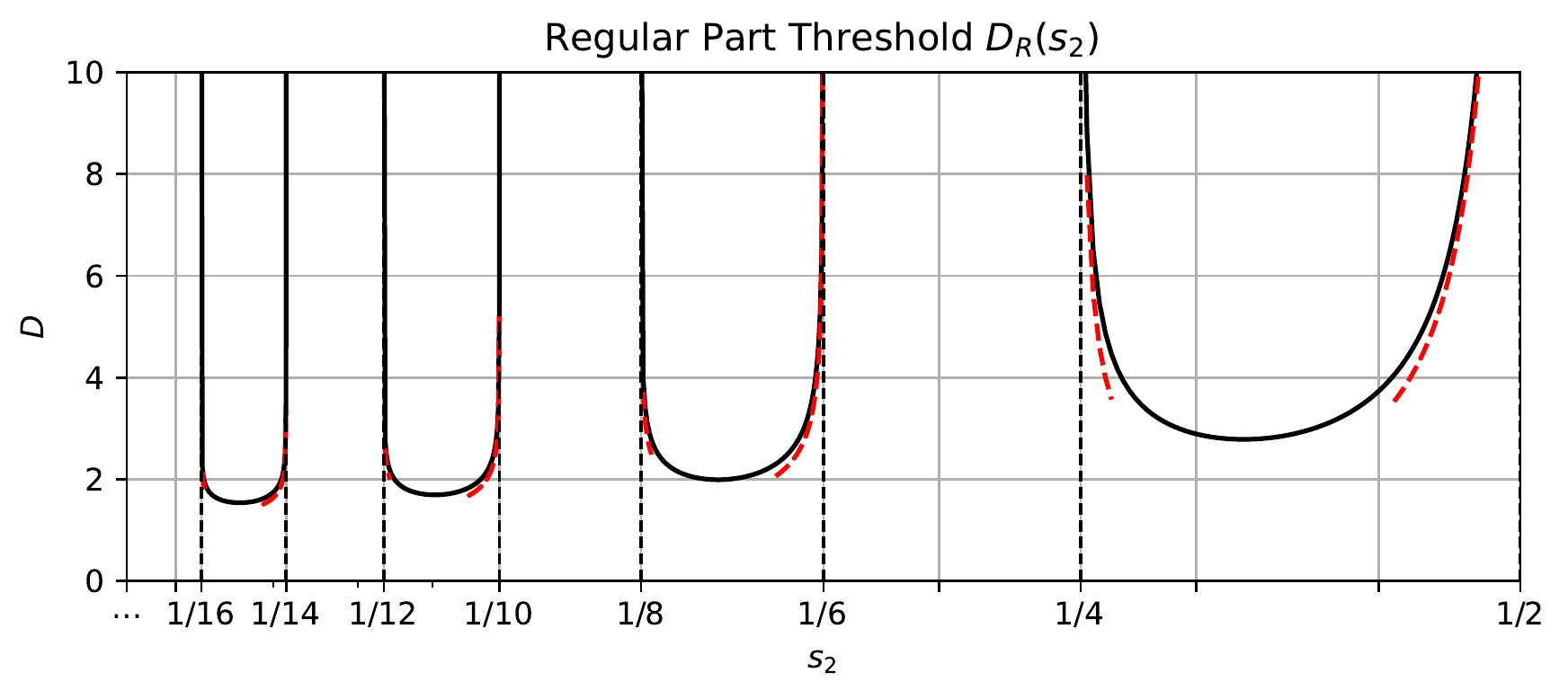}
	\caption{Plots of the threshold $D_R(s_2)$ (solid black) and its limiting behaviour as $s_2\rightarrow \bigl(\tfrac{1}{2r}\bigr)^-$ for odd values of $r\geq 1$ and as $s_2\rightarrow \bigl(\tfrac{1}{2r}\bigr)^+$ for even values of $r\geq 1$ (dashed red).}\label{fig:D_R}
\end{figure}

We now consider the regular part of the Green's function $R_D(x)$ which can be computed using the series expansion \eqref{eq:greens-series-rapid-R_D}. Numerical calculations indicate that $R_D(0)>0$ for all $D>0$ when $s_2\in(\tfrac{1}{2(r+1)},\tfrac{1}{2r})$ for even values of $r\geq 1$, whereas there is a threshold $D_R(s_2)>0$ for which $R_D(0)<0$ for all $D<D_R(s_2)$ when $s_2\in(\tfrac{1}{2(r+1)},\tfrac{1}{2r})$ for odd values of $r\geq 1$. The threshold $D_R(s_2)$ can be numerically computed using \eqref{eq:greens-series-rapid-R_D} and is plotted in Figure \ref{fig:D_R}. Note that special care must be taken when using the series \eqref{eq:greens-series-rapid-R_D} as $s_2\rightarrow \tfrac{1}{2r}^+$ for any integer $r\geq 2$. Specifically, in this case $k_{\max}=\lceil\tfrac{1}{2s_2}-1\rceil < r$ so that in the limit the second term in \eqref{eq:greens-series-rapid-R_D} does not converge. This is easily fixed by letting $k_{\max}=r$ when $s_2$ is sufficiently close to $\tfrac{1}{2r}$ and we use this in our numerical computations. In addition our numerical calculations indicate that $D_R(s_2)\rightarrow+\infty$ as $s_2\rightarrow\bigl(\tfrac{1}{2(r+1)}\bigr)^+$ and $s_2\rightarrow \bigl(\tfrac{1}{2r}\bigr)^-$ for odd values of $k\geq 1$. This diverging behaviour can be explicitly characterized by balancing dominant terms in the series \eqref{eq:greens-series-rapid-R_D}. Specifically by noting that $\Gamma(-z)\sim -(1-z)^{-1}$ as $z\rightarrow 1$ we deduce that $\mathfrak{a}_{rs_2}\sim \pi^{-1}(1-2rs_2)^{-1}$ as $s_r\rightarrow \tfrac{1}{2r}$. Assuming that $D=D_R(s_2)\gg 1$ in \eqref{eq:greens-series-rapid-R_D} and balancing dominant terms then implies that
\begin{equation}
	D_R(s_2)\sim \begin{cases} \bigl(\frac{2}{\pi(1-2rs_2)}\bigr)^{1/r}, & \text{as}\quad s_2\rightarrow \bigl(\tfrac{1}{2r}\bigr)^{-}\quad\text{ for odd $r$},\\ \bigl(\frac{2}{\pi(2rs_2-1)}\bigr)^{1/r}, & \text{as}\quad s_2\rightarrow \bigl(\tfrac{1}{2r}\bigr)^{+}\quad \text{for even $r$}.  \end{cases}
\end{equation}
which we plot using dashed lines in Figure \ref{fig:D_R}. This diverging behaviour of the threshold $D_R(s_2)$ is suggestive of an alternative scaling that arises in this limit. Indeed the appearance of a logarithmic singularity at values of $s_2=\tfrac{1}{2r}$ for integer $r\geq 1$ suggests that an additional small parameter $\nu=-\tfrac{1}{\log \varepsilon}$ must be incorporated into the asymptotic theory, likely leading to alternative distinguished asymptotic regimes for the diffusivity.

We conclude by remarking that the negativity of $R_D(0)$ for $D<D_R(0)$ and $s_2\in(\tfrac{1}{2(r+1)},\tfrac{1}{2r})$ for even values of $r\geq 1$ poses a challenge to the application of our asymptotic theory. Indeed, recalling the NAS \eqref{eq:NAS} we observe that $R_D(0)<0$ in this regime contradicts the positivity of $\mu(S)$. By restricting our attention to the case $D=D_0\varepsilon^{2s_2-1}$ this difficulty can be circumvented at least in theory when $\varepsilon\ll 1$ and for which $D>D_R(s_2)$. However to validate our asymptotic theory with numerical simulations we have to use a finite value of $\varepsilon>0$ which may lead to $D<D_R(s_2)$ especially as $s_2$ approaches any of the values for which $D_R(s_2)$ diverges. Such behaviour is not in the range of validity of our asymptotic theory and we will henceforth ignore it though we would be remiss to not at least suggest approaches for handling this issue. One possibility is to develop a higher order asymptotic theory though this falls out of the scope of this paper. An alternative approach is to consider a $\varepsilon$-dependent core problem \eqref{eq:core_problem} posed on the truncated domain $|y|<L/\varepsilon$ for some $L>0$ in which case negative values of $\mu(S)$ are permissible provided that the solution $V_c$ remains positive. This approach however has two major shortcomings: it requires an appropriate assignment for $V_c$ in $|y|\geq L/\varepsilon$ in order to have a well-posed problem, and the core problem will need to be recomputed anew for different values of $\varepsilon$.

\section{Linear Stability: The Large, $O(1)$, Eigenvalues}\label{sec:stability}

In this section we consider the linear stability on an $O(1)$ timescale of the $N$-spike equilibrium solution $u_e$ and $v_e$ constructed in Section \ref{sec:equilibrium} above. We proceed by substituting into \eqref{eq:frac-gm-full-system} the perturbed solutions $u = u_e + e^{\lambda t}\phi$ and $v = v_e + e^{\lambda t}\psi$ where $|\phi|,|\psi|\ll 1$. To linear order in $\phi$ and $\psi$ we then have the spectral problem
\begin{subequations}\label{eq:stability}
	\begin{align}[left=\empheqlbrace]
		&\lambda\phi + \varepsilon^{2s_1}(-\Delta)^{s_1}\phi + \phi - 2v_e^{-1}u_e\phi + v_e^{-2}u_e^2\psi = 0, & -1<x<1, \label{eq:stability-phi}\\
		&\tau\lambda\psi + D(-\Delta)^{s_2}\psi + \psi - 2u_e\phi = 0, & -1<x<1, \label{eq:stability-psi}\\
		&\phi(x+2)=\phi(x),\qquad \psi(x+2)=\psi(x),& -1<x<1,
	\end{align}
\end{subequations}
for which we seek $\lambda=O(1)$ eigenvalues. If $\Re(\lambda)>0$ (resp. $\Re(\lambda)<0$) then the $N$-spike equilibrium solution is linearly unstable (resp. stable) and we will commonly refer to such an eigenvalue as being unstable (resp. stable). Noting that the diffusivity $\varepsilon^{2s_1}\ll 1$ appearing in \eqref{eq:stability-phi} is asymptotically small we will use the method of matched asymptotic expansions to derive a globally coupled eigenvalue problem (GCEP) from which distinct modes of instabilities and their respective thresholds can be determined.

For each $i=1,..,N$ and $y=O(1)$ we begin by substituting
\begin{equation*}
	\phi(x_i+\varepsilon y) = \Phi_i^\varepsilon(y),\quad \psi(x_i+\varepsilon y) = \Psi_i^\varepsilon(y),
\end{equation*}
into \eqref{eq:stability} to get
\begin{subequations}
	\begin{align}[left=\empheqlbrace]
		& \lambda\Phi_i + (-\Delta)^{s_1}\Phi_i + \Phi_i - 2V_i^{-1}U_i\Phi_i + V_i^{-2}U_i^2\Psi_i = 0,&  -1+x_i <\varepsilon y< 1-x_i \\
		& \tau\lambda\varepsilon^{2s_2}D^{-1}\Psi_i + (-\Delta)^{s_2}\Psi + \varepsilon^{2s_2}D^{-1}\Psi_i - 2U_i\Phi_i = 0,&  -1+x_i <\varepsilon y< 1-x_i.
	\end{align}
\end{subequations}
Assuming that $D\gg O(\varepsilon^{2s_2})$ and exploiting the homogeneity of this system  we obtain the leading order asymptotic expansion
\begin{equation*}
	\Phi_i^\varepsilon \sim c_i\Phi_c^\lambda(y;S_i) + o(1),\qquad \Psi_i^\varepsilon \sim c_i\Psi_c^\lambda(y;S_i) + o(1),
\end{equation*}
where $\Phi_c^\lambda(y;S)$ and $\Psi_c^\lambda(y;S)$ satisfy
\begin{subequations}\label{eq:inner-stability-temp-0}
	\begin{align}[left=\empheqlbrace]
		&(-\Delta)^{s_1}\Phi_c^\lambda + \Phi_c^\lambda - 2V_c^{-1}U_c\Phi_c^\lambda + V_c^{-2}U_c^2\Psi_c^\lambda = -\lambda\Phi_c^\lambda , & -\infty<y<\infty, \\
		&(-\Delta)^{s_2}\Psi_c^\lambda - 2U_c\Phi_c^\lambda = 0, & -\infty<y<\infty, \label{eq:inner-stability-temp-0-psi}
	\end{align}
\end{subequations}
where we assume the general far-field behaviour
\begin{equation}\label{eq:inner-stability-temp-0-psi-lim}
	\Phi_c^\lambda \rightarrow 0,\quad \Psi_c^\lambda \sim B(\lambda,S) + o(1), \qquad \text{as}\quad |y|\rightarrow\infty.
\end{equation}
The undetermined constants $c_1,...,c_N$ correspond to distinct instability modes and moreover yield additional degrees of freedom with which we can normalize the behaviour of solutions to \eqref{eq:inner-stability-temp-0}. We can solve \eqref{eq:inner-stability-temp-0-psi} explicitly as
\begin{equation}\label{eq:Psi_c_temp}
	\Psi_c^\lambda(y;S) = B(\lambda,S) + 2\mathfrak{a}_{s_2}\int_{-\infty}^\infty \frac{U_c(z;S)\Phi_c^\lambda(z;S)}{|y-z|^{1-2s_2}}dz.
\end{equation}
Substituting this back into \eqref{eq:inner-stability-temp-0} then results in the inhomogeneous equation
\begin{subequations}
	\begin{equation}\label{eq:stability-scalar}
		\mathscr{M}\Phi_c^\lambda = \lambda\Phi_c^\lambda + B(\lambda,S)V_c^{-2}U_c^{2}
	\end{equation}
	where the nonlocal operator $\mathscr{M}=\mathscr{M}(S)$ is defined by
	\begin{equation}\label{eq:M-def}
		\mathscr{M}\Phi \equiv -(-\Delta)^{s_1}\Phi - \Phi_c^\lambda + 2V_c^{-1}U_c\Phi - 2\mathfrak{a}_{s_2}V_c^{-2}U_c^2\int_{-\infty}^{\infty}\frac{U_c(z)\Phi(z)}{|y-z|^{1-2s_2}}dz.
	\end{equation}
\end{subequations}
Observe that if $B(\lambda,S) = 0$ then $\lambda$ is an eigenvalue of $\mathscr{M}(S)$ and $\Phi_c^\lambda(y;S)$ the corresponding eigenfunction. If $\lambda$ is not an eigenvalue of $\mathscr{M}$ then we can uniquely solve \eqref{eq:stability-scalar} for $\Phi_c$ which gives
\begin{equation}\label{eq:phi_c_explicit}
	\Phi_c^\lambda(y,S) = B(\lambda,S)(\mathscr{M}-\lambda)^{-1}\bigl(V_c(y;S)^{-1}U_c(y;S)\bigr)^{2}.
\end{equation}
In addition we make note of the far-field behaviour
\begin{equation*}
	\Psi_c^\lambda  \sim B(\lambda,S) + 2\mathfrak{a}_{s_2}|y|^{2s_2-1}\int_{-\infty}^\infty U_c(z;S)\Phi_c^\lambda(z;S)dz\qquad\text{as}\quad |y|\rightarrow\infty,
\end{equation*}
in which the second term vanishes if $\Phi_c^\lambda(\cdot;S)$ is odd. On the other hand if $\Phi_c^\lambda(y;S)$ is not odd then using the additional degrees of freedom granted by $c_1,...,c_N$ we can normalize $\Phi_c^\lambda(y;S)$ such that
\begin{equation}\label{eq:phi_c_normalization}
	\int_{-\infty}^\infty U_c(z;S)\Phi_c^\lambda(z;S)dz = \frac{1}{2\mathfrak{a}_{s_2}},
\end{equation}
with which we get the far-field behaviour
\begin{equation}
	\Psi_c^\lambda  \sim B(\lambda,S) + |y|^{2s_2-1}\qquad \text{as}\quad |y|\rightarrow\infty.
\end{equation}
Note in addition that such a normalization fixes $B(\lambda,S)$which we obtain by multiplying \eqref{eq:phi_c_explicit} by $U_c(y;S)$ and integrating to get
\begin{equation}\label{eq:B-expression}
	B(\lambda,S) = \biggl( 2\mathfrak{a}_{s_2}\int_{-\infty}^\infty U_c(z;S)(\mathscr{M}-\lambda)^{-1}\bigl(V_c(z;S)^{-1}U_c(z;S)\bigr)^{2}dz\biggr)^{-1}.
\end{equation}

We next consider the distributional limit
\begin{equation*}
	2u_e\phi\rightarrow 2\varepsilon^{1-2s_2}D\sum_{i=1}^{N}c_i\int_{-\infty}^\infty U_c(y;S_i)\Phi_c^\lambda(y;S_i)dy\delta(x-x_i)
\end{equation*}
from which we observe that any $i\in\{1,...,N\}$ corresponding to odd-valued $\Phi_c^\lambda(y;S_i)$ will not contribute to the outer problem. A modification of the proceeding calculations in which we keep track of such odd-valued $\Phi_c(\cdot,S_i)$ reveals that such terms do not contribute to the linear stability over an $O(1)$ timescale, though they do contribute to drift instabilities considered in \S\ref{sec:slow-dynamics} below. Without loss of generality we therefore assume that none of the $\Phi_c^\lambda(y;S_i)$ ($i=1,...,N$) are odd-valued. Using the normalization \eqref{eq:phi_c_normalization} we thus obtain the outer problem
\begin{subequations}
	\begin{equation}\label{eq:stability-outer-pde}
		(-\Delta)^{s_2}\psi + \frac{1+\tau\lambda}{D}\psi = \varepsilon^{1-2s_2}\mathfrak{a}_{s_2}^{-1}\sum_{i=1}^Nc_i\delta(x-x_i), \qquad x\in(-1,1)\setminus\{x_1,...,x_N\},
	\end{equation}
	together with the singular behaviour
	\begin{equation}\label{eq:stability-outer-limit}
		\psi(x)\sim c_i\bigl(B(\lambda,S_i)+\varepsilon^{1-2s_2}|x-x_i|^{2s_2-1}\bigr),\quad x\rightarrow x_i,
	\end{equation}
\end{subequations}
for each $i=1,...,N$. The solution to \eqref{eq:stability-outer-pde} can then be expressed in terms of the Green's function satisfying \eqref{eq:greens-equation} as
\begin{equation}
	\psi(x) =\frac{\varepsilon^{1-2s_2}}{\mathfrak{a}_{s_2}}\sum_{i=1}^Nc_iG_{D_\lambda}(x-x_i),\qquad D_\lambda \equiv \frac{D}{1+\tau\lambda}.
\end{equation}
Using \eqref{eq:greens-series-rapid} the matching condition \eqref{eq:stability-outer-limit} then becomes
\begin{equation*}
	\begin{split}
		c_i\bigl(B(\lambda,S_i)+\varepsilon^{1-2s_2}|x-x_i|^{2s_2-1}\bigr) \sim \frac{\varepsilon^{1-2s_2}}{\mathfrak{a}_{s_2}}\biggl( c_i \sum_{k=1}^{k_{\max}}\frac{(-1)^{k-1}\mathfrak{a}_{ks_2}}{D_\lambda^{k-1}}|x-x_i|^{2ks_2-1} + c_iR_{D_\lambda}&(0) \\
		+ \sum_{j\neq i}c_j G_{D_\lambda}(|x_i-x_j|) + O(|x-x_i|)   & \biggr),
	\end{split}
\end{equation*}
as $x\rightarrow x_i$ for each $i=1,...,N$. The leading order singular behaviour immediately balances whereas balancing the leading order constants for each $i=1,...,N$ yields the GCEP
\begin{subequations}\label{eq:gcep}
	\begin{equation}\label{eq:gcep-equation}
		\mathcal{B}(\lambda,\pmb{S})\pmb{c} = \varepsilon^{1-2s_2}\mathfrak{a}_{s_2}^{-1}\mathcal{G}_{D_\lambda}\pmb{c},
	\end{equation}
	where $\pmb{S} = (S_1,...,S_N)^T$, $\pmb{c}=(c_1,...,c_N)^T$, and $\mathcal{B}(\lambda,\pmb{S})$ and $\mathcal{G}_{D_\lambda}$ are $N\times N$ matrices with entries
	\begin{equation}\label{eq:gcep-matrices}
		(\mathcal{B}(\lambda,\pmb{S}))_{ij} = \begin{cases}B(\lambda,S_i), &i=j,\\ 0,& i\neq j,\end{cases}\qquad (\mathcal{G}_{D_\lambda})_{ij} = \begin{cases}R_{D_\lambda}(0), &i=j,\\ G_{D_\lambda}(|x_i-x_j|),& i\neq j.\end{cases}
	\end{equation}
\end{subequations}

In the following subsections we consider the leading order behaviour of the GCEP \eqref{eq:gcep} when $D\ll O(\varepsilon^{2s_2-1})$ and when $D=D_0\varepsilon^{2s_2-1}$. This leading order behaviour will provide insights into the modes of instabilities arising in each of these asymptotic regimes. However, as in the case of the NAS \eqref{eq:NAS} analyzed in \S \ref{sec:equilibrium} we remind the reader that the errors in such leading order approximations will typically be unacceptably large for moderately small value of $\varepsilon>0$. Therefore when we perform full numerical simulations of \eqref{eq:frac-gm-full-system} to support our asymptotic predictions in \S \ref{sec:simulations} we will be numerically computing the relevant stability thresholds from the $\varepsilon$-dependent GCEP \eqref{eq:gcep} directly.

\begin{figure}[t!]
	\centering 
	\begin{subfigure}{0.25\textwidth}
		\includegraphics[width=\linewidth]{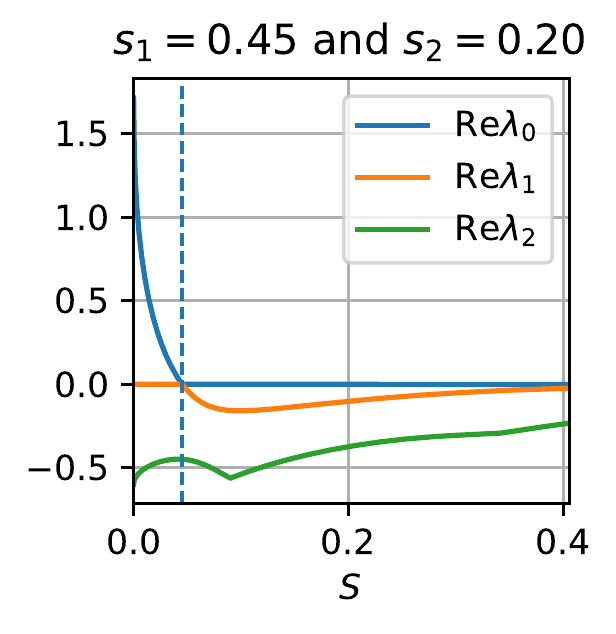}
		\caption{}
		\label{fig:M_spec_su_0.450_sv_0.200}
	\end{subfigure}\hfil 
	\begin{subfigure}{0.25\textwidth}
		\includegraphics[width=\linewidth]{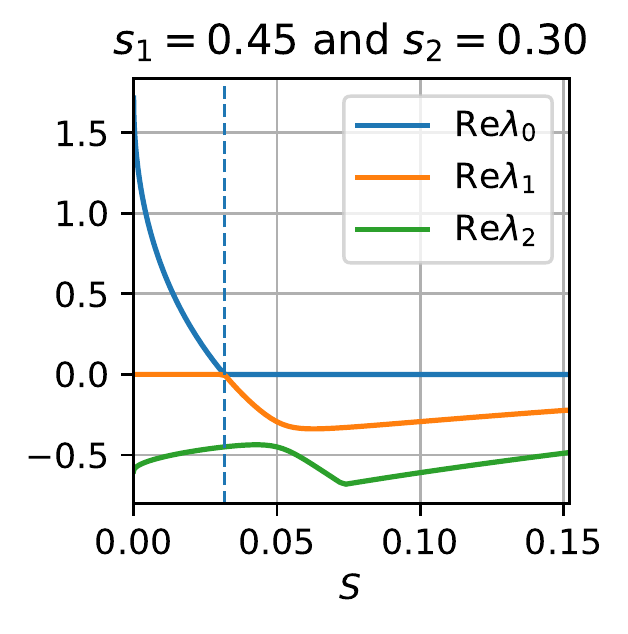}
		\caption{}
		\label{fig:M_spec_su_0.450_sv_0.300}
	\end{subfigure}\hfil 
	\begin{subfigure}{0.25\textwidth}
		\includegraphics[width=\linewidth]{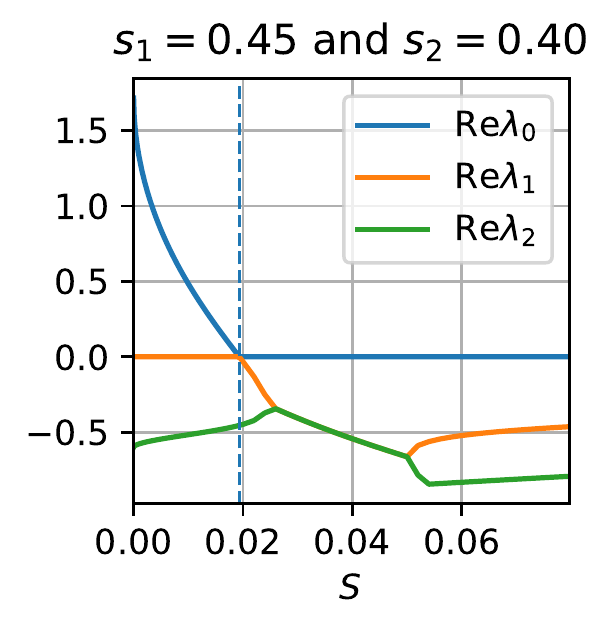}
		\caption{}
		\label{fig:M_spec_su_0.450_sv_0.400}
	\end{subfigure}\hfil 
	\begin{subfigure}{0.25\textwidth}
		\includegraphics[width=\linewidth]{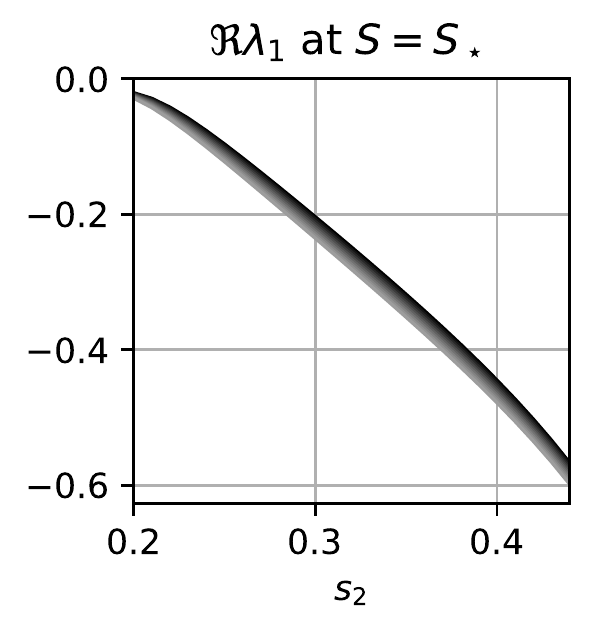}
		\caption{}
		\label{fig:M_spec_S_star}
	\end{subfigure}
	\caption{(A)-(C) Plots of the real part of the eigenvalues of $\mathscr{M}$ versus $0<S<S_\star$ at the indicated values of $s_1$ and $s_2$. In each plot the dashed vertical line corresponds to the value of $S=S_\text{crit}$ at which $\mu'(S_\text{crit})=0$. (D) Plots of $\Re\lambda_1$ at $S=S_\star$ versus $0.2<s_2<0.5$ for different values of $s_1$. The darkest (uppermost) and lightest (lowermost) curves correspond to values of $s_1=0.4$ and $s_1=0.49$ respectively, with the intermediate curves being separated by intervals of $0.01$.}
	\label{fig:M_spectrum}
\end{figure}

\subsection{Linear Stability in the $D\ll O(\varepsilon^{2s_2-1})$ Regime}\label{subsec:stability-reg-1}

We consider first perhaps the simplest case which is when $D\ll O(\varepsilon^{2s_2-1})$ or $D=O(1)$ in particular. From our discussion in \S \ref{sec:equilibrium} we know that in this case all $N$-spike solutions are symmetric to leading order in $\varepsilon\ll 1$ with $S_1=...=S_N=S_\star$. Moreover in this regime the GCEP \eqref{eq:gcep} reduces to the single scalar equation $B(\lambda,S_\star)=0$. Therefore $\lambda$ must be an eigenvalue of the operator $\mathscr{M}(S_\star)$ defined in \eqref{eq:M-def} above with the far-field asymptotics \eqref{eq:inner-stability-temp-0-psi-lim}.

By numerically calculating the spectrum of $\mathscr{M}$ as outlined in Appendix \ref{app:numerical-implementation} we have observed that the dominant eigenvalue is always stable when $S=S_\star$. In Figures \ref{fig:M_spec_su_0.450_sv_0.200}-\ref{fig:M_spec_su_0.450_sv_0.400} we plot the three largest eigenvalues of $\mathscr{M}$. Note that $\lambda=0$ is always an eigenvalue of $\mathscr{M}$ but that this corresponds to the \textit{translational mode} $\Phi = \partial U_c/\partial y$ and $\Psi = \partial V_c/\partial y$ whose analysis is deferred to \S\ref{sec:slow-dynamics} below. Therefore $\lambda_1$ is the appropriate eigenvalue for which \eqref{eq:inner-stability-temp-0-psi-lim} is satisfied when $S=S_\star$ and the plot of $\Re\lambda_1$ at $S=S_\star$ versus $0.2<s_2<0.5$ for select values of $s_1$ in Figure \ref{fig:M_spec_S_star} indicates that this eigenvalue is always stable. In summary, when $D\ll O(\varepsilon^{2s_2-1})$ all $N$-spike solutions are linearly stable to leading order in $\varepsilon\ll 1$.

\subsection{Linear Stability in the $D=O(\varepsilon^{2s_2-1})$ Regime}\label{subsec:stability-reg-2}

In this section we consider the case when $D=D_0\varepsilon^{2s_2-1}$ and for which we will consider the case $D_0\rightarrow \infty$ as a special case. Using the large $D$ asymptotics of the Green's function \eqref{eq:greens-large-D} the GCEP \eqref{eq:gcep} becomes to leading order in $\varepsilon\ll 1$
\begin{equation}\label{eq:gcep-D0}
	\mathcal{B}(\lambda,\pmb{S})\pmb{c} = \frac{\kappa}{1+\tau\lambda}\mathcal{E}_N\pmb{c},\qquad \mathcal{E}_N=\frac{1}{N}\pmb{e}\pmb{e}^T,
\end{equation}
where $\pmb{e}=(1,\cdots,1)^T$ and where we remind the reader that $\kappa = ND_0/(2\mathfrak{a}_{s_2})$. In this section we will consider the linear stability of both the symmetric and asymmetric solutions described in \S \ref{subsec:sym_asym_regime}. We demonstrate that the symmetric $N$-spike solutions are susceptible to two types of instabilities: oscillatory instabilities arising through a Hopf bifurcation, and competition instabilities arising through a zero eigenvalue crossing. On the other hand we will show that asymmetric solutions are always linearly unstable with respect to competition instabilities. The proceeding analysis closely follows previous work done on the three-dimensional Gierer-Meinhardt model \cite{gomez_2021} with its successful adaptation to the present one-dimensional fractional case being due to the properties of $\mu(S)$ described in \S \ref{subsec:core-problem}.

\begin{figure}[t!]
	\centering 
	\begin{subfigure}{0.25\textwidth}
		\includegraphics[scale=0.675]{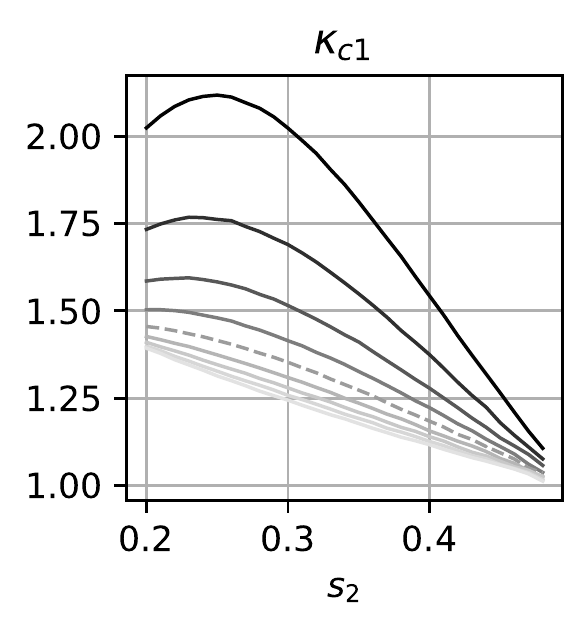}
		\caption{}
		\label{fig:comp_leading_order}
	\end{subfigure}\hfil 
	\begin{subfigure}{0.25\textwidth}
		\includegraphics[scale=0.675]{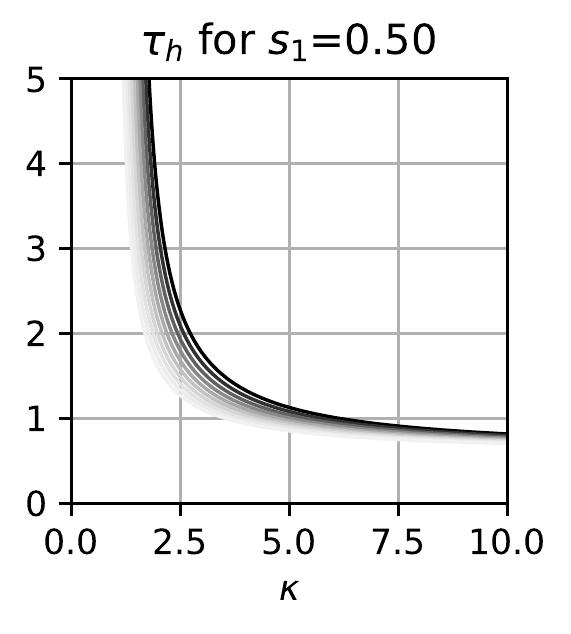}
		\caption{}
		\label{fig:hopf_leading_order_su_0.500}
	\end{subfigure}\hfil 
	\begin{subfigure}{0.25\textwidth}
		\includegraphics[scale=0.675]{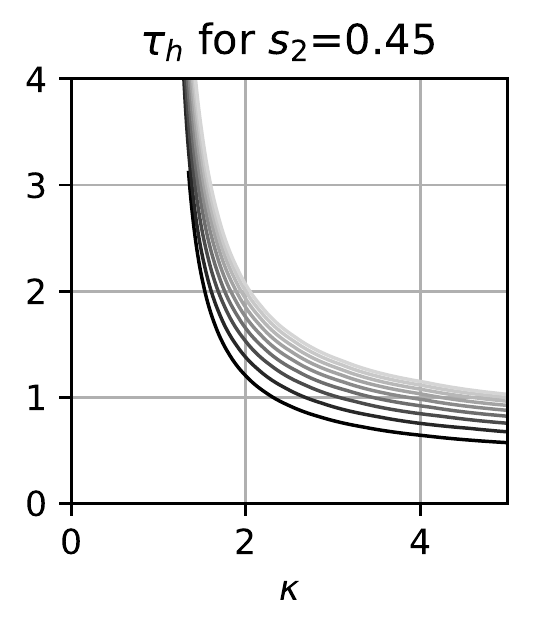}
		\caption{}
		\label{fig:hopf_leading_order_sv_0.450}
	\end{subfigure}%
	\begin{subfigure}{0.25\textwidth}
		\includegraphics[scale=0.675]{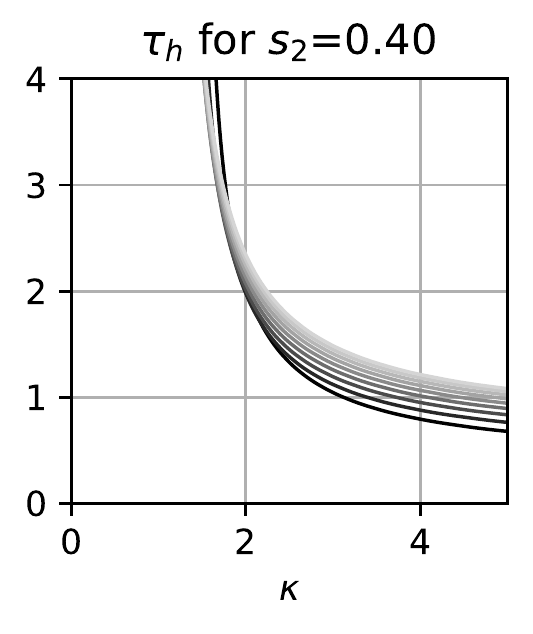}
		\caption{}
		\label{fig:hopf_leading_order_sv_0.400}
	\end{subfigure}
	\caption{(A) Leading order competition instability threshold $\kappa_{c1}$ versus $s_2$. The darkest (uppermost) and lightest (lowermost) curves correspond to values of $s_1=0.3$ and $s_1=0.7$ respectively, with the intermediate curves corresponding to increments of $0.05$. The dashed line corresponds to $s_1=0.5$. (B) The leading order Hopf bifurcation threshold $\tau_h$ at $s_1=0.5$. The darkest and lightest curves correspond to $s_2=0.3$ and $s_2=0.48$ respectively, with the intermediate curves corresponding to $0.02$ increments in $s_2$. (C-D) The leading order Hopf bifurcation threshold $\tau_h$ at indicated values of $s_2=0.45$ and $s_2=0.40$. The darkest and lightest curves in both plots correspond to $s_1=0.3$ and $s_1=0.7$ respectively, with the intermediate curves corresponding to $0.05$ increments in $s_1$.}
	\label{fig:hopf-bifurcations}
\end{figure}

\subsubsection{The Shadow Limit $D_0\rightarrow\infty$}

Before analyzing \eqref{eq:gcep-D0} in general we first consider the \textit{shadow} limit obtained by letting $D_0\rightarrow\infty$. As discussed in \S \ref{subsec:sym_asym_regime} all $N$-spike solutions are then symmetric with $S_c\sim (\mathfrak{b}_{s_1}\mathfrak{a}_{s_2}\kappa^2)^{-1}\ll 1$. Moreover by using the small $S$ asymptotics \eqref{eq:small-S-asymptotics} and the definition of $\mathscr{M}$ given in \eqref{eq:M-def} we readily deduce that
\begin{equation}
	\mathscr{M}\Phi\sim \mathscr{L}\Phi+O(\kappa^{-1}),\qquad\mathscr{L}\Phi\equiv -(-\Delta)^{s_1}\Phi-\Phi+2w_{s_1}\Phi.
\end{equation}
From \eqref{eq:B-expression} we obtain
\begin{equation*}
	B\bigl(\lambda,(\mathfrak{b}_{s_1}\mathfrak{a}_{s_2}\kappa^2)^{-1}\bigr) \sim \frac{\kappa \int_{-\infty}^\infty w_{s_1}(y)^2dy}{ 2 \int_{-\infty}^{\infty} w_{s_1}(y)\bigl(\mathscr{L}-\lambda\bigr)^{-1}w_{s_1}(y)^2 dy},
\end{equation*}
with which \eqref{eq:gcep-D0} becomes
\begin{equation}\label{eq:gcep-D0-infty}
	\frac{\int_{-\infty}^\infty w_{s_1}(y)^2dy}{ 2 \int_{-\infty}^{\infty} w_{s_1}(y)\bigl(\mathscr{L}-\lambda\bigr)^{-1}w_{s_1}(y)^2 dy}\pmb{c} = \frac{1}{1+\tau\lambda}\mathcal{E}_N\pmb{c}.
\end{equation}
Note that the shadow limit case is independent of $s_2$. If $N\geq 2$ then this equation is satisfied if $\pmb{c}$ is any \textit{competition} mode satisfying $c_1+...+c_N=0$ and $\lambda$ is the dominant eigenvalue of $\mathscr{L}$. Since the dominant eigenvalue, $\Lambda_0$, of $\mathscr{L}$ has a positive real part (see Section 4 of \cite{frank_2013_uniqueness}) we therefore deduce that multi-spike solutions in the $D_0\rightarrow\infty$ are always linearly unstable. We refer to the resulting instabilities as \textit{competition} instabilities since the condition $c_1+...+c_N=0$ leads to the growth of some spikes at the expense of the decay of others. If on the other hand $N=1$ then \eqref{eq:gcep-D0-infty} becomes the scalar nonlocal eigenvalue problem (NLEP)
\begin{equation}\label{eq:gcep-D0-infty-scalar}
	1 - \frac{2}{1+\tau\lambda}\frac{\int_{-\infty}^{\infty} w_{s_1}(y)\bigl(\mathscr{L}-\lambda\bigr)^{-1}w_{s_1}(y)^2 dy}{\int_{-\infty}^\infty w_{s_1}(y)^2dy} = 0.
\end{equation}
Following the arguments used in the classical one-dimensional Gierer-Meinhardt system in \cite{ward_2003_shadow} it can be shown that there is a Hopf bifurcation threshold $\tau_h^\infty$ such that all eigenvalues are stable if $\tau<\tau_h^\infty$ whereas there is exactly one complex conjugate pair of unstable eigenvalues when $\tau>\tau_h^\infty$. To calculate this Hopf bifurcation threshold we substitute the purely imaginary eigenvalue $\lambda=i\lambda_I$ into \eqref{eq:gcep-D0-infty-scalar} and isolate real and imaginary parts to get the system
\begin{equation}
	\begin{cases}
		2\Re \bigl(\int_{-\infty}^\infty w_{s_1}(y)\bigl(\mathscr{L}-i\lambda_I\bigr)^{-1}w_{s_1}(y)^2dy\bigr) = \int_{-\infty}^\infty w_{s_1}(y)^2dy,& \\
		2\Im \bigl(\int_{-\infty}^\infty w_{s_1}(y)\bigl(\mathscr{L}-i\lambda_I\bigr)^{-1}w_{s_1}(y)^2dy\bigr) = \tau\lambda_I\int_{-\infty}^\infty w_{s_1}(y)^2dy.&
	\end{cases}
\end{equation}
We can then find the Hopf bifurcation threshold by numerically solving the first equation for $\lambda_I=\lambda_h^\infty(s_1)$ and then substituting into the second equation to get a value for the Hopf bifurcation threshold $\tau=\tau_h^\infty(s_1)$ (for plots of $\tau_h^\infty$ and $\lambda_h^\infty$ see Figure 1A of \cite{gomez_2022}).

In summary, when $D_0\rightarrow\infty$ multi-spike solutions are always linearly unstable due to competition instabilities whereas single spike solutions are linearly stable provided $\tau$ does not exceed the numerically calculated Hopf bifurcation threshold $\tau=\tau_h^\infty(s_1)$. We now address the question of what happens to these competition instability and Hopf bifurcation thresholds for symmetric $N$-spike solutions when $D_0$ is finite.

\subsubsection{Stability Threshold for Symmetric Solutions}\label{subsubsec:symmetric-stability}

We now consider the linear stability of symmetric $N$-spike solutions for which we remind the reader that $S_1=...=S_N=S_c$ where $S_c$ satisfies \eqref{eq:NAS-leading-order-distinguished-regime-symmetric}. To determine the linear stability of these solutions with respect to competition modes we first let $\pmb{c}$ satisfy $c_1+...+c_N=0$. It follows that \eqref{eq:gcep-D0} reduces to $B(\lambda,S_c)=0$ so that $\lambda$ is an eigenvalue of $\mathscr{M}$ whose eigenfunction satisfies the far-field behaviour \eqref{eq:inner-stability-temp-0-psi-lim}. Numerical calculations of the spectrum of $\mathscr{M}$ indicate that this eigenvalue is positive if $S<S_\text{crit}$ whereas it is negative if $S_\text{crit}<S<S_\star$ (see Figure \ref{fig:M_spec_su_0.450_sv_0.200}-\ref{fig:M_spec_su_0.450_sv_0.400}). From \eqref{eq:NAS-leading-order-distinguished-regime-symmetric} and the plots of $\mu(S)$ in Figure \ref{fig:core-far-field} we therefore conclude that symmetric $N$-spike solutions are linearly stable with respect to competition instabilities if $\kappa<\kappa_{c_1}$ and linearly unstable otherwise. Recall that $\kappa_{c1}=\mu(S_\text{crit})/S_\text{crit}$ was previously encountered in \eqref{eq:asym-existence-1} when considering the existence of asymmetric solutions. From the definition of $\kappa$ we can alternatively express this threshold for $\kappa$ as a threshold for the diffusivity
\begin{equation}
	D_{0,\text{comp}} = \frac{2\mathfrak{a}_{s_2}}{N}\frac{\mu(S_\text{crit})}{S_\text{crit}}.
\end{equation}
As in the classical Gierer-Meinhardt model (and other singularly perturbed reaction diffusion systems) the stability of multi-spike solutions decreases as the number of spikes increases. In Figure \ref{fig:comp_leading_order} we plot the leading order competition instability threshold $\kappa_{c_1}$ versus $s_2$ for several values of $s_1$. From which we observe that the competition instability threshold is monotone decreasing in $s_1$. Moreover the threshold decreases monotonically with $s_2$ for $s_1>0.5$ whereas we see that for $s_1<0.5$ it is non-monotone, increasing for smaller values of $s_2$ and then decreasing.

Since $c_1+...+c_N=0$ spans an $(N-1)$-dimensional subspace of $\mathbb{R}^N$ it remains only to consider the \textit{synchronous} modes $\pmb{c}$ for which $c_1=...=c_N$. By substituting such a synchronous mode $\pmb{c}$ into \eqref{eq:gcep-D0} we get
\begin{equation}\label{eq:gcep-D0-synch}
	B(\lambda,S_c) - \frac{\kappa}{1+\tau\lambda} = 0.
\end{equation}
First we show that $\lambda=0$ is not a solution of \eqref{eq:gcep-D0-synch}. Differentiating the core problem \eqref{eq:core_problem} with respect to $S$ we first make the observation that $B(0,S_c)=\mu'(S_c)$ so that after solving \eqref{eq:NAS-leading-order-distinguished-regime-symmetric} for $\kappa$ \eqref{eq:gcep-D0-synch} becomes
\begin{equation*}
	S_c\mu'(S_c) - \mu(S_c) = 0,
\end{equation*}
for which we claim the left-hand-side is strictly negative. This is clearly true for $S_c\geq S_\text{crit}$ since $\mu'(S_c)<0$ (see Figure \ref{fig:core-far-field}). On the other hand for $S_c<S_\text{crit}$ the claim follows by observing that the derivative of the left-hand-side is $\mu''(S_c)<0$ whereas the small $S$ asymptotics \eqref{eq:small-S-asymptotics} imply $S_c\mu'(S_c)-\mu(S_c)\sim -\tfrac{1}{2}\sqrt{S_c/(\mathfrak{b}_{s_1}\mathfrak{a}_{s_2})}<0$ as $S_c\rightarrow 0^+$. Therefore instabilities with respect to the synchronous mode must arise through a Hopf bifurcation. Seeking purely imaginary eigenvalues $\lambda=i\lambda_I$ and separating the real and imaginary parts of \eqref{eq:gcep-D0-synch} we obtain the system
\begin{equation}
	\frac{|B(i\lambda_I,S_c)|^2}{\text{Re}[B(i\lambda_I,S_c)]} - \frac{\mu(S_c)}{S_c} = 0,\qquad \tau = -\frac{\text{Im}[B(i\lambda_I,S_c)]}{\lambda_I\text{Re}[B(i\lambda_I,S_c)]},
\end{equation}
which we can numerically solve for $\lambda_I=\lambda_h(S_c,s_1,s_2)$ from the first equation and then calculate the Hopf bifurcation threshold $\tau=\tau_h(S,s_1,s_2)$ from the second equation. The first equation is numerically solved using Newton's method by slowly increasing $S_c$ starting from a small value for which the shadow-limit value $\lambda_h^\infty(s_1,s_2)$ provides an accurate initial guess. The resulting (leading order) Hopf bifurcation thresholds $\tau_h(S_c,s_1,s_2)$ and associated eigenvalue $\lambda_h(S_c,s_1,s_2)$ are plotted in Figures \ref{fig:hopf_leading_order_su_0.500} and \ref{fig:hopf_leading_order_sv_0.450}. In all cases we observe that the Hopf bifurcation threshold diverges toward $+\infty$ as $S_c\rightarrow S_\text{crit}^-$ and this is a consequence of the nonlocal operator $\mathscr{M}$ having a zero eigenvalue for this value of $S_c$. As discussed in \S \ref{sec:simulations} the Hopf bifurcation threshold can be extended beyond this critical value of $S_c=S_\text{crit}$ but this requires calculating the Hopf bifurcation threshold from the $\varepsilon$-dependent GCEP \eqref{eq:gcep} directly.

\subsubsection{Asymmetric $N$-Spike Solutions are Always Unstable}

We conclude this section on the leading order stability of multi-spike solutions by adapting the analysis for the three-dimensional Gierer-Meinhardt model \cite{gomez_2021} to show that the asymmetric solutions of \S \ref{subsec:sym_asym_regime} are linearly unstable. The analysis follows closely that previously done in \cite{gomez_2021} so we provide only an outline, highlighting the key properties of $\mu(S)$ which allow the adaptation of the analysis in \cite{gomez_2021}.

The key idea in showing that the asymmetric solutions are always linearly unstable is to construct specific modes $\pmb{c}$ for which an instability is guaranteed. Assuming without loss of generality that $S_1=...=S_n=S_r>S_\text{crit}$ and $S_{n+1}=...=S_N=S_l(S_r)$ the leading order GCEP \eqref{eq:gcep-D0} becomes
\begin{equation}\label{eq:asy-stab}
	\begin{pmatrix} B(\lambda,S_r)\mathcal{I}_n & \mathcal{O}_{n,N-n} \\ \mathcal{O}_{N-n,n} & B(\lambda,S_l(S_r))\mathcal{I}_{N-n} \end{pmatrix}\pmb{c} = \frac{\kappa}{1+\tau\lambda}\mathcal{E}_N\pmb{c},
\end{equation}
where $\mathcal{I}_{n}$ is the $n\times n$ identity matrix and $\mathcal{O}_{n,m}$ is the $n\times m$ zero matrix. If $1\leq n\leq N-2$ then the mode $\pmb{c}$ with $c_1=...=c_n=0$ and $c_{n+1}+...+c_N=0$ is immediately seen to be unstable since \eqref{eq:asy-stab} reduces to $B(\lambda,S_l(S_r))=0$ and $S_l(S_r)<S_\text{crit}$ implies the dominant eigenvalue of $\mathscr{M}$ with the far-field behaviour \eqref{eq:inner-stability-temp-0-psi-lim} is unstable (see Figures \ref{fig:M_spec_su_0.450_sv_0.200}-\ref{fig:M_spec_su_0.450_sv_0.400}). Thus, competition \textit{between} the $N-n$ \textit{small} spikes is always destabilizing.

Since the modes considered above are trivial when $n=N-1$ a different argument must be used to show the instability of asymmetric solutions in this case. In particular if $n\geq N-n$ then it it was shown in \cite{gomez_2021} that unstable modes of the form $c_1=...=c_{n}=c_r$ and $c_{n+1}=...=c_N=c_l$ can always be found. The argument used in \cite{gomez_2021} relies on the some key properties of $\mu(S)$ and the spectrum of $\mathscr{M}$. First it requires that $\mu'(S_l)>0$, $\mu'(S_r)<0$, and $S_l'(S_r)>-1$ all of which our numerical calculations indicate are satisfied in the present case (see Figures \ref{fig:core-far-field} and \ref{fig:Sl_prime_plots}). Second it requires that the eigenvalues of $\mathscr{M}$ satisfying the appropriate far-field behaviour \eqref{eq:inner-stability-temp-0-psi-lim} are stable for $S_r>S_\text{crit}$. Since this condition is likewise satisfied (see for example Figures \ref{fig:M_spec_su_0.450_sv_0.200}-\ref{fig:M_spec_su_0.450_sv_0.400}) we are able to adapt the argument from \cite{gomez_2021}  and therefore conclude that all asymmetric $N$-spike solutions are linearly unstable.

\section{Slow Spike Dynamics and the Equilibrium Configurations}\label{sec:slow-dynamics}

It is well known that localization solutions to a variety of singularly perturbed reaction diffusion systems exhibit slow dynamics \cite{iron_2002,tzou_2017_schnakenberg,gomez_2021}. Similar behaviour has likewise been observed for the fractional Gierer-Meinhardt system in one-dimension when $s_2>1/2$ \cite{gomez_2022}. In this section we establish that these slow dynamics persist in the case $1/4<s_2<1/2$ albeit at a different time scale. The dynamics in this parameter regime share qualitative similarities with their classical counterparts in one-, two-, and three-dimensions. Specifically the dynamics are determined by the gradient of the Green's function which leads to a mutual repulsion between spikes. The derivation of the slow dynamics however more closely resembles that for the three-dimensional Gierer-Meinhardt and Schnakenberg systems \cite{gomez_2021,tzou_2017_schnakenberg} owing to, as in previous sections, the strong coupling between the activator and inhibitor in the inner region. In this section we formally derive the equations governing the slow dynamics of a multi-spike solution and in \S \ref{subsec:simulations-slow-dynamics} we validate our theory with numerical examples of two-spike solutions (see also Figure \ref{fig:slow-dynamics-waterfall} in \S \ref{sec:intro}).

Slow spike dynamics are the result of higher order corrections so we begin by first substituting  $x=x_i+\varepsilon y$ into \eqref{eq:v-outer} to obtain the higher order expansion
\begin{gather*}
	v\sim \varepsilon^{1-4s_2}D\mathfrak{a}_{s_2}^{-1}\biggl[ S_i^\varepsilon\biggl(\sum_{k=1}^{k_\text{max}} \frac{(-1)^{k+1}\mathfrak{a}_{ks_2}}{D^{k-1}}|y|^{2ks_2-1}\varepsilon^{2ks_2-1} + R_D(0) + \varepsilon R_D'(0)y + O(\varepsilon^2)\biggr) \\
	+ \sum_{j\neq i} S_j^\varepsilon G_D(x_i-x_j) + \varepsilon \beta_{i,1} y + O(\varepsilon^2)\biggr],
\end{gather*}
where we have defined
\begin{equation}
	\beta_{i,1} = \sum_{j\neq i} S_j^\varepsilon G_D'(x_i-x_j).
\end{equation}
Note that due to the periodic boundary conditions we have $R_D'(0)=0$. This implies, as subsequent calculations will show, that the dynamics of individual spikes are independent of their absolute position in the interval $-1<x<1$ but are due solely to interactions between spikes. Next we refine the inner expansion \eqref{eq:inner-expansion-0} by letting
\begin{subequations}\label{eq:inner-expansion-1}
	\begin{gather}
		u(x_i+\varepsilon y) \sim \varepsilon^{-2s_2}D\bigl(U_i^\varepsilon + \Phi_i^\varepsilon + \text{h.o.t.} \bigr),\\
		v(x_i+\varepsilon y) \sim \varepsilon^{-2s_2}D\bigl(V_i^\varepsilon + \Psi_i^\varepsilon + \varepsilon^{2-2s_2}\mathfrak{a}_{s_2}^{-1}\beta_{i,1}y + \text{h.o.t.} \bigr),
	\end{gather}
\end{subequations}
where $U_i^\varepsilon \equiv U_c(y;S_i^\varepsilon)$, $V_i^\varepsilon\equiv V_c(y;S_i^\varepsilon)$, $|\Phi_i^\varepsilon|\ll U_i^\varepsilon$ and $|\Psi_i^\varepsilon|\ll V_i^\varepsilon$, and where $\text{h.o.t.}$ refers to higher order terms whose order will become evident after the asymptotic expansions are carried out. Substituting \eqref{eq:inner-expansion-1} into \eqref{eq:frac-gm-full-system} we find that $\pmb{\Phi}_i^\varepsilon \equiv (\Phi_i^\varepsilon,\Psi_i^\varepsilon)^T$ satisfies $\mathcal{L}_i^\varepsilon \pmb{\Phi_i}^\varepsilon = \pmb{f}_i^\varepsilon$, where
\begin{equation}
	\mathcal{L}_i^\varepsilon \equiv \begin{pmatrix} (-\Delta)^{s_1} + 1 - 2 \tfrac{U_i^\varepsilon}{V_i^\varepsilon} & \bigl(\tfrac{U_i^\varepsilon}{V_i^\varepsilon}\bigr)^2 \\ -2U_i^\varepsilon & (-\Delta)^{s_2} \end{pmatrix},\quad \pmb{f}_i^\varepsilon \equiv \begin{pmatrix} \tfrac{1}{\varepsilon}\tfrac{dx_i}{dt}\tfrac{dU_i^\varepsilon}{dy} - \varepsilon^{2-2s_2}\mathfrak{a}_{s_2}^{-1}\bigl(\tfrac{U_i^\varepsilon}{V_i^\varepsilon}\bigr)^2 \beta_{i,1}y \\ -\tfrac{\varepsilon^{2s_2}}{D}V_i^\varepsilon \end{pmatrix}.
\end{equation}
We observe that $(\tfrac{d}{dy}U_i^\varepsilon,\tfrac{d}{dy}V_i^\varepsilon)^T$ is in the kernel of $\mathcal{L}_i^\varepsilon$ and assume that the kernel of $(\mathcal{L}_i^\varepsilon)^T$ is likewise one-dimensional and spanned by $\pmb{P}_i^\varepsilon \equiv (P_i^\varepsilon,Q_i^\varepsilon)^T$. We can then impose the solvability condition
\begin{equation*}
	0 = \int_{-\infty}^\infty (\pmb{P}_i^\varepsilon)^T\mathcal{L}_i^\varepsilon\pmb{\Phi}_i^\varepsilon dy = \tfrac{1}{\varepsilon}\tfrac{dx_i}{dt}\int_{-\infty}^{\infty} P_i^\varepsilon\tfrac{dU_i^\varepsilon}{dy}dy - \varepsilon^{2-2s_2}\mathfrak{a}_{s_2}^{-1}\beta_{i,1}\int_{-\infty}^{\infty}yP_i^\varepsilon\bigl(\tfrac{U_i^\varepsilon}{V_i^\varepsilon}\bigr)^{2}dy - \tfrac{\varepsilon^{2s_2}}{D}\int_{-\infty}^\infty Q_i^\varepsilon V_i^\varepsilon dy.
\end{equation*}
Numerical calculations indicate that $P_i^\varepsilon$ and $Q_i^\varepsilon$ are odd so that the final term vanishes and therefore
\begin{equation}\label{eq:slow-dynamics}
	\frac{dx_i}{dt} \sim \mathfrak{a}_{s_2}\varepsilon^{3-2s_2}\frac{\int_{-\infty}^\infty y P_i^\varepsilon\bigl(\tfrac{U_i^\varepsilon}{V_i^\varepsilon}\bigr)^{2}dy}{\int_{-\infty}^\infty P_i^\varepsilon \tfrac{dU_i^\varepsilon}{dy}dy}\sum_{j\neq i} S_j^\varepsilon G_D'(x_i-x_j),\qquad (i=1,...,N).
\end{equation}
Together with the NAS \eqref{eq:NAS} this constitutes a differential algebraic system for the $N$ spike locations $x_1,...,x_N$ and their strengths $S_1^\varepsilon,...,S_N^\varepsilon$.

We immediately observe that \eqref{eq:slow-dynamics} implies that the slow dynamics occur over a slow $O(\varepsilon^{2s_2-3})$ timescale. Furthermore since $P_i^\varepsilon$ is odd in $y$ we also have
\begin{equation}
	\frac{\int_{-\infty}^\infty y P_i^\varepsilon\bigl(\tfrac{U_i^\varepsilon}{V_i^\varepsilon}\bigr)^{2}dy}{\int_{-\infty}^\infty P_i^\varepsilon \tfrac{dU_i^\varepsilon}{dy}dy} \leq 0.
\end{equation}
If $x_i\gtrless x_j$ then $G_D'(x_i-x_j)\lessgtr 0$ and we thus conclude that spikes are mutually repulsing. In particular it is easy to see from \eqref{eq:slow-dynamics} that two-spike solutions are stationary if and only if $|x_1-x_2|=1$. In \S \ref{subsec:simulations-slow-dynamics} we compare the slow-dynamics predicted by the differential algebraic system \eqref{eq:slow-dynamics} and \eqref{eq:NAS} with numerical simulations of \eqref{eq:frac-gm-full-system} for two spike solutions that are initially separated by a distance $|x_1(0)-x_2(0)|<1$.

We conclude this section by outlining how to calculate the function $P_i^\varepsilon$ needed to evaluate the coefficient appearing in \eqref{eq:slow-dynamics}. Following the analysis of \S \ref{sec:stability} we write
\begin{equation*}
	Q_i^\varepsilon = C_i^\varepsilon - \mathfrak{a}_{s_2}\int_{-\infty}^\infty\frac{(U_i^\varepsilon(z)/V_i^\varepsilon(z))^2 P_i^\varepsilon(z)}{|y-z|^{1-2s_2}}dz,
\end{equation*}
from which we deduce that $P_i^\varepsilon$ solves $\mathscr{M}^\star(S_i^\varepsilon) P_i^\varepsilon = -2C_i^\varepsilon U_i^\varepsilon$ where we define the adjoint operator $\mathscr{M}^\star=\mathscr{M}^\star(S)$ by
\begin{equation}
	\mathscr{M}^\star(S) P \equiv -(-\Delta)^{s_1} P - P + 2\frac{U_c}{V_c}P-2\mathfrak{a}_{s_2}U\int_{-\infty}^\infty\frac{(U_c(z)/V_c(z))^2 P(z)}{|y-z|^{1-2s_2}}dz.
\end{equation}
Numerical calculations (not included) indicate that the adjoint operator $\mathscr{M}^\star$, like $\mathscr{M}$ in \S \ref{sec:stability}, has exactly one zero eigenvalue for all $0<S<S_\star$ except at $S=S_\text{crit}$ for which it has exactly two zero eigenvalues. In particular assuming $S_i^\varepsilon \neq S_\text{crit}$ we may set $C_i^\varepsilon=0$ and thus deduce that $P_i^\varepsilon$ is in the kernel of $\mathscr{M}^\star(S_i^\varepsilon)$.

\section{Numerical Simulations}\label{sec:simulations}

In this section we numerically simulate the fractional Gierer-Meinhardt system \eqref{eq:frac-gm-full-system} to support our asymptotic calculations in the preceding section. Using the asymptotically constructed solutions from \S \ref{sec:equilibrium} as initial conditions we choose parameter values to support the stability thresholds found by numerically solving \eqref{eq:NAS}. We proceed in three parts. In the first we consider Hopf bifurcations of single spike solutions, in the second we consider competition instabilities of two-spike solutions, and in the third and final part we consider the slow dynamics of two-spike solutions. In the first two parts we will first numerically compute the corresponding $\varepsilon$-dependent stability thresholds and compare them with their leading order counterparts. As emphasized in \S \ref{sec:stability}, due to the fractional powers of $\varepsilon$ in the asymptotic expansions of the stability thresholds we anticipate that the leading order thresholds deviate substantially from those obtained by solving \eqref{eq:NAS} directly. Finally, when considering the slow-dynamics of two-spike solutions in the third part we will choose parameter values for which the two-spike solutions are linearly stable with respect to Hopf and competition instabilities.

\subsection{Hopf Bifurcation of One-Spike Solutions}\label{subsec:simulations-hopf}

\begin{figure}[t!]
	\centering 
	\begin{subfigure}{0.25\textwidth}
		\includegraphics[width=\linewidth]{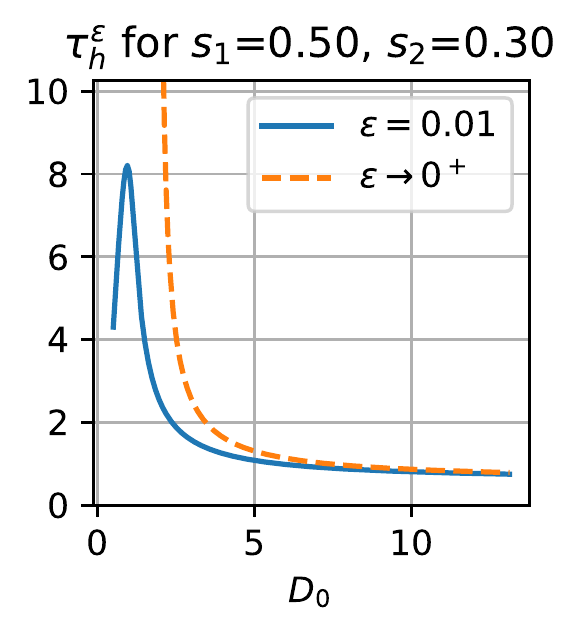}
		\caption{}
		\label{fig:hb-eps-0p01-1}
	\end{subfigure}\hfil 
	\begin{subfigure}{0.25\textwidth}
		\includegraphics[width=\linewidth]{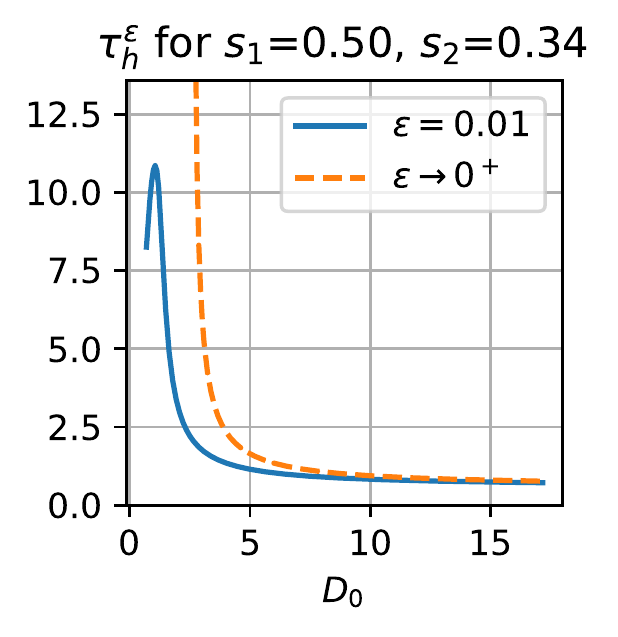}
		\caption{}
		\label{fig:hb-eps-0p01-2}
	\end{subfigure}\hfil 
	\begin{subfigure}{0.25\textwidth}
		\includegraphics[width=\linewidth]{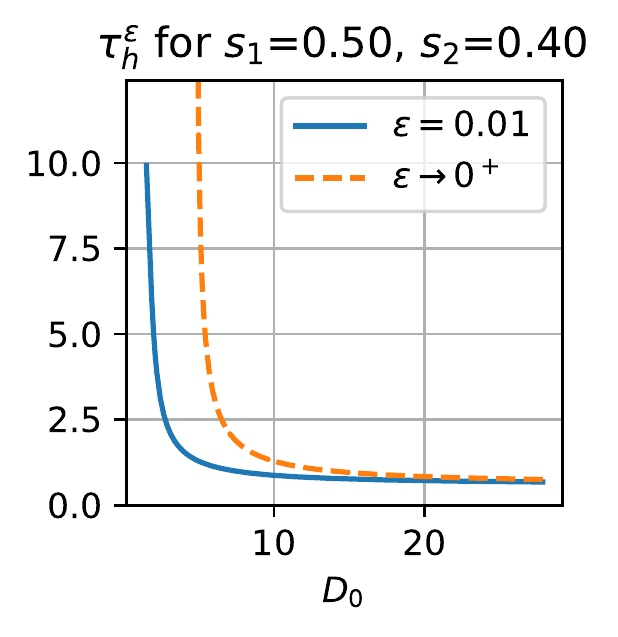}
		\caption{}
		\label{fig:hb-eps-0p01-3}
	\end{subfigure}\hfil 
	\begin{subfigure}{0.25\textwidth}
		\includegraphics[width=\linewidth]{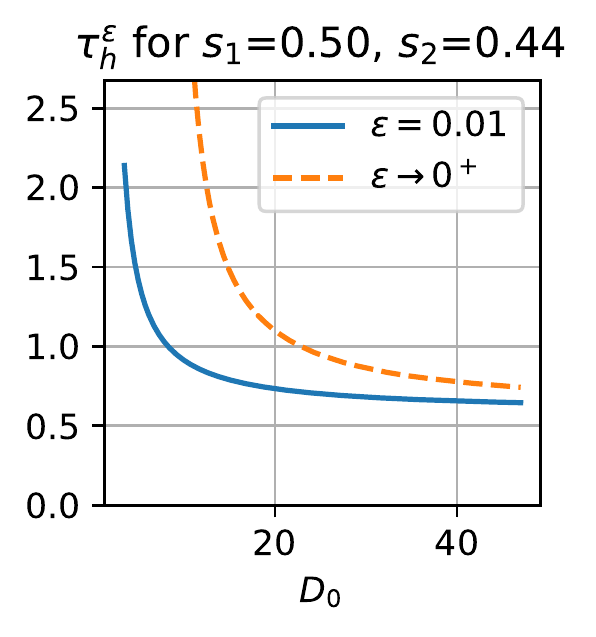}
		\caption{}
		\label{fig:hb-eps-0p01-4}
	\end{subfigure}
	\medskip
	\begin{subfigure}{0.25\textwidth}
		\includegraphics[width=\linewidth]{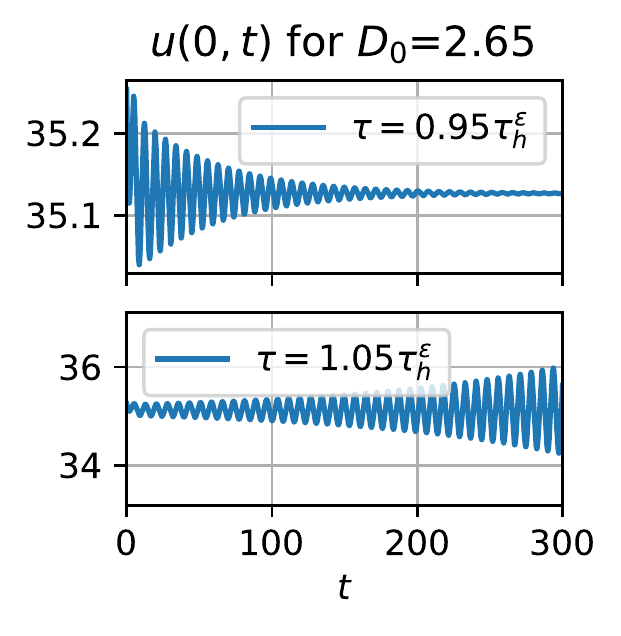}
		\caption{}
		\label{fig:hb-eps-0p01-sim-1}
	\end{subfigure}\hfil 
	\begin{subfigure}{0.25\textwidth}
		\includegraphics[width=\linewidth]{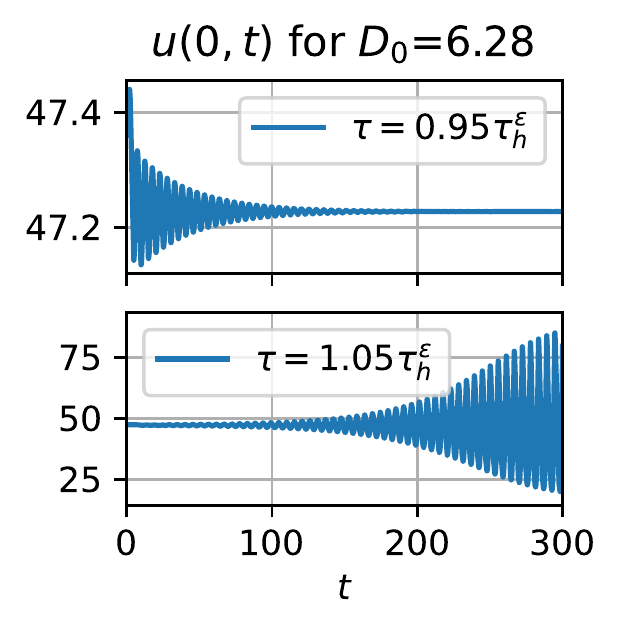}
		\caption{}
		\label{fig:hb-eps-0p01-sim-2}
	\end{subfigure}\hfil 
	\begin{subfigure}{0.25\textwidth}
		\includegraphics[width=\linewidth]{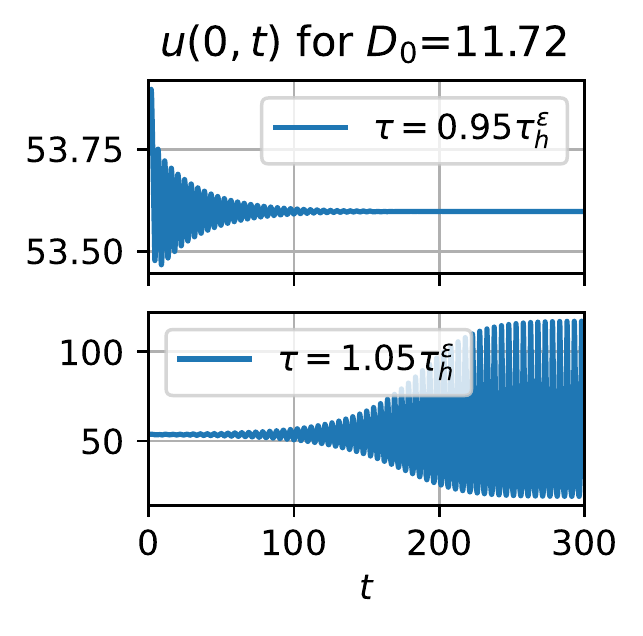}
		\caption{}
		\label{fig:hb-eps-0p01-sim-3}
	\end{subfigure}\hfil 
	\begin{subfigure}{0.25\textwidth}
		\includegraphics[width=\linewidth]{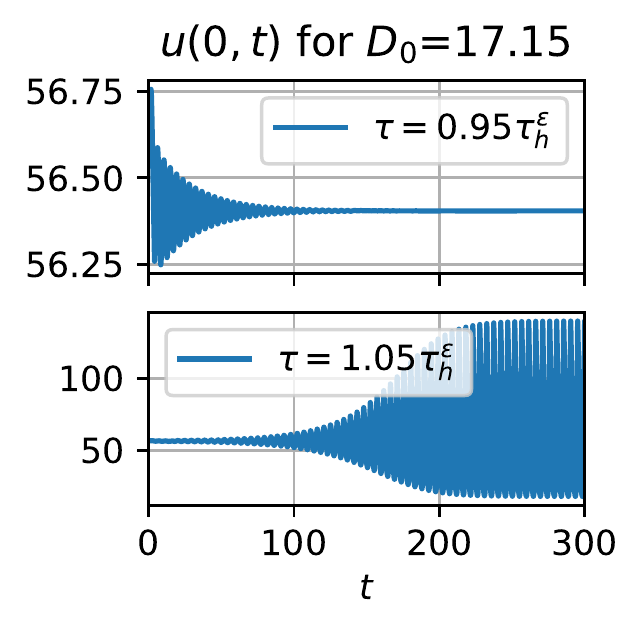}
		\caption{}
		\label{fig:hb-eps-0p01-sim-4}
	\end{subfigure}
	\caption{(A)-(D) Hopf bifurcation thresholds for a one-spike solution obtained by numerically solving the $\varepsilon$-dependent system \eqref{eq:hopf-N=1} with $\varepsilon=0.01$. (E)-(F) Plots of $u(0,t)$ from numerically simulating the fractional Gierer-Meinhardt system with $\varepsilon=0.01$, $s_1=0.5$, $s_2=0.34$, and indicated values of $D_0$ with $\tau=0.95\tau_h^\varepsilon(D_0)$ (top) and $\tau=1.05\tau_h^\varepsilon(D_0)$ (bottom). In each case a single spike solution (obtained using the asymptotics of \S \ref{sec:equilibrium}) centred at $x=0$ with multiplicative noise was used as the initial condition.}
	\label{fig:hopf-bifurcation-eps}
\end{figure}

We first verify the Hopf bifurcation threshold for a single spike solution centred, without loss of generality, at $x=0$. With $N=1$ the NAS \eqref{eq:NAS} and GCEP \eqref{eq:gcep} become
\begin{subequations}\label{eq:hopf-N=1}
	\begin{align}[left=\empheqlbrace]
		\mu(S_c) = \mathfrak{a}_{s_2}^{-1} \varepsilon^{1-2s_2}R_{D_0\varepsilon^{2s_2-1}}(0) S_c,\label{eq:N=1-NAS}\\
		B(i\lambda_I,S_c) = \frac{\mathfrak{a}_{s_2}^{-1}\varepsilon^{1-2s_2}}{1+i\tau\lambda_I}R_{\tfrac{D_0\varepsilon^{2s_2-1}}{1+i\tau\lambda_I}}(0), \label{eq:N=1-GCEP}
	\end{align}
\end{subequations}
where we remind the reader that $R_D(x)$ is given by \eqref{eq:greens-series-rapid-R_D}. For a given value of $D_0$ we first solve \eqref{eq:N=1-NAS} for $S_c=S_c^\varepsilon$. Separating real and imaginary parts in \eqref{eq:N=1-GCEP} we can then numerically solve for the Hopf bifurcation threshold $\tau=\tau_h^\varepsilon$ and accompanying eigenvalue $\lambda_I=\lambda_h^\varepsilon$. Specifically we solve the resulting system with Newton's method starting with a large value of $D_0$ for which the shadow limit solutions $\tau_h^\infty$ and $\lambda_I^\infty$ are good initial guesses. Using $\varepsilon=0.01$ the resulting $\varepsilon$-dependent Hopf bifurcation thresholds are shown in Figures \ref{fig:hb-eps-0p01-1}-\ref{fig:hb-eps-0p01-4} which illustrate the persistence of the Hopf bifurcation threshold for $S>S_\text{crit}$ not captured by the leading order theory. To support our asymptotically predicted threshold we performed several numerical simulations of the full system \eqref{eq:frac-gm-full-system} with $\varepsilon=0.01$ and using a single spike solution centred at the origin as an initial condition. In Figures \ref{fig:hb-eps-0p01-sim-1}-\ref{fig:hb-eps-0p01-sim-4} we plot $u(0,t)$ when $s_1=0.5$ and $s_2=0.34$ for select values of $D_0$ and values of $\tau$ slightly below and slightly above the Hopf bifurcation threshold, all of which validate the Hopf bifurcations thresholds from the asymptotic theory.

\subsection{Competition Instabilities of Two-Spike Solutions}\label{subsec:simulations-competition}

\begin{figure}[t!]
	\centering 
	\begin{subfigure}{0.25\textwidth}
		\includegraphics[width=\linewidth]{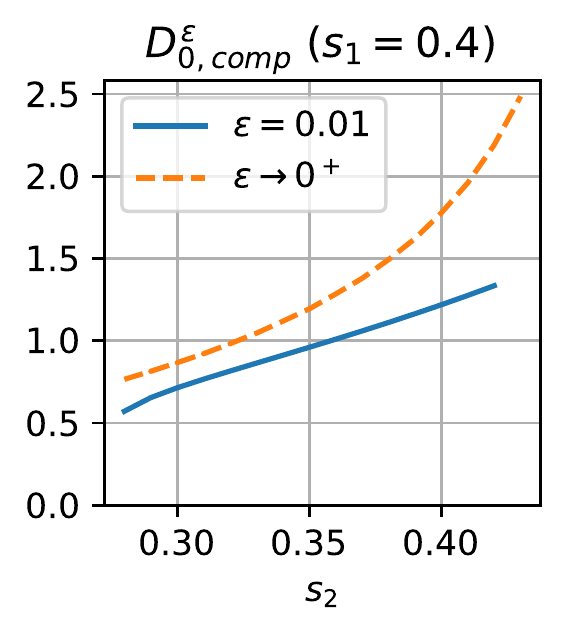}
		\caption{}
		\label{fig:comp-eps-0p01-1}
	\end{subfigure}\hfil 
	\begin{subfigure}{0.25\textwidth}
		\includegraphics[width=\linewidth]{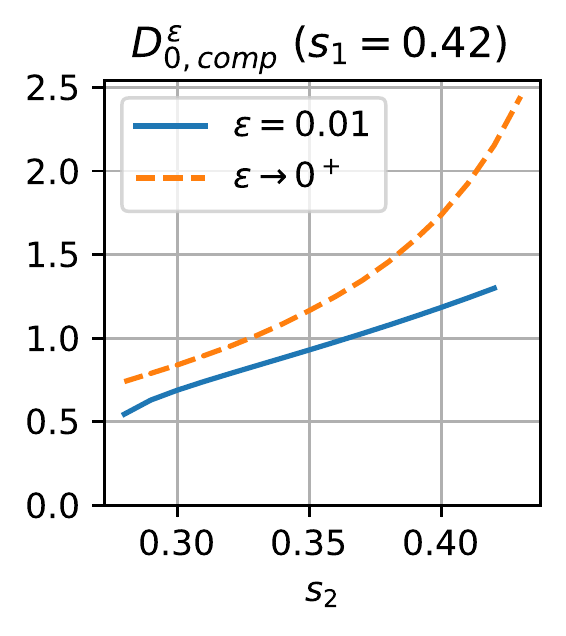}
		\caption{}
		\label{fig:comp-eps-0p01-2}
	\end{subfigure}\hfil 
	\begin{subfigure}{0.25\textwidth}
		\includegraphics[width=\linewidth]{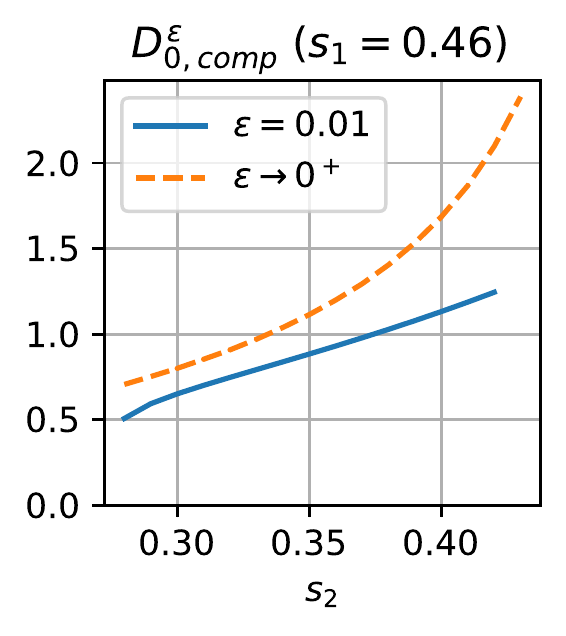}
		\caption{}
		\label{fig:comp-eps-0p01-3}
	\end{subfigure}\hfil 
	\begin{subfigure}{0.25\textwidth}
		\includegraphics[width=\linewidth]{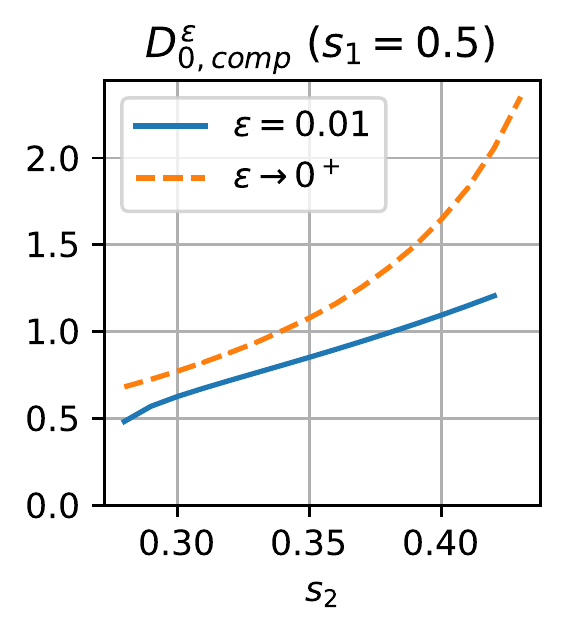}
		\caption{}
		\label{fig:comp-eps-0p01-4}
	\end{subfigure}
	\medskip
	\begin{subfigure}{0.25\textwidth}
		\includegraphics[width=\linewidth]{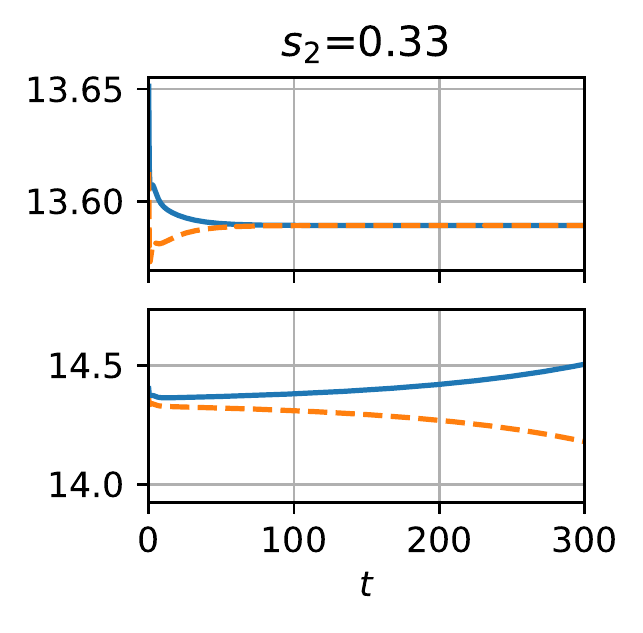}
		\caption{}
		\label{fig:comp-eps-0p01-sim-1}
	\end{subfigure}\hfil 
	\begin{subfigure}{0.25\textwidth}
		\includegraphics[width=\linewidth]{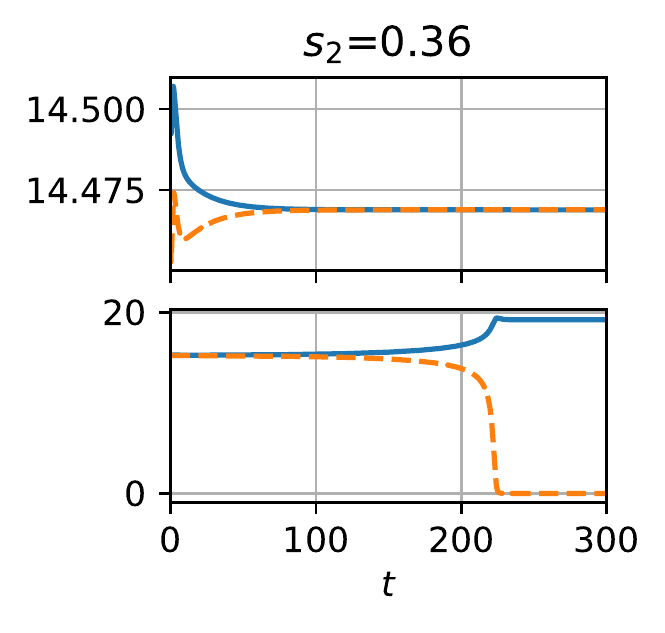}
		\caption{}
		\label{fig:comp-eps-0p01-sim-2}
	\end{subfigure}\hfil 
	\begin{subfigure}{0.25\textwidth}
		\includegraphics[width=\linewidth]{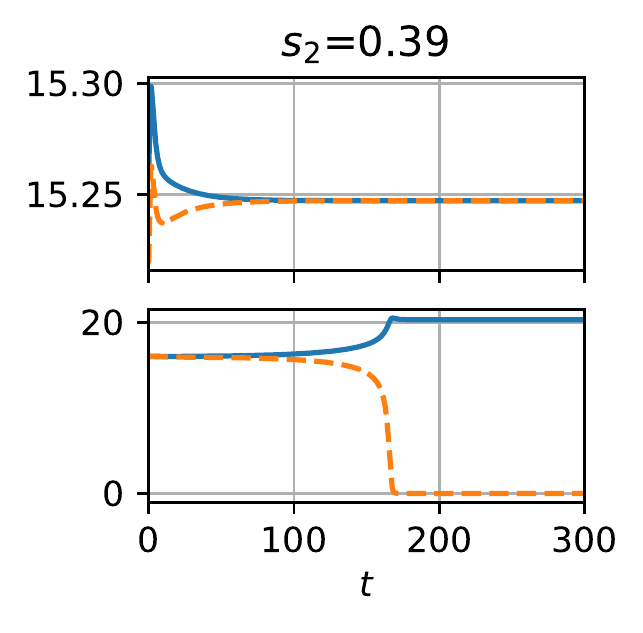}
		\caption{}
		\label{fig:comp-eps-0p01-sim-3}
	\end{subfigure}\hfil 
	\begin{subfigure}{0.25\textwidth}
		\includegraphics[width=\linewidth]{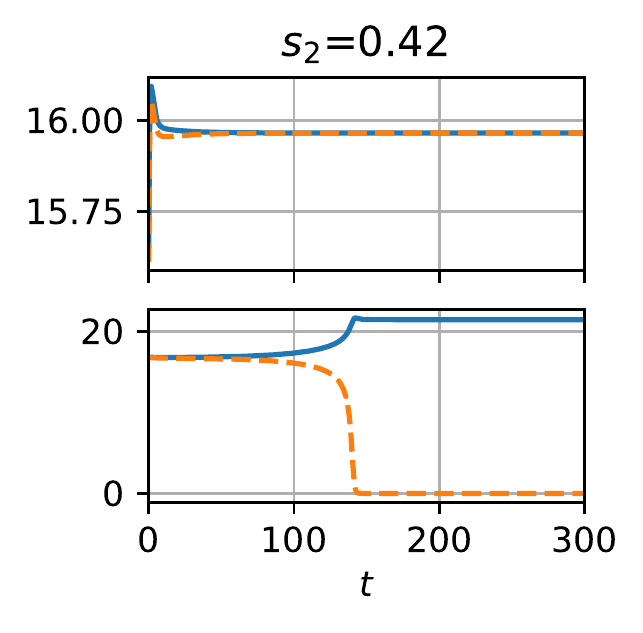}
		\caption{}
		\label{fig:comp-eps-0p01-sim-4}
	\end{subfigure}
	\caption{(A)-(D) Competition instability thresholds for a two-spike solution obtained by numerically solving the $\varepsilon$-dependent system \eqref{eq:comp-N=2} with $\varepsilon=0.01$. (E)-(F) Plots of $u(x_1,t)$ (solid blue) and $u(x_2,t)$ (dashed orange) from numerically simulating the fractional Gierer-Meinhardt system with $\varepsilon=0.01$, $s_1=0.5$, at the indicated values of $s_2$  with $D_0=0.95D_{0,\text{comp}}^\varepsilon$ (top) and $D_0=1.05D_{0,\text{comp}}^\varepsilon$ (bottom). In each case a two spike solution (obtained using the asymptotics of \S \ref{sec:equilibrium}) separated by a distance of $|x_1-x_2|=1$ with multiplicative noise was used as the initial condition.}
	\label{fig:comp-bifurcation-eps}
\end{figure}

Turning our attention now to the case of a symmetric $N=2$-spike solution we perform numerical simulations to verify the onset of competition instabilities as predicted by our stability theory. We assume that $|x_1-x_2|=1$ so that there are no small eigenvalues or, equivalently, there are no slow dynamics as discussed in \S \ref{sec:slow-dynamics}. In this case the NAS \eqref{eq:NAS} and GCEP \eqref{eq:gcep} with $\lambda=0$ become
\begin{subequations}\label{eq:comp-N=2}
	\begin{align}[left=\empheqlbrace]
		\mu(S_c) = \mathfrak{a}_{s_2}^{-1}\varepsilon^{1-2s_2}  \bigl(R_{D_0\varepsilon^{2s_2-1}}(0) + G_{D_0\varepsilon^{2s_2-1}}(1)\bigr)S_c, \\
		\mu'(S_c) = \mathfrak{a}_{s_2}^{-1}\varepsilon^{1-2s_2} \bigl(R_{D_0\varepsilon^{2s_2-1}}(0) - G_{D_0\varepsilon^{2s_2-1}}(1)\bigr).
	\end{align}
\end{subequations}
We can numerically solve this system for $D_0$ as a function of $s_2$ at select values of $s_1$. Doing so with $\varepsilon=0.01$ we obtain the higher order competition instability threshold shown in Figures \ref{fig:comp-eps-0p01-1}-\ref{fig:comp-eps-0p01-4}. In contrast to the leading order competition threshold which can be calculated as in \S \ref{subsubsec:symmetric-stability} there is an upper limit to the value of $s_2$ for which we can compute the higher order $\varepsilon$-dependent threshold from \eqref{eq:comp-N=2}. This is a consequence of the change in sign of $R_D(0)$ for smaller values of $D$ as described in \S \ref{subsec:greens-func-properties}. For sufficiently small values of $\varepsilon$ the value of $D=D_0\varepsilon^{2s_2-1}$ will always exceed this threshold and a competition instability threshold $D_{0,\text{comp}}^\varepsilon$ can be calculated for values of $s_2$ closer to $1/2$. Otherwise higher order correction terms need to be calculated or the inhibitor in the numerical discretization of the core problem \eqref{eq:core_problem} needs to be allowed to become negative as described in \S \ref{subsec:greens-func-properties}. We will not address these additional technical difficulties further.

To support our asymptotically calculated higher order competition instability threshold we performed several numerical experiments. In each experiment we use the methods of \S \ref{sec:equilibrium} to asymptotically construct a symmetric two-spike solution with spikes centred at $x_1=-0.5$ and $x_2=0.5$. Using this solution as the initial condition we then solve \eqref{eq:frac-gm-full-system} numerically for select values of $s_1$ and $s_2$ and with small value of $\tau=0.05$ (so that there are no Hopf bifurcations)  as well as values of $D=D_0\varepsilon^{2s_2-1}$ such that $D_0$ is either slightly below or slightly above the numerically calculated competition instability threshold $D_{0,\text{comp}}^\varepsilon$. In each case we found good agreement with the higher order calculated threshold $D_{0,\text{comp}}^\varepsilon$ and in Figure \ref{fig:comp-eps-0p01-sim-1}-\ref{fig:comp-eps-0p01-sim-4} we show a sampling of numerically calculated values of the spike heights $u(x_1,t)$ and $u(x_2,t)$ for values of $D_0=0.95D_{0,\text{comp}}^\varepsilon$ (top) and $D_0=1.05D_{0,\text{comp}}^\varepsilon$ (bottom).

\subsection{Slow Dynamics of Two-Spike Solutions}\label{subsec:simulations-slow-dynamics}

\begin{figure}[t!]
	\centering 
	\begin{subfigure}{0.25\textwidth}
		\includegraphics[width=\linewidth]{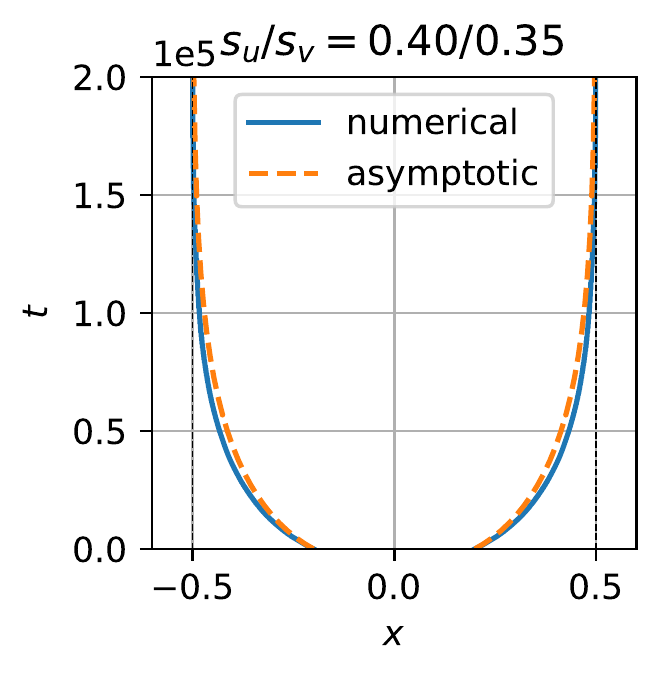}
		\caption{}
		\label{fig:slow-dyn-traj-example-000}
	\end{subfigure}\hfil 
	\begin{subfigure}{0.25\textwidth}
		\includegraphics[width=\linewidth]{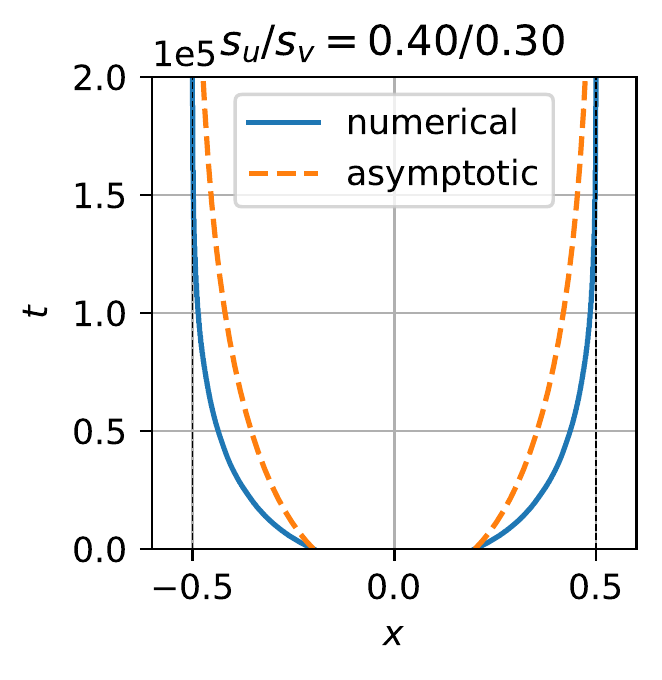}
		\caption{}
		\label{fig:slow-dyn-traj-example-001}
	\end{subfigure}\hfil 
	\begin{subfigure}{0.25\textwidth}
		\includegraphics[width=\linewidth]{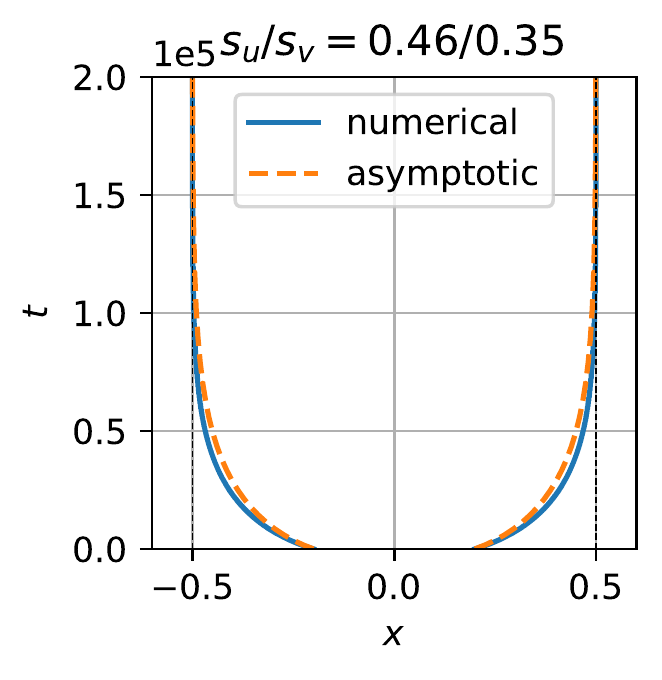}
		\caption{}
		\label{fig:slow-dyn-traj-example-002}
	\end{subfigure}\hfil 
	\begin{subfigure}{0.25\textwidth}
		\includegraphics[width=\linewidth]{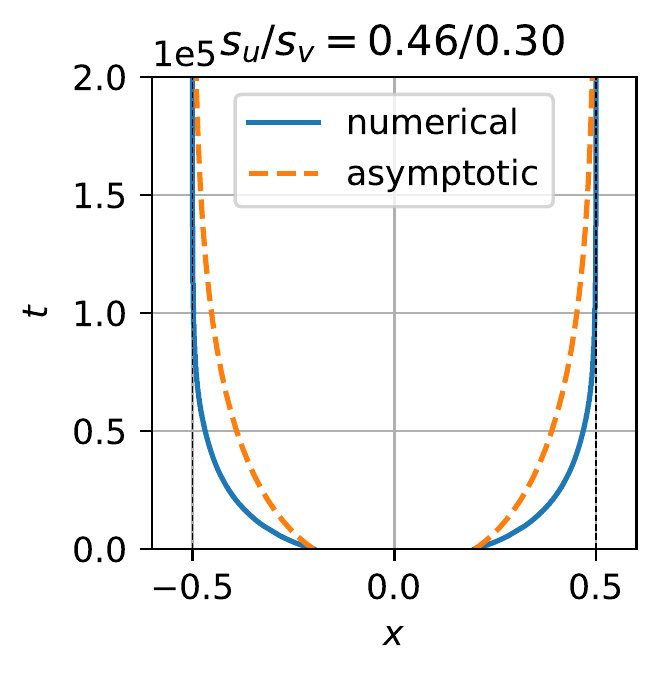}
		\caption{}
		\label{fig:slow-dyn-traj-example-003}
	\end{subfigure}
	\caption{Time evolution of spike locations $x_1(t)$ and $x_2(t)$ in a two-spike solution. In each example $\tau=0.1$, $\varepsilon=0.01$, and $D=D_0\varepsilon^{2s_2-1}$ where $D_0=0.8 D_{0,\text{comp}}^\varepsilon$ and $D_{0,\text{comp}}^\varepsilon$ is the corresponding two-spike competition instability threshold.}
	\label{fig:simulation-trajectories}
\end{figure}

We conclude the numerical validation of our asymptotic theory by considering the slow dynamics of symmetric two-spike solutions. Using the translational invariance granted by the periodic boundary conditions we reduce the differential algebraic system \eqref{eq:slow-dynamics} and \eqref{eq:NAS} to the pair of scalar equations
\begin{subequations}\label{eq:simulations-dynamics-dae}
	\begin{align}[left=\empheqlbrace]
		& \frac{d(x_2-x_1)}{dt} = 2 \mathfrak{a}_{s_2}\varepsilon^{3-2s_2}\frac{\int_{-\infty}^\infty y P_c^\varepsilon\bigl(\tfrac{U_c^\varepsilon}{V_c^\varepsilon}\bigr)^{2}dy}{\int_{-\infty}^\infty P_c^\varepsilon \tfrac{dU_c^\varepsilon}{dy}dy} G_D'(x_2-x_1),\\
		& \mu(S_c^\varepsilon) = \frac{\varepsilon^{1-2s_2}}{\mathfrak{a}_{s_2}}\bigl(R_D(0) + G_D(|x_2-x_1|)\bigr)S_c^\varepsilon.
	\end{align}
\end{subequations}
We remind the reader that the NAS (second equation) determines the common spike strength $S_c^\varepsilon$ for a given spike separation distance $|x_2-x_1|$. The common spike strength is then used to solve \eqref{eq:core_problem} for $U_c^\varepsilon$ and $V_c^\varepsilon$ as well as to solve for $P_c^\varepsilon$ in the adjoint problem of \S \ref{sec:slow-dynamics}. We implement this system numerically by pre-computing $S_c^\varepsilon$ as a function of $0\leq |x_2-x_1|\leq 2$ and then computing each of $U_c^\varepsilon$, $V_c^\varepsilon$, and $P_c^\varepsilon$ as functions of $0<S_c^\varepsilon<S_\star$. The differential algebraic system can then be easily solved with any standard ordinary differential equation library (we used \texttt{solve\_ivp} from the SciPy \texttt{integrate} library).

To validate our asymptotic theory we performed multiple numerical simulations of \eqref{eq:frac-gm-full-system} with an initial condition consisting of a symmetric two-spike solution constructed using the methods in \S \ref{sec:equilibrium} where the spikes are concentrated at $x_1=-0.2$ and $x_2=0.2$. For each of our simulations we set $\varepsilon=0.01$ and used values of $\tau = 0.1$  and $D=0.8 D_{0,\text{comp}}^\varepsilon \varepsilon^{2s_2-1}$ with which we can avoid Hopf bifurcations and competition instabilities (see Sections \ref{subsec:simulations-hopf} and \ref{subsec:simulations-competition}). These simulations were completed for the pairs $(s_1,s_2)=(0.4,0.35), (0.4,0.3), (0.46,0.3),$ and $(0.46,0.35)$ and the resulting spike trajectories are shown as solid blue lines in Figure \ref{fig:simulation-trajectories}. In each of these plots the trajectories predicted by solving \eqref{eq:simulations-dynamics-dae} are shown as dashed orange lines. We see that in each case the asymptotics provide good qualitative agreement of the spike trajectories. Finally we direct the reader to Figure \ref{fig:slow-dynamics-waterfall} where have plotted in more detail the time evolution of the activator to accompany Figure \ref{fig:slow-dyn-traj-example-000}.

\section{Rigorous Results for $s_2\approx1/2$}\label{sec:rigorous}
In this section we shall rigorously study the existence and stability of the ground state solution to the core problem
\begin{subequations}\label{6.1-core}
	\begin{align}[left=\empheqlbrace]
		& (-\Delta)^\frac12 U+U-V^{-1}U^2=0,\quad (-\Delta)^sV-U^2=0,  & -\infty<x<\infty,  \\
		& U,V>0, & -\infty<x<\infty,  \\
		&U,V\to0,&\mbox{as}\quad |x|\to+\infty.
	\end{align}
\end{subequations}
We proceed by first presenting in \S~\ref{sec6.1} several known results which will be used throughout this section. Then in \S~\ref{subsec:rigorous-existence} and \S~\ref{sec6.3} we provide the rigorous study on the existence and stability of the ground state solution respectively.

\subsection{Preliminaries}
\label{sec6.1}
\begin{lemma}
	\label{le2.green}
	Let $s<\frac12$ and $G(x)$ be the Green's function of the equation
	\begin{equation}
		\label{2.green}
		\s G(x)=\delta(x).
	\end{equation}
	Then
	$$G(x)=\dfrac{\Gamma(1-2s)\sin(s\pi)}{\pi}x^{2s-1}.$$
\end{lemma}

\begin{proof}
	Using the Fourier transform, we can write the \eqref{2.green} as
	\begin{equation}
		|\xi|^{2s}\hat G(\xi)=1.
	\end{equation}	
	Therefore, we have
	\begin{equation}
			G(x)=\frac{1}{\pi}\int_0^\infty\frac{\cos(x\xi)}{\xi^{2s}}d\xi
			=x^{2s-1}\frac{1}{\pi}\int_0^\infty\frac{\cos\xi}{\xi^{2s}}d\xi
			=\frac{\Gamma(1-2s)\sin(s\pi)}{\pi}x^{2s-1},
	\end{equation}
	where we used $\int_0^\infty\frac{\cos\xi}{\xi^{2s}}d\xi=\Gamma(1-2s)\sin(s\pi)$ if $2s<1.$
\end{proof}
\medskip

We introduce the transformation
{\begin{equation}
	\label{2.tau}
	U=\tau_sU,\quad V=\tau_sV,\quad \tau_s=\left(\frac{\Gamma(1-2s)\sin(s\pi)}{\pi}\int_{\R}w^2(y)dy\right)^{-1}=\frac{1}{2\Gamma(1-2s)\sin(s\pi)},
\end{equation}
where $w$ is the unique ground state solution to
\begin{equation}
	\label{2.ground}
	{(-\Delta)^{1/2}} w+w-w^2=0\quad \mbox{in}\quad \R,\qquad w(x)\to0\quad \mbox{as}\quad |x|\to\infty.
\end{equation}
In this case we can give the explicit form of $w$ and the integral of $w^2$ on the real line
\begin{equation}
	w(x)=\frac{2}{1+x^2}\quad\mathrm{and}\quad \int_{\R}w^2dx=2\pi.
\end{equation}
Based on \eqref{2.tau} we can write \eqref{6.1-core} as
\begin{subequations}\label{2.sys}
	\begin{align}[left=\empheqlbrace]
		& \sh U+U-V^{-1}U^2=0,\quad \s V-\tau_s U^2=0, & -\infty<x<\infty,\\
		& U,V>0, & -\infty<x<\infty,\\
		& U,V\to0, & \mbox{as}\quad |x|\to\infty.	
	\end{align}	
\end{subequations}
We look for a solution to \eqref{2.sys} in the form $U=w+\phi$ with $\phi$ being a lower order term. Denoting by $T(h)$ the unique solution of the equation
\begin{equation}
	\label{2.h}
	\s V=\tau_sh\quad \mbox{in}\quad \R,\qquad V(x)\to0\quad\mbox{as}\quad |x|\to\infty,
\end{equation}
for $h\in L^\infty(\R)$, then formally we have
\begin{equation}
	T(U^2)=T(w^2)+2T(w\phi)+h.o.t.,
\end{equation}
where $h.o.t.$ indicates the higher order terms. We denote $v_w=T(w^2)$ so that using the Green's function given in Lemma \ref{le2.green} we have
\begin{equation}
	\label{2.v}
	v_w=T(w^2)=\tau_s\int_{\R}w^2(y)G(x-y)dy.
\end{equation}
Expanding the Green function in the following way
\begin{equation}
	\label{2.exp-g}
	\begin{aligned}
		G(x)=~&\frac{\Gamma(1-2s)\sin(s\pi)}{\pi}|x|^{2s-1}=e^{(2s-1)\log |x|+\log\Gamma(1-2s)\sin(s\pi)/\pi}\\
		=~&\frac{\Gamma(1-2s)\sin(s\pi)}{\pi}\left(1+(2s-1)\log |x|+(2s-1)^2(\log |x|)^2/2+\cdots\right),
	\end{aligned}
\end{equation}
and then using the fact $\Gamma(1-2s)\sim (1-2s)^{-1}$ as $s\to\frac12$ we get that
\begin{equation}
	T(w^2)=\tau_s\frac{\Gamma(1-2s)\sin(s\pi)}{\pi}\int_{\R}w^2(y)dy+O(2s-1).
\end{equation}
As a consequence when $x$ is bounded we have
\begin{equation}
	\label{2.v'}
	v_w\equiv T(w^2)=1+O(2s-1)\quad \mbox{and}\quad  T(w\phi)=\frac{1}{\int_{\R}w^2}\int_{\R}w\phi dy+O(2s-1),
\end{equation}
then the nonlinear term of the first equation in \eqref{2.sys} can be written as
\begin{equation}
	\frac{U^2}{V}=\frac{w^2+2w\phi+h.o.t.}{v_w+2T(w\phi)+h.o.t.}
	=\frac{w^2}{v_w}+2w\phi-2\frac{\int_{\R}w\phi dy}{\int_{\R}w^2dy}w^2+h.o.t.+O(2s-1).
\end{equation}
Substituting it into the first equation of \eqref{2.sys} we get
\begin{equation}
	\label{2.lin}
	L(\phi)\equiv \sh \phi+(1-2w)\phi+2\frac{\int_{\R}w\phi dy}{\int_{\R}w^2 dy}w^2=S(w)+N(\phi),
\end{equation}
where
$$S(w)=-\sh w-w+\frac{w^2}{v_w},$$
and
$$N(\phi)=\frac{(w+\phi)^2}{T((w+\phi)^2)}-\frac{w^2}{v_w}-2w\phi+2\frac{\int_{\R}w\phi dy}{\int_{\R}w^2 dy}w^2,$$
and represents the higher order terms in $\phi.$

Concerning the ground state $w$ and the non-local linearized operator $L$, we have the following result
\begin{proposition}
	\label{pr2.1}
	Let $w$ be the unique, positive, radially symmetric solution to	\eqref{2.ground}.
	\begin{itemize}
		
		\item [(a)] Let $L_0=(-\Delta)^\frac12+(1-2w)id$. Then we have
		\begin{equation*}
			\mathrm{Ker}(L_0)=\mathrm{Span}\left\{\frac{d w}{dx}\right\}.	
		\end{equation*}
		
		\item [(b)] Let $L$ be the lineazied operator defined in \eqref{2.lin}
		and
		$$L^*\phi=\sh\phi+(1-2w)\phi+2\frac{\int_{\R}w^2\phi dx}{\int_{\R}w^2dx}w.$$
		Then
		\begin{equation}
			\label{2.kernel}
			\mathrm{Ker}(L)=\mathrm{Ker}(L^*)={\mathrm{Span}}\left\{\frac{d w}{d x}\right\}.
		\end{equation}
	\end{itemize}
\end{proposition}

\begin{proof}
	The proof of part (a) is given by \cite[Proposition 1.1 and Theorem 2.3]{frank_2013_uniqueness}. To prove part (b) we first notice that $L_0w=-w^2$. If $\phi\in\mbox{Ker}(L)$ then
	\begin{equation}
		L_0\phi=-c(\phi)w^2,\quad\mbox{where}\quad c(\phi)=2\frac{\int_{\R}w\phi dx}{\int_{\R}w^2dx}.
	\end{equation}
	Therefore, by conclusion (a) we get  $\phi - c(\phi) w \in\mbox{Ker}(L_0)$ and in particular
	$$
	\phi = \beta\frac{d w}{d x}  + c(\phi)w,
	$$
	for some constant $\beta$. 	As a consequence, we have
	\begin{equation*}
		c(\phi)=2c(\phi)\frac{\int_{\R}w^2dx}{\int_{\R}w^2dx}=2c(\phi),
	\end{equation*}
	which implies $c(\phi)=0$. Hence, $\phi\in\mbox{Ker}(L_0)$ and we get that $\phi\in \mathrm{Span}\left\{\frac{d w}{d x}\right\}$. Similarly, if $\phi\in \mathrm{Ker}(L^*)$ then
	\begin{equation*}
		L^*\phi=-c_1(\phi)w,\quad \mbox{where}\quad c_1(\phi)=2\frac{\int_{\R}w^2\phi dx}{\int_{\R}w^2dx}.
	\end{equation*}
	Using the fact
	$$L_0(w+x\cdot\partial_xw)=-w,$$
	we have
	\begin{equation}
		\phi-2\frac{\int_{\R}w^2\phi dx}{\int_{\R}w^2dx}(w+x\cdot\partial_xw)\in\mbox{Ker}(L_0).	
	\end{equation}
	Then
\begin{equation*} c_1(\phi)=2c_1(\phi)\frac{\int_{\R}(w+x\cdot\partial_xw)w^2dx}{\int_{\R}w^2dx}=2c_1(\phi)
=2c_1(\phi)\frac{\frac23\int_{\R}w^3dx}{\int_{\R}w^2dx}=2c_1(\phi),
\end{equation*}
	where we used
	$$\int_{\R}w^3dx=3\pi,\quad \int_{\R}w^2dx=2\pi.$$
	Thus $c_1(\phi)=0$ and $\phi\in \mathrm{Span}\left\{\frac{d w}{d x}\right\}$ which proves the third conclusion.
\end{proof}

In the end of this subsection, we provide the analysis of the linear operator $L$ in a framework of weighted $L^\infty$ spaces. For this purpose we consider the following norm for a function defined on $\R$. We define
\begin{equation}
	\label{2.weighted}
	\|\phi\|_*=\|\rho(x)^{-1}\phi\|_{L^\infty(\R)},\quad \mbox{where}\quad \rho(x)=\frac{1}{(1+|x|)^{\mu}},\qquad \frac12<\mu\leq2.
\end{equation}
Given a function $h$ with $\|h\|_*<\infty$, due to the fact that $\mathrm{Ker}(L)={ {Span}}\left\{\frac{d w}{d x}\right\}$, we need to study the related linear problem in the following form
\begin{equation}\label{2.lin-e}	
		L\phi=h+c\frac{dw}{dx}\quad -\infty<x<\infty,\qquad \phi(x)\to0\quad \mbox{as } |x|\to\infty,\qquad\langle \phi,\tfrac{d w}{d x}\rangle=0.
\end{equation}
Our aim is to find $(\phi,c)$ such that \eqref{2.lin-e} holds. Concerning \eqref{2.lin-e} we have the following existence result and a-priori estimate  {for which a proof can be found in \cite[Theorem 4.2]{wei_2019_multi_bump}.}

\begin{theorem}\label{th2.1}
	 If $h$ satisfies $\|h\|_*<\infty$ then problem \eqref{2.lin-e} has an unique solution $\phi=\mathcal{T}(h)$ and $c=c(h)$. Moreover there exists a constant $C>0$ such that for any such $h$
	\begin{equation}
		\label{2.estimate-a}
		\|\mathcal{T}h\|_*\leq C\|h\|_*.
	\end{equation}
\end{theorem}

\subsection{The rigorous proof of the existence results}\label{subsec:rigorous-existence}
In this section we shall give rigorous proof of Theorem \ref{th1.exist}.

\subsubsection{Error estimates}
We begin by studying $v_w(x)$ for which we prove improved estimates. By the definition \eqref{2.v} we see that
\begin{equation}
	\s v_w=\tau_sw^2,\quad v_w(x)\to0~\ \mbox{as}~\ |x|\to+\infty,
\end{equation}
where $\tau_s$ is given in \eqref{2.tau}. {We will consider $v_w(x)$ in the two disjoint regions $x\in I_s$ and $x\in\mathbb{R}\setminus I_s$ where we define the interval}
\begin{equation}
	\label{3.region}
	I_s\equiv \left[-100(1-2s)^{-1},~100(1-2s)^{-1}\right].
\end{equation}
Starting with $x\in I_s$ we use the Green representation formula {\eqref{2.v} and the asymptotics  \eqref{2.exp-g}} to get
\begin{equation}
	\begin{split}
		v_w(x)=&\tau_s\frac{\Gamma(1-2s)\sin(s\pi)}{\pi}\int_{\R}w^2(y)dy
+(2s-1)\tau_s\frac{\Gamma(1-2s)\sin(s\pi)}{\pi}\int_{\R}\log|x-y|w^2(y)dy\\
		&+(2s-1)^2\tau_s\frac{\Gamma(1-2s)\sin(s\pi)}{2\pi}\int_{\R}w^2(y)(\log|x-y|)^2dy+o((2s-1)^2)\\
		=&1+(2s-1)\frac{1}{\int_{\R}w^2(y)dy}\int_{\R}\log|x-y|w^2(y)dy\\
		&+\frac{(2s-1)^2}{2}\frac{1}{\int_{\R}w^2(y)dy}\int_{\R}(\log|x-y|)^2w^2(y)dy+o((2s-1)^2).
	\end{split}
\end{equation}
Next we define
\begin{equation}
	H_i(x)=\frac{\int_{\R}w^2(y)(\log|x-y|)^idy}{\int_{\R}w^2(y)dy},
\end{equation}
and readily deduce that $H_i(x)$ is even since $w(y)$ is even.
Furthermore, as $|x|$ is sufficiently large, by standard potential analysis we can write
\begin{equation}
	H_i(x)=(\log|x|)^i+f(x),
\end{equation}
where $f$ is an even function and itself with its first derivative are uniformly bounded.

While for $|x|\geq 100(1-2s)^{-1}$, using the potential analysis, we get
that
\begin{equation}
	\label{3.1.far}
	v_w(x)\geq c\tau_s|x|^{2s-1}\quad \mbox{for}\quad |x|\geq100(1-2s)^{-1}.
\end{equation}

Summarizing the above estimates, we have the following conclusion.

\begin{lemma}
	\label{le3.1}
	 {Letting} $v_w$ be defined as in \eqref{2.v}  {we have the following estimates}:
	\begin{itemize}
		\item [(a).]  {If} $x\in I_s$,  {then}
		\begin{equation}
			v_w(x)=1+(2s-1)H_1(x)+\frac{(2s-1)^2}{2}H_2(x)+o((1-2s)^3).
		\end{equation}
		
		\item [(b).]  {If} $x\in \R\setminus I_s$,  {then}
		\begin{equation}
			v_w(x)\geq c\tau_s|x|^{2s-1}
		\end{equation}
		{for some constant $c>0$.}
	\end{itemize}
\end{lemma}

We now focus on estimating the quantity {$S(w) = -\sh w - w + v_w^{-1}w^2$ which, using \eqref{2.ground}, can be rewritten as }
$$S(w)=\frac{w^2}{v_w}-w^2.$$
Let us first analyze the term $S(w)$ in the interval $I_s$ introduced in \eqref{3.region}. It is easy to see that in this region we have
\begin{equation*}
	v_w(x)=1+O\left((1-2s)^{1-\delta}\right),
\end{equation*}
where $\delta$ is any small positive number, {and therefore}
\begin{equation*}
	\frac{1-v_w}{v_w}w^2=O\left((1-2s)^{1-\delta}w^2\right).
\end{equation*}
{Writing
\begin{equation}
	\label{3.error}
	S(w)=\frac{1-v_w}{v_w}w^2.
\end{equation}
we deduce that for $x\in I_s$}
\begin{equation}
	|S(w)|=O\left((1-2s)^{1-\delta}{\rho(x)}\right)\quad \mbox{for}\quad |1-2s|\ll1.
\end{equation}

On the other hand, by Lemma \ref{le3.1} we find that for $x\in\mathbb{R}\setminus I_s$
\begin{equation}
	\label{3.error-1}
	\begin{aligned}
		|S(w)|\leq C(1-2s)^{-1}|x|^{1-2s}w^2=C(1-2s)^{-1}|x|^{-2}\rho(x)
		=O(1-2s)\rho(x).
	\end{aligned}
\end{equation}
In conclusion, we have
\begin{lemma}
	\label{le3.2}
	Let $\mu=2$ in the definition of $\|\cdot\|_*$. If $1-2s$ is sufficiently small then we have
	$$\|S(w)\|_*\leq C(1-2s)^{1-\delta},$$	
	where $C$ is some constant independent of $\e$ and $\delta$ is any small positive number independent of $\e$.
\end{lemma}

\subsubsection{The existence of solution}
Recall that the original problem was cast in the form
\begin{equation}
	\label{3.2.1}
	\sh U+U-\frac{U^2}{T(V^2)}=0.
\end{equation}
Rather than solving \eqref{3.2.1} directly we consider instead the problem of finding $A$  satisfying
\begin{equation}
	\sh A+A-\frac{A^2}{T(A^2)}=c\frac{d w}{d x},
\end{equation}
for a certain constant $c$, and such that $\langle A-w, Z\rangle=0$. Rewriting $A=w+\phi$ we get that this problem is equivalent to
\begin{equation}
	\begin{aligned}
		&\sh\phi+\phi-2w\phi+2w^2\frac{\int_{\R}w\phi dx}{\int_{\R}w^2dx}\\	
		&=-\sh w-w+\frac{w^2}{v_w}+\frac{(w+\phi)^2}{T((w+\phi)^2)}-\frac{w^2}{v_w}-2w\phi+2w^2\frac{\int_{\R}w\phi dx}{\int_{\R}w^2dx}+c\frac{d w}{dx}\\	
		&=S(w)+N(\phi)+c \frac{dw}{dx}
	\end{aligned}
\end{equation}
and
\begin{equation}
	N(\phi)=\frac{(w+\phi)^2}{T((w+\phi)^2)}-\frac{w^2}{v_w}-2w\phi
	+2w^2\frac{\int_{\mathbb{R}}w\phi dx}{\int_{\mathbb{R}}w^2dx}.
\end{equation}
Using the operator $\mathcal{T}$ introduced in Theorem \ref{th2.1}, we see that the problem is equivalent to finding a $\phi\in\mathcal{H}$ so that
$$
{\phi=Q(\phi)\equiv\mathcal{T}(S(w)+N(\phi))}.
$$
We shall show that this fixed point problem has a unique solution in the region of the form
\begin{equation}
	\label{3.setD}
	\mathcal{D}=\left\{\phi\in\mathcal{H}\mid \|\phi\|_*\leq C(1-s)^{1-\delta}\right\},
\end{equation}
for any small positive constant $\delta$, provided that $1-2s$ is sufficiently small. Here
\begin{equation}
	\mathcal{H}=\left\{\phi\in L^\infty\bigr| \big\langle\phi,\tfrac{d w}{d x}\big\rangle=0\right\}.	
\end{equation}

We have already proved that $\|S(w)\|_*\leq C(1-2s)^{1-\delta}$. In the following lemma we estimate the higher order error term $N(\phi)$.

\begin{lemma}
	\label{le3.order}
	Assume that $\phi\in\mathcal{D}$, then for $1-2s$ sufficiently small, we have
	\begin{equation}
		\|N(\phi)\|_*\leq C(\|\phi\|_*+\sigma(1-2s))\|\phi\|_*,
	\end{equation}
	where $\sigma(1-2s)\leq C(1-2s)^{1-\delta}$ as $1-2s\to0$.
\end{lemma}

\begin{proof}
	Let us assume first $x\in\R\setminus I_s$. In this region we have $w(x)\leq C\rho(x) $. Combined with the standard potential analysis one can show that
	\begin{align*}
		& T((w+\phi)^2)\geq C(1-2s)|x|^{2s-1},\\
		& T(w\phi)\leq C(1-2s)|x|^{2s-1}\|\phi\|_*, \\
		&T(\phi^2)\leq C(1-2s)^{2-\delta}|x|^{2s-1}\|\phi\|_*.
	\end{align*}
	As a consequence,
	\begin{equation*}
		\begin{aligned}
			|N(\phi)|\leq&~\left(\frac{2wv_w\phi+v_w\phi^2-2w^2T(w\phi)-w^2T(\phi^2)
			}{v_wT((w+\phi)^2)}-2w\phi+2w^2\frac{\int_{\R}w\phi dx}{\int_{\R}w^2dx}\right)\\
			\leq&~C\left(\frac{\rho(x)}{(1-2s)(1+|x|)^{2s+1}}+
			\frac{\rho(x)}{(1-2s)(1+|x|)^{2s+1}}\|\phi\|_*\right)\|\phi\|_*+C\rho(x)^2\|\phi\|_*.
		\end{aligned}	
	\end{equation*}
	Therefore we have
	\begin{equation}
		\label{3.n-1}
		|\rho^{-1}N(\phi)|\leq C(\|\phi\|_*+(1-2s)^{6s-2})\|\phi\|_*,
	\end{equation}	
	provided $s\to\frac12$.
	
	Considering next the case $x\in I_s$ we decompose $N(\phi)$ in the form
	$$N(\phi)=N_1(\phi)+N_2(\phi),$$
	where
	\begin{align*}
		N_1(\phi)=(w+\phi)^2\Big[\frac{1}{T((w+\phi)^2)}-\frac{1}{v_w}+\frac{2T(w\phi)}{v_w^2}\Big]
		-(2w+\phi)\phi\frac{2T(w\phi)}{V^2}
	\end{align*}
	and
	\begin{align*}
		N_2(\phi)=-2\phi w\left(1-\frac{1}{v_w}\right)+2U^2\left(\frac{\int_{\R}w\phi dx}{\int_{\R}w^2dx}-\frac{T(w\phi)}{v_w^2}\right)+\frac{\phi^2}{v_w}.
	\end{align*}
	It is known that
	\begin{align*}
		v_w(x)=1+O((1-2s)^{1-\delta})
	\end{align*}
	and
	\begin{align*}
		T(w\phi)=\frac{\int_{\R}w\phi dx}{\int_{\R}w^2dx}+O((1-2s)^{1-\delta}).
	\end{align*}
	and in particular $|T(w\phi)|=O(\|\phi\|_*)$. Likewise, $T(\phi^2)=O(\|\phi\|_*^2)$. Combining these facts we obtain
	\begin{align*}
		|N_1(\phi)|\leq C(w+\phi)^2T(\phi^2)+C\left(2w\phi+\phi^2\right)T(w\phi)
		\leq C\rho(x)\|\phi\|_*^2.
	\end{align*}
	A similar analysis yields
	\begin{align*}
		|N_2(\phi)|\leq C(1-2s)^{1-\delta}\left(|\phi|w+\rho^2\|\phi\|_*\right)+C|\phi|^2,
	\end{align*}
	and therefore
	\begin{align*}
		\|N(\phi)\|_*\leq C(\|\phi\|_*^2+\sigma(1-2s)\|\phi\|_*)
	\end{align*}
	 for $x\in I_s$. Together with \eqref{3.n-1} this proves the lemma.
\end{proof}

With Lemma \ref{le3.order} we are able to give the proof of Theorem \ref{th1.exist}.

\begin{proof}[Proof of Theorem \ref{th1.exist}.]
	Using the definition of the corresponding norms, repeating almost the same arguments as Lemma \ref{le3.order} one can prove that if  {$\|\phi_i\|_*\leq C(1-2s)^{1-\delta}$ for $i=1,2$, }
	then, given any small $\kappa\in(0,1)$, we have the following inequality
	\begin{equation}
		\|N(\phi_1)-N(\phi_2)\|_*\leq \kappa\|\phi_1-\phi_2\|_*,
	\end{equation}
	provided $1-2s$ is sufficiently small. As a consequence, we get that the operator $Q$ is a contraction mapping in the set $\mathcal{D}$ defined in \eqref{3.setD}. On the other hand, we also get from Lemma \ref{le3.order} that $Q$ maps $\mathcal{D}$ into itself. Thus, by using the Banach fixed point theorem, we get the existence of a unique fixed point of $Q$ in $\mathcal{D},$ that is,
	\begin{equation}
		\label{3.f-lin}
		\sh \phi+\phi-2w\phi+2w^2\frac{\int_{\R}w\phi dx}{\int_{\R}w^2dx}=S(w)+N(\phi)+c\frac{d w}{d x}.
	\end{equation}
	Next, we notice that $S(w)$ is an even function and the linearized problem can be solved in the even symmetric function class. Without loss of generality, we can pose the further restriction on the set $\mathcal{H}$ that all the perturbations $\phi$ are even symmetric functions. As a consequence, we see that apart for the term $\frac{d w}{d x}$ all the remaining terms are even symmetric and this implies that $c=0.$ Hence $w+\phi$ is a solution to the original Gierer-Meinhardt system \eqref{6.1-core}.
\end{proof}

\subsection{Stability Analysis: large and small eigenvalues}
\label{sec6.3}

{In this section we characterize the linear stability of the ground state solution constructed in \S \ref{subsec:rigorous-existence} above by considering both large and small eigenvalues.}

\subsubsection{Large eigenvalue}
{Linearizing \eqref{6.1-core} about the equilibrium solution $(u,v)$ we obtain the following eigenvalue problem}
\begin{subequations}
	\label{4.eig}
	\begin{align}[left=\empheqlbrace]
		& \sh\phi+\phi-2V^{-1}U\phi+V^{-2}U^2\psi+\la_s\phi=0, & -\infty<x<\infty,\\
		& \s\psi-2U\phi+\tau\la_s\psi=0, & -\infty<x<\infty,
	\end{align}
\end{subequations}
where $\la_s\in\mathbb{C}$, $\phi\in H^{1}(\R)$, and $\psi\in H^{2s}(\R)$. Let
$$\hat U=\tau_s^{-1}U,\quad \hat V=\tau_s^{-1}V.$$
Then \eqref{4.eig} can be rewritten as
\begin{subequations}
	\label{4.eig-1}
	\begin{align}[left=\empheqlbrace]
		& \sh\phi+\phi-2\hat{V}^{-1}\hat{U}\phi+ \hat{V}^{-2}\hat{U}^2\psi+\la_s\phi=0, & -\infty<x<\infty, \label{4.eig-1a}\\
		& \s\psi-2\tau_s\hat U\phi+\tau\la_s\psi=0, & -\infty<x<\infty.\label{4.eig-1b}
	\end{align}
\end{subequations}
Our aim is to study the large eigenvalues, i.e.\@ those for which we may assume that there exists $c>0$ such that $|\la_s|\geq c>0$ for $1-2s$ is small. If $\Re(\la_s)<-c$ then we are done and we therefore may assume that $\Re(\la_s)\geq-c$. For a subsequence $1-2s\to0$ and $\la_s\to\la_0$ we shall derive a limiting NLEP satisfied by $\la_0$.

To simplify our argument, we shall assume $\tau=0$ and the general case can be proved by a perturbation argument.  When $x\in I_s$, we calculate
\begin{equation}
	\label{4.psi}
		\psi(x)=2\tau_s\int_{\R} {G(x-y)\hat U(y)\phi(y)dy}
		=2\frac{\int_{\R}w\phi dy}{\int_{\R}w^2dy}+O((1-2s)^{1-\delta})\|\phi\|_{H^1(\R)}.
\end{equation}
Substituting this into \eqref{4.eig-1a}, and letting $2s-1\rightarrow 0$, we derive the following nonlocal eigenvalue problem
\begin{equation}
	\label{4.non}
	\sh \phi+\phi-2w\phi+2\frac{\int_{\R}w\phi dx}{\int_{\R}w^2dx}w^2+\la_0\phi=0.
\end{equation}
By Theorem 3.2 in \cite{gomez_2022} we see that $\lambda_0<0$, which implies that the large eigenvalues are stable.

\subsubsection{Small eigenvalue}
{We next consider the} small eigenvalues of \eqref{4.eig-1}, i.e. those for which $\la_s\to0$ as $s\to\frac12$. In last section, we have already shown the existence of solutions $(\hat U,\hat V)$ to \eqref{2.sys}. We notice that this equation is translation invariant. By differentiating \eqref{2.sys} we derive that
\begin{subequations}
	\begin{align}[left=\empheqlbrace]
		& \sh \frac{d \hat U}{d x}+ \frac{d \hat U}{d x}
		-2\frac{\hat U}{\hat V}\frac{d \hat U}{d x}+\frac{\hat U^2}{\hat V^2}\frac{d \hat V}{d x}=0, & -\infty<x<\infty,\\
		& \s\frac{d\hat V}{d x}-2\tau_s\hat U\frac{d \hat U}{dx}=0, & -\infty<x<\infty.
	\end{align}
\end{subequations}
This suggests that $(\phi,\psi)$ of \eqref{4.eig-1} can be written as
\begin{equation}
	\phi=a\frac{d\hat U}{d x}+\phi^\perp, \quad\mbox{and}\quad \psi=a\frac{d\hat V}{d x}+\psi^\perp,
\end{equation}
where $\phi^\perp\perp\frac{d\hat U}{d x}$ and  $\psi^\perp$ satisfy
\begin{subequations}
	\label{4.small-1}
	\begin{align}[left=\empheqlbrace]
		& \sh\phi^\perp+\phi^\perp-2\hat{V}^{-1}\hat{U}\phi^\perp+\hat{V}^{-2}\hat{U}^2\psi^\perp+\la_s\frac{d\hu}{dx}+\la_s\phi^\perp=0, & -\infty<x<\infty, \label{4.small-1a}\\
		& \s\psi^\perp-2\tau_s\hat U\psi^\perp=0, & -\infty<x<\infty.
	\end{align}
\end{subequations}
As $s\to\frac12$, we know that
\begin{equation*}
	\frac{\hu}{\hv}\to w\quad \mbox{and}\quad \frac{\hu^2}{\hv^2}\psi^\perp
	\to 2\frac{\int_{\R}w\phi^\perp dy}{\int_{\R}w^2dy}w^2.
\end{equation*}
Multiplying \eqref{4.small-1a} by $\phi^\perp$ we have
\begin{equation}
	\label{4.small-2}
	\la_s\int_{\R}|\phi^\perp|^2dx
	=-\int_{\R}\left(\sh\phi^\perp+\phi^\perp-2\frac{\hat U}{\hat V}\phi^\perp+\frac{\hu^2}{\hv^2}\psi^\perp\right)\phi^\perp dx.
\end{equation}
From Lemma A.2 in \cite{gomez_2022} we have that
\begin{equation*}
	\begin{aligned}
		L_1(\phi^\perp,\phi^\perp)
		=~&\int_{\R}\left(|(-\Delta)^{\frac{1}{4}}\phi^\perp|^2
		+|\phi^\perp|^2-2w|\phi^\perp|^2+2\frac{\int_{\R}w\phi^\perp dx\int_{\R}w^2\phi^\perp dx}{\int_{\R}w^2dx}\right)\\
		\geq ~&\frac{\int_{\R}w^3dx\left(\int_{\R}w\phi^\perp dx\right)^2}{\left(\int_{\R}w^2dx\right)^2}+a_1 {\inf_{\psi\in X_1}\|\phi^\perp - \psi\|_{L^2(\mathbb{R})}},
	\end{aligned}
\end{equation*}
where $a_1>0$ and
$$X_1=\mbox{Span}\left\{w,\frac{dw}{d x}\right\}.$$
Since $\phi^\perp\perp\frac{d\hu}{d x}$ and $\hat U$ is well approximated by $w$, we get from \eqref{4.small-2} that
\begin{equation}
	\la_s\int_{\R}|\phi^\perp|^2dx\leq 0.
\end{equation}
Hence, we have shown all the small eigenvalues are stable. Thus, Theorem \ref{th1.stable} follows by combining the conclusions of the last two sections.

\section{Discussion}\label{sec:discussion}

In this paper we have used formal asymptotic methods to study the existence and linear stability of localized solutions for the fractional Gierer-Meinhardt system where the fractional order of the inhibitor is $s_2\in(0,1/2)$. These results extend those previously obtained in \cite{gomez_2022} and \cite{medeiros_2022} for $s_2\in(1/2,1)$ and $s_2=1/2$ respectively. Using the method of matched asymptotic expansions the construction of localized solutions was reduced to solving a system of nonlinear algebraic equations while the study of their linear stability was reduced to analyzing a globally coupled eigenvalue problem. We found that when $D=O(\varepsilon^{2	s_1-1})$ both symmetric and asymmetric multi-spike solutions can be constructed though the latter were found to always be linearly unstable. On the other hand symmetric spikes were found to have stability regions outside of which they may undergo either a competition instability or a Hopf bifurcation. Using a leading order theory we found that the competition instability threshold is monotone decreasing in $s_1$ and it is either monotone decreasing in $s_2$ when $s_1>0.5$ or non-monotonic (first increasing and then decreasing) when $s_1<0.5$. In addition we found that the Hopf bifurcation threshold increases with $1/4<s_1<1$ provided $s_2$ and $\kappa$ are large enough, whereas it decreases with $0<s_2<1/2$ for all values of $s_1$ and $\kappa$. We also computed higher-order stability thresholds for specific cases of one- and two-spike solutions and these were supported by full numerical simulations of the system \eqref{eq:frac-gm-full-system}. Finally, in addition to the linear stability over an $O(1)$ timescale we also determined that spike solutions may be susceptible to drift instabilities leading to mutual repulsion between spikes, though these arise over a much slower $O(\varepsilon^{3-2s_2})$ timescale.

A key component in the formal construction of multi-spike solutions is the core problem \eqref{eq:core_problem} which was considered in detail numerically in \S\ref{subsec:core-problem} for general $s_2\in(0,1/2)$ and rigorously in \S\ref{sec:rigorous} for $s_2\approx 1/2$. We found that the behaviour of the far-field constant $\mu(S)$ shares some properties with its counterpart in the three-dimensional Gierer-Meinhardt system previously studied in \cite{gomez_2021}. In particular we used numerical continuation to deduce the existence of a value $S=S_\star$ for which the core problem admits a ground state solution (i.e. one for which $\mu(S_\star)=0$). The existence and linear stability of such a ground state was then rigorously established in \S\ref{sec:rigorous} for $s_2\approx 1/2$.

Finally, throughout our paper we have highlighted the similarities between both the analysis and structure of localized solutions for the one-dimensional fractional Gierer-Meinhardt system when $s_2\in(0,1/2)$ and the corresponding localized solutions in the three-dimensional Gierer-Meinhardt system \cite{gomez_2021}. This connection is a result of the leading order algebraic singularity of the Green's function which in particular fixes the far-field behaviour of solutions to the core problem \eqref{eq:core_problem} and also plays a key role in the asymptotic matching. In Appendix \ref{app:greens-function} we provide an expression for the Green's function which makes explicit its singular behaviour, showing in particular that the singular behaviour consists of multiple algebraic singularities when $s_2\in(0,1/2)\setminus\{\tfrac{1}{2r}\,|\,r\in\mathbb{Z},r\geq 1\}$ (see Proposition \ref{prop:greens}) as well as logarithmic singularities for $s_2=\tfrac{1}{2r}$ for $r\in\mathbb{Z}$ with $r\geq 1$ (see Proposition \ref{prop:greens-log}) We believe that these expressions for the Green's function will be particularly useful for future studies of localized solutions in one-dimensional fractional reaction-diffusion systems.

We conclude by highlighting some outstanding problems and suggestions for future research. One of the first outstanding problems is to derive a higher-order asymptotic theory in the case when $s_2=\tfrac{1}{2r}$ for $r=1,2,...$. The key hurdle in this direction is the emergence of both logarithmic and algebraic singularities in the Green's functions and we believe that a resolution of this would spark some interesting mathematics. Additionally, it would be interesting to provide a rigorous justification for the existence and linear stability results we have formally derived for general values of $0<s_2<1/2$. Extensions of the current model to incorporate non-periodic boundary conditions as well as different reaction-kinetics would also be an interesting direction for future research. Moreover the consideration of such fractional problems in two- and three-dimensional domains will also lead to interesting mathematical questions.

\section*{Acknowledgement}

D.\@ Gomez is supported by NSERC and the Simons Foundation, M.\@ Medeiros is partially supported by NSERC, J.\@ Wei is partially supported by NSERC, and W.\@ Yang is partially supported by NSFC No.11801550 and 1187147.

\section*{Conflicts of Interest}

The authors don't have any financial or non-financial conflicts of interest to disclose in relation to the contents of this paper.

\section*{Data Availability}

The data generated during and/or analysed during the current study is available from the corresponding author upon a reasonable request.

\addcontentsline{toc}{section}{References}
\bibliographystyle{abbrv}
\bibliography{biblio}

\appendix

\section{A Rapidly Converging Series for the Green's Function}\label{app:greens-function}

In this appendix we derive a rapidly converging series expansion of the periodic fractional Green's function satisfying \eqref{eq:greens-equation}. We begin by formally computing the Fourier series
\begin{equation}\label{eq:greens-series}
	G_D(x) = \frac{1}{2}D + \sum_{n=1}^\infty\biggl(1+ \frac{1}{D(\pi n)^{2s_2}}\biggr)^{-1} \frac{\cos \pi n x}{(\pi n)^{2s_2}},
\end{equation}
In contrast to the classical Green's function, the fractional Green's function contains multiple singular terms. By identifying and removing these singular terms from the series expansion \eqref{eq:greens-series} we can therefore obtain a rapidly converging series expansion.

We will first restrict $s_2$ to be such that $s_2\neq \tfrac{1}{2r}$ for any integers $r\geq 1$, returning to the remaining cases at the end of this appendix. We first note that for any $l>1$ and sufficiently large values of $n>0$ such that $D(\pi n)^{2s_2}>1$ we have
\begin{equation*}
	\bigl(1 + \tfrac{1}{D(\pi n)^{2s_2}}\bigr)^{-1} = \bigl(1 + \tfrac{1}{D(\pi n)^{2s_2}}\bigr)^{-1}\bigl(\tfrac{-1}{D(\pi n)^{2s_2}}\bigr)^l - D(\pi n)^{2s_2} \sum_{k=1}^{l}\bigl(\tfrac{-1}{D(\pi n)^{2s_2}}\bigr)^k.
\end{equation*}
Substituting into \eqref{eq:greens-series} then gives
\begin{equation}\label{eq:greens-series-temp-0}
	G_D(x) = \tfrac{D}{2} + \tfrac{(-1)^{k_{\max}}}{D^{k_{\max}}}\sum_{n=1}^\infty\bigl(1+ \tfrac{1}{D(\pi n)^{2s_2}}\bigr)^{-1}\tfrac{\cos \pi n x}{(\pi n)^{2(1+k_{\max})s_2}} - D\sum_{k=1}^{k_{\max}}\tfrac{(-1)^k}{D^k}\sum_{n=1}^\infty \tfrac{\cos\pi n x}{(\pi n)^{2ks_2}},
\end{equation}
where we choose $k_{\max}$ to be the smallest positive integer such that $2(1+k_{\max})s_2>1$, i.e.\@  $k_{\max}=\lceil\tfrac{1}{2s_2}-1\rceil$. This choice of $k_{\max}$ guarantees the second term in \eqref{eq:greens-series-temp-0} converges. We remark that other choices are also possible. For example, if second order derivatives at $x=0$ are needed then it will be more convenient to choose $k_{\max}$ to be the smallest integer such that $k_{\max}>\frac{3}{2s_2}-1$.

To determine the singular terms from the remaining sum in \eqref{eq:greens-series-temp-0} we first let $\beta>0$ and consider the series expansion
\begin{equation}
	|x|^{\beta - 1} = \tfrac{1}{\beta} + 2\sum_{n=1}^{\infty}c_{\beta,n}\tfrac{\cos \pi n x}{(\pi n)^{\beta}},\qquad c_{\beta,n} =\int_0^{\pi n}x^{\beta-1}\cos x dx \label{eq:power-x-series}
\end{equation}
Repeated integration by parts yields the identities
\begin{equation*}
	\begin{cases}
		c_{\beta,n} = (-1)^n(\beta-1)(n\pi)^{\beta-2} - (\beta-1)(\beta-2)c_{\beta-2,n}, & \beta>2, \\
		c_{\beta,n} = -(\beta-1)\int_0^{n\pi}x^{\beta-2}\sin x dx, & \beta > 0.
	\end{cases}
\end{equation*}
Letting $q\geq 1$ be the largest positive integer such that $\beta - 2(q-1) > 0$ the above identities give
\begin{equation}\label{eq:c-temp-01}
	c_{\beta,n} = (-1)^{n+1}\sum_{r=1}^{q-1}(-1)^{r}(\pi n)^{\beta-2r}\prod_{l=1}^{2r-1}(\beta-l) + (-1)^{q}\prod_{l=1}^{2q-1}(\beta-l)\int_0^{n\pi}x^{\beta-2q}\sin x dx.
\end{equation}
Note that if $\beta>0$ is an integer then \eqref{eq:c-temp-01} yields an explicit expression for $c_{\beta,n}$ Indeed if $\beta$ is an odd integer then $\beta-2(q-1) = 1$ so that the product in the second term of \eqref{eq:c-temp-01} vanishes whereas if $\beta$ is an even integer then $\beta - 2(q-1) = 2$ and the integral in the second term of \eqref{eq:c-temp-01} evaluates to $1-(-1)^n$. In particular when $\beta=3$ we obtain the useful series
\begin{equation}\label{eq:power-x-series-2}
	|x|^2 = \frac{1}{3} + 4\sum_{n=1}^\infty (-1)^n\frac{\cos \pi n x}{(\pi n)^2}.
\end{equation}
If instead $\beta>0$ is not an integer then we first write
\begin{equation*}
	\int_0^\infty x^{\beta - 2q}\sin x dx = \int_0^\infty x^{\beta - 2q}\sin x dx - \int_{n\pi}^\infty x^{\beta - 2q}\sin x dx.
\end{equation*}
Since in this case $-1 < \beta - 2q + 1 < 1$ we can use standard properties of the Gamma function (see equations (5.9.7) and (5.5.3) in \cite{NIST:DLMF}) to write
\begin{equation*}
	\int_0^\infty x^{\beta - 2q}\sin x dx = \frac{\pi}{2\Gamma(2q-\beta)\cos(\pi(\beta-2q+1)/2)} = \frac{(-1)^q}{2\prod_{l=1}^{2q-1}(\beta-l)}\mathfrak{a}_{\beta/2}^{-1},
\end{equation*}
where $\mathfrak{a}_{\beta/2} = -\pi^{-1}\beta\Gamma(-\beta)\sin(\pi\beta/2)$. Moreover since $\beta - 2q < 0$ we can integrate by parts to get
\begin{equation*}
	\int_{n\pi}^\infty x^{\beta-2q}\sin x dx = (-1)^n(n\pi)^{\beta-2q} + (\beta-2q)\int_{n\pi}^{\infty}x^{\beta-2q-1}\cos xdx,
\end{equation*}
where we remark that the last term is $O(n^{\beta-2q-2})$. In summary, for non-integer values of $\beta>0$ we have the expression
\begin{subequations}
	\begin{equation}\label{eq:c-temp-0}
		c_{\beta,n} = \frac{1}{2\mathfrak{a}_{\beta/2}} + (-1)^n(\beta-1)(\pi n)^{\beta-2} + a_{\beta,n},
	\end{equation}
	where
	\begin{equation}\label{eq:a-temp-0}
		a_{\beta,n} \equiv (-1)^{n+1}\sum_{r=2}^{q}(-1)^r(\pi n)^{\beta-2r}\prod_{l=1}^{2r-1}(\beta-l) + (-1)^{q+1}\prod_{l=1}^{2q}(\beta-l)\int_{n\pi}^\infty x^{\beta-2q-1}\cos x dx.
	\end{equation}
\end{subequations}

The constant term appearing in \eqref{eq:c-temp-0} allows us to relate $|x|^{\beta-1}$ with the series \eqref{eq:greens-series-temp-0} whereas the decay $a_{\beta,n}=O(n^{\beta-2q-2})$ from \eqref{eq:a-temp-0} yields a quickly converging series. Specifically, since our choice of $k_{\max}=\lceil\tfrac{1}{2s_2}-1\rceil$ implies that $\beta \leq 1-2s_2$ we deduce that $\beta = 2ks_2\leq 1-2s_2$ is not an integer for all $1\leq k\leq k_{\max}$. Therefore for any $0<s_2<1/2$ we can rewrite the summands in the rightmost term of \eqref{eq:greens-series-temp-0} as
\begin{equation*}
	\sum_{n=1}^\infty\frac{\cos \pi n x}{(\pi n)^{2ks_2}} = \mathfrak{a}_{ks_2}\biggl( |x|^{2ks_2-1}-\frac{1}{2ks_2} - \frac{2ks_2-1}{2}\biggl(x^2 - \frac{1}{3} \biggr) - 2\sum_{n=1}^\infty a_{2ks_2,n}\frac{\cos\pi n x}{(\pi n)^{2ks_2}} \biggr).
\end{equation*}
Substituting this back into \eqref{eq:greens-series-temp-0} we reach our final result.

\begin{proposition}\label{prop:greens}
	Let $s_2\in(0,\tfrac{1}{2})\setminus\{ \tfrac{1}{2r} \,|\, r\in\mathbb{Z},\,r\geq 1\}$. Then the Green's function satisfying \eqref{eq:greens-equation} is given by
	\begin{subequations}\label{eq:greens-series-rapid-full}
		\begin{equation}
			G_D(x) = \sum_{k=1}^{k_{\max}}\frac{(-1)^{k-1}\mathfrak{a}_{ks_2}}{D^{k-1}}|x|^{2ks_2-1} + R_D(x)
		\end{equation}
		where  $k_{\max} = \lceil \tfrac{1}{2s_2}-1 \rceil$ and the \textit{regular part} $R_D(x)$ is given by
		\begin{equation}\label{eq:greens-series-rapid-R_D}
			\begin{split}
				R_D(x) = & \frac{1}{2}D + \frac{(-1)^{k_{\max}}}{D^{k_{\max}}}\sum_{n=1}^\infty \biggl(1+ \frac{1}{D(\pi n)^{2s_2}}\biggr)^{-1} \frac{\cos\pi n x}{(\pi n)^{2(1+{k_{\max}})s_2}} +  \frac{1}{2s_2}\sum_{k=1}^{k_{\max}}\frac{(-1)^{k}\mathfrak{a}_{ks_2}}{kD^{k-1}} \\
				&   + \sum_{k=1}^{k_{\max}}\frac{(-1)^k (2ks_2-1)\mathfrak{a}_{ks_2}}{2 D^{k-1}}\biggl(x^2-\frac{1}{3}\biggr) + 2\sum_{n=1}^\infty \biggl(\sum_{k=1}^{k_{\max}}\frac{(-1)^k\mathfrak{a}_{ks_2} a_{2ks_2,n}}{D^{k-1}(\pi n)^{2ks_2}}\biggr)\cos \pi n x,
			\end{split}
		\end{equation}
		where $a_{2ks_2,n}$ is given by \eqref{eq:a-temp-0} with $\beta=2ks_2$ and $q=\lceil ks_2\rceil$, while $\mathfrak{a}_{ks_2}$ is given by
		\begin{equation}
			\mathfrak{a}_{ks_2} = -\frac{2ks_2}{\pi}\Gamma(-2ks_2)\sin(\pi ks_2).
		\end{equation}
	\end{subequations}
\end{proposition}

We conclude this appendix by revisiting our restriction that $s_2\neq \tfrac{1}{2r}$ for all integers $r\geq 1$. If instead $s_2=\tfrac{1}{2r}$ for some integer $r\geq 1$, then $k_{\max} = r$ and we see that $\mathfrak{a}_{k_{\max}s_2} = -2k_{\max}s_2\pi\Gamma(-1)\sin(\pi/2)$ is undefined. This extends beyond a technical difficulty and is closely tied to the singular behaviour of the Green's function \eqref{eq:greens-series-rapid}. Indeed when $s_2>1/2$ the Green's function has no singular terms, but if $1/4<s_2<1/2$ then it has one singular term, when $1/6<s_2<1/2$ it has two, and so on. At the transition points, i.e. $s_2=2^{-1},4^{-1},6^{-1}$, etc.\@,  the Green's function has an additional logarithmic singularity. One way to see this formally is to observe that if $s_2\rightarrow \tfrac{1}{2r}$ then $|x|^{2rs_2-1}\sim (2rs_2-1)\log|x|$ whereas $\mathfrak{a}_{rs_2}\sim\pi^{-1}(1-2rs_2)^{-1}$ so that the $k=k_{\max}=r$ term in \eqref{eq:greens-series-rapid-full} behaves like $(-1)^{r}D^{1-r}\pi^{-1}\log|x|$. To make this more precise we can calculate the Fourier series of $\log|x|$ and integrate by parts to get
\begin{equation}\label{eq:log-series}
	\sum_{n=1}^\infty \frac{\cos n\pi x}{n\pi} = -\frac{1}{\pi}\log|x| -\frac{1}{\pi} - \frac{2}{\pi}\sum_{n=1}^{\infty}\biggl(\text{Si}(n\pi) - \frac{\pi}{2}\biggr)\frac{\cos n\pi x}{n\pi}
\end{equation}
where $\text{Si}(z)\equiv \int_0^z\tfrac{\sin t}{t}dt$ is the Sine integral. Note that the rightmost term converges for all $-1<x<1$ since $\text{Si}(z)\sim \pi/2 + O(z^{-1})$ as $z\rightarrow+\infty$. Comparing \eqref{eq:log-series} with the rightmost term in \eqref{eq:greens-series-temp-0} we readily deduce our next result.
\begin{proposition}\label{prop:greens-log}
	Let $s_2=\tfrac{1}{2r}$ for some integer $r\geq 1$. Then the Green's function satisfying \eqref{eq:greens-equation} is given by
	\begin{subequations}\label{eq:greens-series-rapid-full-log}
		\begin{equation}
			G_D(x) = \sum_{k=1}^{r-1}\frac{(-1)^{k-1}\mathfrak{a}_{ks_2}}{D^{k-1}}|x|^{2ks_2-1} + \frac{(-1)^r}{\pi D^{r-1}}\log|x| + R_D(x)
		\end{equation}
		where $R_D(x)$ is given by
		\begin{equation}\label{eq:greens-series-rapid-R_D-log}
			\begin{split}
				R_D(x) = & \frac{1}{2}D  + \frac{(-1)^{r}}{D^{r}}\sum_{n=1}^\infty \biggl(1+ \frac{1}{D(\pi n)^{1/r}}\biggr)^{-1} \frac{\cos\pi n x}{(\pi n)^{1+1/r}} +  r\sum_{k=1}^{r-1}\frac{(-1)^{k}\mathfrak{a}_{k/(2r)}}{kD^{k-1}} \\
				&   + \sum_{k=1}^{r-1}\frac{(-1)^k (k/r-1)\mathfrak{a}_{k/(2r)}}{2 D^{k-1}}\biggl(x^2-\frac{1}{3}\biggr) + 2\sum_{n=1}^\infty \biggl(\sum_{k=1}^{r-1}\frac{(-1)^k\mathfrak{a}_{k/(2r)} a_{k/r,n}}{D^{k-1}(\pi n)^{k/r}}\biggr)\cos \pi n x \\
				& + \frac{(-1)^r}{\pi D^{r-1}}\biggl( 1 + 2\sum_{n=1}^{\infty}\bigl(\text{Si}(n\pi)-\tfrac{\pi}{2}\bigr)\frac{\cos n\pi x}{n\pi}\biggr),
			\end{split}
		\end{equation}
	\end{subequations}
	where  $\text{Si}(z)$ is the sine integral and where $a_{k/r,n}$ and $\mathfrak{a}_{k/(2r)}$ are defined as in \eqref{eq:greens-series-rapid-full}.
\end{proposition}

\section{Numerical Implementation}\label{app:numerical-implementation}

In this appendix we outline the numerical methods used for the numerous computations in this paper. Specifically we first describe the numerical implementation of the fractional Laplacian in both $\mathbb{R} $ with a prescribed far-field behaviour as well as in the interval $-1<x<1$ with periodic boundary conditions. We then outline the key steps in the numerical continuation used to construct solutions to the core problem \eqref{eq:core_problem}. This is followed by a discussion of the numerical computation of the spectrum of the nonlocal operator $\mathscr{M}$ defined in \eqref{eq:M-def} as well as its adjoint $\mathscr{M}^\star$ from \S \ref{sec:slow-dynamics}. We then outline  the implementation of IMEX methods for numerically simulating \eqref{eq:frac-gm-full-system}.

\subsection{Numerical Computation of The Fractional Laplacian}\label{subapp:fractional-laplacian}

To numerically solve the core problem \eqref{eq:core_problem} as well as to determine the spectrum of the nonlocal operator $\mathscr{M}$ defined in \eqref{eq:M-def} and its adjoint $\mathscr{M}^*$ considered in \S \ref{sec:slow-dynamics} we must first  calculate an appropriate discretization of the fractional Laplacian in $\mathbb{R} $. On the other hand to numerically simulate the full system \eqref{eq:frac-gm-full-system} we must discretize the fractional Laplacian on $-1<x<1$ with periodic boundary conditions.

\subsubsection{The Fractional Laplacian in $\mathbb{R} $}\label{subsubapp:fractional-laplacian-R1}

We use the finite-difference method of Huang and Oberman \cite{huang_2013} to calculate the discretized fractional Laplacian. For completeness we include here the most important details in the implementation. To simplify our presentation we consider the problem of discretizing $(-\Delta)^s \varphi$ for $0<s<1/2$ when $\varphi$ satisfies the far-field behaviour $\varphi\sim a_{\pm} + b_{\pm}|y|^{-\beta}$ for constants $a_{\pm}$, $b_{\pm}$, and $\beta>0$.  We first introduce the truncated domain $-L<y<L$ and consider the approximate boundary conditions $\varphi(y) \approx a_\pm + b_\pm |y|^{-\beta}$ for $\pm y\geq L$.  Next we introduce the discretization $y_i = ih$ for $i=-2N,...,2N$ where $h=L/N$. This leads to the computational domain $-2L<y<2L$ which has been expanded from the original truncated domain to account for nonlocal contributions. For each $i>N$ we impose a fixed value for $\varphi_{\pm i}\equiv \varphi(y_{\pm i})$ that depends on its value at $i=\pm N$ in one of two ways
\begin{equation}\label{eq:app-far-field-cases}
	\begin{cases}
		\text{if $a_\pm=0$ then $b_\pm = \varphi_{\pm N}L^{\beta}$ and  $\varphi_{\pm i} = \varphi_{\pm N} |y_{\pm i}/L|^{-\beta}$} ,& \text{(Case 1)},\\
		\text{if $b_\pm$ is given then $a_\pm = \varphi_{\pm N} - b_{\pm}L^{-\beta}$ and $\varphi_{\pm i} = \varphi_{\pm N} + b_{\pm} (|y_{\pm i}|^{-\beta} - L^{-\beta})$},& \text{(Case 2)}.
	\end{cases}
\end{equation}
These two cases account, respectively, for the activator and inhibitor in the core problem \eqref{eq:core_problem}. We next let $\nu_s(y)\equiv C_s |y|^{-(1+2s)}$ and for each $|i|\leq N$ we decompose \eqref{eq:fractional-laplacian} as
{\small \begin{equation}\label{eq:app-discrete-temp-0}
		(-\Delta)^s \varphi(y_i) =
		\underbrace{\int_{-2L}^{2L}\bigl(\varphi(y_i)-\varphi(y_i-y)\bigr)\nu_s(y)dy}_{\text{I}_i} + \underbrace{\varphi_i\int_{|y|>2L}\nu_s(y)dy}_\text{II} - \underbrace{\int_{|y|>2L}\varphi(y_i-y)\nu_s(y)dy}_{\text{III}_{i}}.
\end{equation}}

The first integral in \eqref{eq:app-discrete-temp-0} is approximated by performing a piecewise quadratic interpolation of $\varphi$ which gives us
\begin{equation}\label{eq:app-R1-I}
	\text{I}_i \approx \sum_{j=-2N}^{2N}(\varphi_{i} - \varphi_{i-j})w_j^Q,
\end{equation}
where $w_j^Q$ ($j=-2N,...,2N$) are the quadratic interpolation weights calculated using Definition 3.2 of \cite{huang_2013} with $\alpha=2s$ and explicitly given by
\begin{subequations}\label{eq:app-weight-frac-lap}
	\begin{equation}
		w_j^Q = \frac{C_s}{h^{2s}}\begin{cases}    \tfrac{1}{2-2s} - G''(1) - \tfrac{1}{2}(G'(3) + 3G'(1))  +G(3) - G(1), & j=\pm 1, \\ 2(G'(j+1) + G'(j-1) - G(j+1) + G(j-1)), & j=\pm 2, \pm 4, ..., \\ -\tfrac{1}{2}(G'(j+2) + 6G'(j) + G'(j-2)) + G(j+2)-G(j-2), & j=\pm 3, \pm 5,..., \end{cases}
	\end{equation}
	where
	\begin{equation}\label{eq:weight-G}
		G(t) \equiv \begin{cases} \frac{1}{2s(2-2s)(2s-1)}|t|^{2-2s}, & s\neq 1/2, \\ t-t\log|t|, & s = 1/2.\end{cases}
	\end{equation}
\end{subequations}
Note that some of the entries in \eqref{eq:app-R1-I} will have $|i-j|>N$ and for these we use \eqref{eq:app-far-field-cases}.

Next, by using the definition of $\nu_s(y)$ we readily calculate
\begin{equation}\label{eq:app-R1-II}
	\text{II} = \frac{C_s}{s(2L)^{2s}}.
\end{equation}

To calculate III$_i$ we first consider the portion of the integral for which $y>2L$ and hence $y_i-y<-L$ so that by using \eqref{eq:app-far-field-cases} we calculate
\begin{equation*}
	\int_{2L}^\infty \varphi(y_i-y)\nu_s(y)dy = \begin{cases} \varphi_{-N} L^{\beta}\int_{2L}^\infty |y_i-y|^{-\beta}\nu_s(y)dy, & \text{(Case 1)}, \\
		\tfrac{1}{2}\bigl(\varphi_{-N} - b_-L^{-\beta}\bigr)\text{II} + b_-\int_{2L}^\infty |y_i-y|^{-\beta}\nu_s(y)dy, & \text{(Case 2)}.
	\end{cases}
\end{equation*}
The rightmost integral in both cases can be written in terms of the Gauss Hypergeometric function $_2F_1(a,b,c,z)$ which we recall has the integral representation (see (9.111) in \cite{gradshteyn_2015})
\begin{equation*}
	_2F_1(a,b,c,z) = \frac{\Gamma(c)}{\Gamma(b)\Gamma(c-b)}\int_0^1\frac{t^{b-1}(1-t)^{c-b-1}}{(1-z t)^a} dt,\quad (b,c>0).
\end{equation*}
A simple change of variables then immediately yields that for any $\beta$ and $\alpha$ with $\beta+\alpha>0$
\begin{equation}\label{eq:app-tail-integral-general}
	\int_{2L}^\infty |y_i-y|^{-\beta}|y|^{-1-\alpha}dy = \frac{1}{(2L)^{\alpha+\beta}(\alpha+\beta)}{ }_2F_1(\beta,\alpha+\beta,\alpha+\beta+1,(2L)^{-1}y_i),
\end{equation}
with an analogous result for the integration over $-\infty<y<-2L$. Thus for Case 1 and Case 2 we respectively obtain
\begin{subequations}\label{eq:app-R1-III}
	{\small \begin{equation}\label{eq:app-R1-III-1}
			\text{III}_{i} \approx C_s\frac{\varphi_{-N}{}_2F_1(\beta,\beta+2s,\beta+2s+1,(2L)^{-1}y_i) + \varphi_{N}{}_2F_1(\beta,\beta+2s,\beta+2s+1,-(2L)^{-1}y_i)}{2^{\beta+2s}L^{2s}(\beta+2s)}
	\end{equation}}
	and
	{\small \begin{equation}\label{eq:app-R1-III-2}
			\begin{split}
				\text{III}_{i} \approx & \tfrac{1}{2}\bigl(\varphi_{-N}+\varphi_{N}-(b_-+b_+)L^{-\beta}\bigr)\text{II} \\
				& + C_s\frac{b_-{}_2F_1(\beta,\beta+2s,\beta+2s+1,(2L)^{-1}y_i) + b_+{}_2F_1(\beta,\beta+2s,\beta+2s+1,-(2L)^{-1}y_i)}{(2L)^{\beta+2s}(\beta+2s)}.
			\end{split}
	\end{equation}}
\end{subequations}

Substituting equations \eqref{eq:app-R1-I}, \eqref{eq:app-R1-II}, and \eqref{eq:app-R1-III} into \eqref{eq:app-discrete-temp-0} thus yields a discretization of the fractional Laplacian incorporating the relevant far-field behaviour for the core problem and spectrum calculations. We conclude by remarking that in Case 1 this discretization leads to multiplication by a dense matrix whereas for Case 2 the discretization involves both multiplication by a dense matrix as well as the addition of an inhomogeneous term that captures the far-field behaviour.

\subsubsection{The Periodic Fractional Laplacian in $-1<x<1$}\label{subsubapp:fractional-laplacian-periodic}

We now consider the discretization of $(-\Delta)^s\varphi(x)$ in $-1<x<1$ with periodic boundary conditions. In contrast to the discretization of the fractional Laplacian in $\mathbb{R} $ considered above, periodic boundary conditions lead to a significantly simpler implementation. We discretize the domain $-1<x<1$ by letting $x_i=-1+2ih$ for $i=0,...,N-1$ where $h=1/N$. By approximating $\varphi$ with a piecewise quadratic interpolator and letting $\varphi_i\equiv\varphi(x_i)$ we then get
\begin{equation}\label{eq:discrete-laplacian-periodic}
	(-\Delta)^s\varphi(x_i) \approx \sum_{j=0}^{N-1}(\varphi_i-\varphi_{i-j})W_{i-j}^Q,\quad W_{n}^Q = w_n^Q + \sum_{k=1}^\infty (w_{k+Nn}^Q + w_{k-Nn}^Q),
\end{equation}
where $w_i^Q$ are again the quadratic weight functions found in \cite{huang_2013}. This discretization is readily implemented and leads to a dense matrix. To compute $W_{n}^Q$ we truncate the sum by taking the first $10^3$ terms.

\subsection{Solving the Core Problem}\label{subapp:core-problem}

In this section we outline the key steps for calculating $\mu(S)$ and $\nu(S)$ for $S>0$. This computation has two main steps. In the first step we numerically compute the fractional homoclinic $w_{s_1}$ satisfying \eqref{eq:fractional-homoclinic} for values of $1/4<s_1<1$. In the second step we use the $S\ll 1$ asymptotics \eqref{eq:small-S-asymptotics} to initiate a numerical continuation in $S$ to solve \eqref{eq:core_problem} from which $\mu(S)$ and $\nu(S)$ can then be computed from \eqref{eq:nu_mu_def}.

The fractional homoclinic solution to \eqref{eq:fractional-homoclinic} was previously numerically computed for values of $1/4<s_1<1$ in Appendix B of \cite{gomez_2022} and therefore we provide only an outline here and refer the reader to that paper for more details. The first step is to discretize \eqref{eq:fractional-homoclinic} using the method in Appendix \ref{subsubapp:fractional-laplacian-R1}. This yields a nonlinear system for the discretized solution which we can solve using Newton's method. Next we note that the exact solution to \eqref{eq:fractional-homoclinic} when $s_1=1/2$ is $w_{1/2}(y) = 2 / (1+y^2)$. Letting $(s_1)_0=1/2<(s_1)_1<...<(s_1)_M<1$ for a large $M>0$ we can then numerically solve the discretized nonlinear system for $s_1=(s_1)_i$ when $i=1,\cdots,M$ by using the solution for $s_1=(s_1)_{i-1}$ as an initial guess where the exact solution is used for $i=1$. A similar continuation is likewise performed for values of $1/4<s_1<1/2$.

The numerical solution of the core problem \eqref{eq:core_problem} follows a similar procedure. Numerically discretizing both the fractional Laplacians appearing in \eqref{eq:core_problem} using the method in \eqref{subsubapp:fractional-laplacian-R1} yields a nonlinear system for the discretized solutions. Note that when performing the discretization the far-field  condition for Case 1 and Case 2 (see Equation \eqref{eq:app-far-field-cases}) is used for $U_c$ and $V_c$ respectively. It is then straightforward to perform a numerical continuation in $S>0$ by slowly incrementing it and using the small $S$ asymptotics \eqref{eq:small-S-asymptotics} as the initializing guess.

\subsection{Computing the Spectrum of $\mathscr{M}$}

In this section we consider the numerical computation of the spectrum of the nonlocal operator $\mathscr{M}$ defined in \eqref{eq:M-def}. This is done by discretizing $\mathscr{M}$ and then calculating the spectrum of the resulting matrix operator using standard eigenvalue libraries (we used the \texttt{eig} function from the SciPy \texttt{linalg} library). In this section it therefore suffices to describe the discretization of $\mathscr{M}$. The operator $\mathscr{M}$ consists of two nonlocal contributions: the fractional Laplacian $(-\Delta)^{s_1}$ and the convolution with the fractional Green's function. The numerical discretization of the former was considered in detail in Appendix \ref{subsubapp:fractional-laplacian-R1} and so it remains only to discuss the discretization of the latter.

To simplify (and generalize) our presentation we henceforth focus on numerically approximating
\begin{equation}
	J[\varphi](y) \equiv \int_{-\infty}^\infty \frac{\varphi(z)}{|y-z|^{1-2s}} dz,
\end{equation}
where $0<s<1/2$ and $\varphi$ is assumed to have the far-field behaviour
\begin{equation}
	\varphi(y)\sim \varphi_{\pm\infty}|y|^{-\beta}\quad\text{as}\quad y\rightarrow\pm\infty,
\end{equation}
where $\varphi_{\pm\infty}$ are unknown. Note that we use $s=s_2$ in the discretization of both $\mathscr{M}$ and $\mathscr{M}^\star$ whereas we use $\varphi(y) = U_c(y;S)\Phi_c^\lambda(y;S)$ so that $\beta = 2+2s_1$ in the former and $\varphi(y)=(U_c(y;S)/V_c(y;S))^2P(y)$ so that $\beta = 3+3s_1$ in the latter.  We fix our notation by introducing the discretization $y_i = ih$ and letting $\Phi_i = \varphi(y_i)$ for $i=-2N,...,2N$ where $h=L/N$. Then we decompose
\begin{equation}
	J[\varphi](y_i) = \underbrace{\int_{|z|<h}\frac{\varphi(y_i-z)}{|z|^{1-2s}}dz}_{J_1[\varphi](y_i)} +  \underbrace{\int_{h<|z|<2L}\frac{\varphi(y_i-z)}{|z|^{1-2s}}dz}_{J_2[\varphi](y_i)} + \underbrace{\int_{|z|>2L}\frac{\varphi(y_i-z)}{|z|^{1-2s}}dz}_{J_3[\varphi](y_i)}.
\end{equation}
The first integral is easily approximated using a Taylor series as
\begin{equation}
	\begin{split}
		J_1[\varphi](y_i) = \int_{-h}^h\frac{\varphi(y_i-z)}{|z|^{1-2s}}dz = \int_{-h}^h \frac{\varphi(y_i) -  \varphi'(y_i)z + O(z^2)}{|z|^{1-2s}}dz = \frac{h^{2s}}{s}\varphi_i + O(h^{2+2s}).
	\end{split}
\end{equation}
The approximation of the remaining two integrals proceeds as for the fractional Laplacian above. Specifically imposing that $\varphi_j = \varphi_{\pm N} |y_j/L|^{-\beta}$ for all $\pm j > N$ we first approximate $J_2$ with the finite sum
\begin{equation}
	J_2[\varphi](y_i) \approx \sum_{1\leq |j| \leq 2N} \varphi_{i-j}\tilde{w}_j^Q,
\end{equation}
where $\tilde{w}_j^Q$ ($j=-2N,...,2N$) are the quadratic interpolation weights given by
\begin{equation}\label{eq:app-weight-green}
	\tilde{w}_j^Q = h^{2s}\begin{cases}    - G''(1) - \tfrac{1}{2}(G'(3) + 3G'(1))  +G(3) - G(1), & j=\pm 1, \\ 2(G'(j+1) + G'(j-1) - G(j+1) + G(j-1)), & j=\pm 2, \pm 4, ..., \\ -\tfrac{1}{2}(G'(j+2) + 6G'(j) + G'(j-2)) + G(j+2)-G(j-2), & j=\pm 3, \pm 5,..., \end{cases}
\end{equation}
where $G(t)$ is given by \eqref{eq:weight-G}. Observe that the multiplicative factor $C_s$ as well as the $(2-2s)^{-1}$ term for the $j=\pm 1$ case present in \eqref{eq:app-weight-frac-lap} are omitted in \eqref{eq:app-weight-green}. The former is omitted because it is specific to the definition of the fractional Laplacian while the latter is omitted because it arises from a finite difference approximation for the singular part of the fractional Laplacian which in the present case is contained in the $J_1$ contribution.

The remaining integral contributions $J_3$ are then computed using \eqref{eq:app-tail-integral-general} which yields
{\small \begin{equation}
		J_3[\varphi](y_i) \approx L^{2s}\frac{\varphi_{-N}{}_2F_1(\beta,\beta-2s,\beta-2s+1,\tfrac{y_i}{2L}) + \varphi_{N}{}_2F_1(\beta,\beta-2s,\beta-2s+1,-\tfrac{y_i}{2L})}{2^{\beta-2s}(\beta-2s)}.
\end{equation}}

\subsection{Time Stepping of the Fractional Gierer-Meinhardt System}

We performed full numerical simulations of \eqref{eq:frac-gm-full-system} by first discretizing the system using the methods in \S\ref{subsubapp:fractional-laplacian-periodic} which yields a system of $2N$ ordinary differential equations of the form
\begin{equation}\label{eq:mol-ode}
	\frac{d\bm{\Phi}}{dt}  + \mathcal{A}\bm{\Phi}  + \bm{\mathcal{N}}(\bm{\Phi}) = \bm{0}.
\end{equation}
In this expression $\bm{\Phi}(t) = (u(x_0,t),\cdots,u(x_{N-1},t),v(x_0,t),\cdots,v(x_{N-1},t))^T$ approximates the solution at the discretization points $x_i=-1+2ih$ for $i=0,...,N-1$ as in \S\ref{subsubapp:fractional-laplacian-periodic}. The $2N\times 2N$ block-diagonal matrix $\mathcal{A}$ then has entries corresponding to the discretization \eqref{eq:discrete-laplacian-periodic} while the $2N$ dimensional vector $\bm{\mathcal{N}}(\bm{\Phi})$ accounts for the nonlinearities in \eqref{eq:frac-gm-full-system}. We integrate the ODE system using a second-order semi-implicit backwards difference scheme (2-SBDF) \cite{ruuth_1995} which leads to the linear system
\begin{equation}
	(3\mathcal{I}-2\Delta t \mathcal{A})\bm{\Phi}_{n+1} = 4\bm{\Phi}_n - \bm{\Phi}_{n-1} + 4\Delta t \bm{\mathcal{N}}(\bm{\Phi}_n) -  2\Delta t \bm{\mathcal{N}}(\bm{\Phi}_{n-1}),
\end{equation}
where $\Delta t>0$ is the time-step size and $\bm{\Phi}_n = \bm{\Phi}(n\Delta t)$. The initial condition is given by $\bm{\Phi}_0$ and since 2-SBDF is second-order we also need $\bm{\Phi}_1$. We obtain this second initial value by integrating \eqref{eq:mol-ode} using the first-order semi-implicit backward difference scheme (1-SBDF) with a smaller time step $\Delta t / \tilde{N}$ for $\tilde{N}>0$ steps. Specifically, we do this by solving
\begin{equation*}
	(\mathcal{I} - \tfrac{\Delta t}{\tilde{N}} \mathcal{A})\bm{\Phi}_{(n+1)/\tilde{N}} = \bm{\Phi}_{n/\tilde{N}} + \tfrac{\Delta t}{\tilde{N}}\bm{\mathcal{N}}(\bm{\Phi}_{n/\tilde{N}}),
\end{equation*}
for $n=1,...,\tilde{N}$.  For our numerical simulations we used $N=2000$, $\tilde{N}=5$, and $\Delta t = 0.01$. Moreover in the spatial discretization of the fractional Laplacian we truncated the infinite sum appearing in \eqref{eq:discrete-laplacian-periodic} after $250$ terms.

\end{document}